\documentclass[journal, romanappendices, final]{IEEEtran}
\usepackage{latexsym,amsfonts,amssymb,amsthm, setspace, verbatim,cite,mathrsfs,graphicx,bm,stfloats, hyperref, tikz,xcolor,soul}
\usepackage[cmex10]{amsmath}
\newcommand*{\qedfilled}{\hfill\ensuremath{\blacksquare}}
\newcommand{\un}{\underline}
\newcommand{\be}{\begin{equation}}
\newcommand{\ee}{\end{equation}}
\newcommand{\ben}{\begin{equation*}}
\newcommand{\een}{\end{equation*}}
\newcommand{\mc}{\mathcal}
\newtheorem{lem}{Lemma}

\newtheorem{defi}{Definition}
\newtheorem{thm}{Theorem}
\newtheorem{fact}{Fact}

\newtheorem{prop}{Proposition}

\newcounter{mytempeqncnt}
\newcommand{\e}{\epsilon}
\newcommand{\expec}{\mathbf{E}}
\newcommand{\abs}[1]{\lvert#1\rvert}

\begin{document}
\title{Achievable Rates for Channels with \\ Deletions and Insertions}
\author{Ramji Venkataramanan,~\IEEEmembership{Member,~IEEE,}
Sekhar Tatikonda,~\IEEEmembership{Senior Member,~IEEE,}\\
and Kannan Ramchandran,~\IEEEmembership{Fellow,~IEEE}

\thanks{This work was partially supported by NSF Grant CCF-1017744. Part of this paper was presented at the 2011 IEEE International Symposium on Information Theory.}%
\thanks{R.~Venkataramanan was with the Department of Electrical Engineering, Yale University. He is now with the Department of Engineering, University of Cambridge, Cambridge CB2 1PZ, UK (e-mail: ramji.v@eng.cam.ac.uk).}%
\thanks{S. Tatikonda is with the Department of Electrical Engineering, Yale University, New Haven CT 06511, USA (e-mail: sekhar.tatikonda@yale.edu).}
\thanks{K. Ramchandran is with Department of Electrical Engineering and Computer Science, University of California, Berkeley, CA
94704 USA (e-mail: kannanr@eecs.berkeley.edu).}
}
\maketitle

\begin{abstract}
  This paper considers a binary channel with deletions and insertions, where each input bit is transformed in one of the following ways: it is deleted with probability $d$, or an extra bit is added after it with probability $i$, or it is transmitted unmodified with probability $1-d-i$. A computable lower bound on the capacity of this channel is derived.
  The transformation of the input sequence by the channel may be viewed in terms of runs as follows: some runs of the input sequence get  shorter/longer,  some runs  get deleted, and some new runs are added. It is difficult for the decoder to synchronize the channel output sequence to the transmitted codeword  mainly due to deleted runs and new inserted runs.

  The main idea is a mutual information decomposition in terms of the rate achieved by  a sub-optimal decoder that  determines the positions of the deleted and inserted runs in addition to decoding the transmitted codeword. The mutual information between the channel input and output sequences is expressed as the sum of the rate achieved by this decoder and the rate loss due to its sub-optimality. Obtaining computable lower bounds on each of these quantities yields a lower bound on the capacity.
  The  bounds proposed in this paper provide the first characterization of achievable rates for channels with general insertions, and for channels with both deletions and insertions.   For the special case of the deletion channel, the proposed bound improves on the previous best lower bound for deletion probabilities up to $0.3$.
\end{abstract}


\section{Introduction}
\label{sec:intro}
Consider a binary input channel where for each bit (denoted $x$), the output is generated in one of the following ways:
\begin{itemize}
\item The bit is deleted with probability $d$,
\item An extra bit is inserted after $x$ with probability $i$. The extra bit is equal to $x$ (a \emph{duplication}) with probability $\alpha$, and equal to
$1-x$ (a \emph{complementary insertion}) with probability $1-\alpha$,
\item No deletions or insertions occur, and the output is $x$  with probability $1-d-i$.
\end{itemize}
The channel acts independently on each bit. We refer to this channel as the InDel channel with parameters $(d,i,\alpha)$.
If the input to the channel is a sequence of $n$ bits,  the length of the output sequence
will be close to $n(1+i-d)$   for large $n$ due to the law of large numbers.

Channels with synchronization errors can be used to model timing mismatch in  communication systems. Channels with deletions and insertions also occur in magnetic recording \cite{ISiegelWolf}. The problem of synchronization also appears in file backup and file sharing \cite{MaKTse11,ITASynch11},
where distributed nodes may have different versions of the same file which differ by a small number of edits. The edits  may include deletions, insertions, and substitutions. The minimum communication rate required to synchronize the remote sources is closely related to the capacity of an associated synchronization channel.  This connection is discussed at the end of this paper.

The model above with $i=0$ corresponds to the deletion channel, which has been studied in several recent papers, e.g., \cite{DiggaviG06,DrineaM07,DrineaK10,DMP07,Mitzenmacher09,FertonaniD10,KalaiMM10,KanoriaM10}. When $d=0$, we obtain the insertion channel with parameters $(i,\alpha)$. The insertion channel with $\alpha=1$ is the elementary sticky channel \cite{Mitz_sticky}, where all insertions are duplications.

In this paper, we obtain lower bounds on the capacity of the InDel channel. Our starting point is  the result of Dobrushin \cite{Dobrushin67} for general synchronization channels which states that the capacity is given by the maximum of the mutual information per bit between the input and output sequences. There are two challenges to computing the capacity through this characterization. The first is evaluating the mutual information, which is a difficult task because of the memory inherent in the joint distribution of the input and output sequences. The second challenge is to optimize the mutual information over all input distributions.

In this work, we choose the input distribution to be the class of first-order Markov processes and focus on the problem of evaluating the mutual
information.  It is known that first-order Markov input distributions yield good capacity lower bounds for the deletion channel \cite{DiggaviG06,DrineaM07} and the elementary sticky channel \cite{Mitz_sticky}, both special cases of the InDel channel. This suggests they are likely  to perform well on the InDel channel as well. The runs of a binary sequence are its alternating blocks of contiguous zeros and ones. First-order Markov sequences have runs that  are  independent, and the average run length can be controlled via the Markov parameter.   This fits well with the techniques used in this paper, which are based on the relationship between input and output runs of the channel.

For a synchronization channel, it is useful to think of the input and output sequences in terms of runs of symbols rather than individual symbols.
If there was a  one-to-one correspondence between the runs of  the input sequence $\un{X}$ and those of the output sequence $\un{Y}$, we could write the conditional distribution $P(\un{Y}|\un{X})$ as a product distribution of run-length transformations; computing the mutual information would then be straightforward. Unfortunately, such a correspondence is not possible since deletions can lead to some runs being lost, and insertions to new runs being inserted.

The main idea of the paper is to use auxiliary sequences which indicate the positions (in the output sequence) where \emph{runs} were deleted and inserted. Consider a decoder that decodes the auxiliary sequences in addition to the  transmitted codeword.  Such a decoder is sub-optimal compared to a maximum-likelihood decoder because of the extra information  it decodes, but its performance is tractable. The mutual information between the channel input and output sequences is decomposed as the sum of two terms: 1) the rate achieved by the sub-optimal decoder, and 2) the rate loss due to the sub-optimality of the decoder. We obtain a lower bound on the channel capacity via lower bounds on each of the terms above. For the special case of the deletion channel, the rate achieved by the sub-optimal decoder can be precisely calculated.

To gain insight, we first consider the special cases of the insertion channel and the deletion channel separately. The insertion channel with parameters $(i,\alpha)$ introduces approximately $n i$ insertions in a sufficiently long input sequence of length $n$. A fraction  nearly $\alpha$ of these insertions are duplications, and the rest are complementary insertions. Note that new runs can only be introduced  by complementary insertions. We consider a sub-optimal decoder that first decodes the positions of the complementary insertions. For the deletion channel,  we consider a decoder that first decodes an auxiliary sequence whose symbols indicate the number of runs deleted between each pair of adjacent bits in the output sequence. Augmenting the output sequence with the positions of deleted runs results in a one-to-one correspondence between input and output runs. For the InDel channel, the sub-optimal decoder decodes both auxiliary sequences described above. In each case, a capacity lower bound  is obtained by combining bounds on the rate achieved by the sub-optimal decoder and the rate loss due to sub-optimality.

The main contributions of the paper are the following:
\begin{enumerate}
\item Theorems \ref{thm:ins_lb1} and \ref{thm:ins_lb2} together provide the first characterization of achievable rates for the general insertion channel ($d=0$).
Previous results exist only for the special case of the sticky channel ($\alpha=1$, i.e., only duplications),

\item For the special case of the deletion channel ($i=0$), Theorem \ref{thm:del_thm} improves on the best known capacity lower bound in\cite{DrineaM07} for
$0 < d \leq 0.3$.

\item Theorem \ref{thm:delins_thm} provides the first characterization of achievable rates for the InDel channel.
\end{enumerate}

Our approach provides a general framework to compute the capacity of channels with synchronization errors, and suggests several directions
to obtain sharper capacity bounds. For example, results on the structure of optimal input distributions for these channels (in the spirit of \cite{KanoriaM10,KalaiMM10}) could be combined with our approach to improve the lower bounds.
One could also obtain upper bounds on the capacity by assuming that the auxiliary sequences are available `for free' at the decoder, as done in \cite{DiggaviG06} for the deletion channel. For clarity, we only consider the binary InDel channel. The results presented here can be extended to channels with any finite alphabet. This is briefly discussed in Section \ref{sec:conc}.

\subsection{Related Work}

\emph{Jigsaw Decoding}: The best previous lower bounds for the deletion channel are due to Drinea and Mitzenmacher \cite{DrineaM07}. They use a `jigsaw' decoder which decodes  the transmitted codeword by determining exactly which group of runs in $\un{X}$ give rise to each run in $\un{Y}$ (this is called the `type' of the $\un{Y}$-run in \cite{DrineaM07}). Analyzing the performance of such a decoder yields a lower bound on the deletion capacity. The sub-optimality of the jigsaw decoder is due to the fact that there may be many sequences of types consistent with a given  pair $(\un{X}, \un{Y})$. The rate loss due to this sub-optimality is precisely characterized in \cite{DrineaK10} in terms of a mutual information decomposition. For a given input distribution that is i.i.d across runs, \cite{DrineaK10} expresses the mutual information as the sum of two quantities -- the first is the rate achieved by the jigsaw decoder, the second is a conditional entropy term that is the rate loss incurred due to using a jigsaw decoder rather than an optimal (maximum-likelihood) decoder. This conditional entropy is a multi-letter expression that is hard to compute and is estimated via simulation in \cite{DrineaK10} for a few values of $d$.

Our approach to the deletion channel in Section \ref{sec:deletion} also involves a mutual information decomposition, but in terms of a different sub-optimal decoder.  The first term in the decomposition is the rate achieved by decoding the positions of the deleted runs in addition to the transmitted codeword; the second term is  rate penalty incurred by such a decoder. An interesting observation (discussed in Section \ref{subsec:jigsaw_comp}) is that the decoder we consider is actually inferior to the jigsaw decoder.  However, the \emph{penalty term} of our decoder is easier to bound analytically. As a consequence, our mutual information decomposition yields better lower bounds on the deletion capacity for a range of deletion probabilities.  Additionally, the idea of a decoder that synchronizes input and output runs naturally extends to channels with general insertions: we decompose the mutual information in terms of the rate achieved by imposing that the positions of complementary insertions be decoded, and the rate penalty incurred by such  decoder. The jigsaw decoder, in contrast, requires that each output run be associated with a set of complete input runs, which is not possible when there are complementary insertions.

\emph{Other Related Work}:
 Dobrushin's capacity characterization was used in \cite{KanoriaM10} to establish a a series expansion for  the deletion capacity at small values of $d$. The capacity is estimated by computing the leading terms of the expansion, and it is shown that the optimal input distribution can be obtained by smoothly perturbing the i.i.d Bernoulli$(\tfrac{1}{2})$ process. In \cite{DMP07}, a genie-aided decoder with access to the locations of  deleted runs was used to upper bound the deletion capacity  using an equivalent discrete memoryless channel (DMC).  In \cite{FertonaniD10},  bounds on the deletion capacity were obtained by considering a decoder equipped with side-information specifying the number of output bits corresponding to successive blocks of $L$ input bits, for any positive integer $L$. This new channel is equivalent to a DMC with an input alphabet of size $2^L$, whose capacity can be numerically computed using the Blahut-Arimoto algorithm  (for $L$ as large as computationally feasible). The upper bound in \cite{FertonaniD10} is the best known for a wide range of $d$, but the lower bound is weaker than that  of\cite{DrineaM07} and the one proposed here.

 In \cite{IyengarSW11}, bounds are obtained on the capacity of a channel with deletions and duplications  by converting it to an equivalent channel with states. Various results on the capacity of channels with  synchronization and substitution errors are obtained in \cite{MercierTL12}.
Finally, we note that a different channel model with bit flips and synchronization errors was studied in \cite{Gallager61, FDE11}. In this model, an insertion is defined as an input bit being replaced by two random bits. We have only mentioned the papers that are closely related to the results of this work. The reader is referred to \cite{Mitzenmacher09} for an exhaustive list of references on synchronization channels. \cite{MercierTL12} also contains a review of existing results on these channels.

 After laying down the formal definitions and technical machinery in Section \ref{sec:prelim},  in Section \ref{sec:coding_scheme} we describe coding schemes which give intuition about our bounding techniques.  In Section \ref{sec:insertion}, we consider the insertion channel ($d=0$) and derive two lower bounds on its capacity.  For this channel, previous bounds exist only for the special case of the elementary sticky channel ($\alpha=1$) \cite{Mitz_sticky}.  In Section \ref{sec:deletion}, we derive a lower bound on the capacity of  the deletion channel ($i=0$)  and compare it with the best previous lower bound \cite{DrineaM07}. We also compare the sub-optimality of decoding the positions of deleted runs with the sub-optimality of the jigsaw decoder of \cite{DrineaM07}.  In Section \ref{sec:delins_channel}, we combine the ideas of Sections  \ref{sec:insertion} and  \ref{sec:deletion} to obtain a lower bound for the InDel channel. Section \ref{sec:conc} concludes the paper with a discussion of open questions.

 \section{Preliminaries}
\label{sec:prelim}
\emph{Notation}: $\mathbb{N}_0$ denotes the set of non-negative integers, and $\mathbb{N}$ the set of natural numbers. For $\alpha \in [0,1]$, $\bar{\alpha} \triangleq 1-\alpha$.  Logarithms are with base $2$, and entropy is measured in bits.   $h(.)$  is the binary entropy function and $\mathbf{1}_{\mc{A}}$ is the indicator function of the set $\mc{A}$. We use uppercase letters to denote random variables, bold-face letters for random processes, and superscript notation to denote random vectors.

The communication over the channel is characterized by three random processes defined over the same probability space: the input process $\mathbf{X}=\{X_n\}_{n \geq 1}$, the output process $\mathbf{Y}=\{Y_n\}_{n \geq 1}$, and $\mathbf{M}=\{M_n\}_{n \geq 1}$, where $M_n$ is the number of output symbols corresponding to the first $n$ input symbols. If the underlying probability space is $(\Omega, \mc{F}, P)$, each  realization $\omega \in \Omega$ determines the sample paths $\mathbf{X}(\omega)=\{X_n(\omega)\}_{n \geq 1}$, $\mathbf{Y}(\omega)=\{Y_n(\omega)\}_{n \geq 1}$, and $\mathbf{M}(\omega)=\{M_n(\omega)\}_{n \geq 1}$.

The length $n$ channel input sequence is  $X^n = (X_1, \ldots, X_n)$ and the  output sequence is $Y^{M_n}$.  Note that $M_n$ is a random variable determined by the channel realization. For brevity, we sometimes use underlined notation for random vectors when we do not need to be explicit about their length. Thus $\un{X}=(X_1,X_2,\ldots,X_n)$, and $\un{Y} =(Y_1,\ldots,Y_{M_n})$.

\begin{defi}
An $(n,2^{nR})$  code with block
length $n$ and rate $R$ consists of
\begin{enumerate}
\item An encoder mapping $e:  \{1, \ldots, 2^{nR} \} \to \{0,1\}^n$, and
\item A decoder mapping $g:  \Sigma \to \{1, \ldots, 2^{nR} \}$
where $\Sigma$ is $\cup_{k=0}^n \{0,1\}^k$ for the deletion channel,
$\cup_{k=n}^{2n} \{0,1\}^k$ for the insertion channel, and $\cup_{k=0}^{2n} \{0,1\}^k$
for the InDel channel.
\end{enumerate}
\end{defi}
Assuming the message $W$ is drawn uniformly  from the set $\{1, \ldots, 2^{nR}\}$, the probability of error of a $(n,2^{nR})$ code is
\[
\begin{split}
P_{e,n}= \frac{1}{2^{nR}}\sum_{l=1}^{2^{nR}}\text{Pr}(g({Y}^{M_n}) \neq l | W=l)
\end{split}
\]
A rate $R$ is achievable if there exists a sequence of $(n,2^{nR})$ codes such that $P_{e,n} \to 0$ as $n \to \infty$. The supremum of all achievable rates is the capacity $C$. The following characterization of capacity follows from a result proved for a  general class of synchronization channels by Dobrushin \cite{Dobrushin67}.

\begin{fact}
Let $C_n =  \sup_{P_{X^n}} \frac{1}{n} I(X^n; Y^{M_n}).$
Then $C \triangleq \lim_{n \to \infty} C_n$ exists, and is equal to the capacity of the InDel channel.
\label{fact:dob}
\end{fact}
\begin{IEEEproof}
Dobrushin proved the following general result in \cite{Dobrushin67}. Consider a  channel with $\mc{X}$ and $\mc{Y}$ denoting the alphabets of possible symbols at the input and output, respectively. For each input symbol in $\mc{X}$, the output belongs to $\mc{\bar{Y}}$, the set of all finite sequences of elements of $\mc{Y}$, including the empty sequence.  The channel is memoryless and is specified by the stochastic matrix $\{P(\bar{y}|x), \ \bar{y} \in \bar{\mc{Y}}, {x} \in  \mc{X} \}$. Also assume that for each input symbol $x$, the length of the (possibly empty) output sequence has non-zero finite expected value. Then $\lim_{n \to \infty} C_n$ exists, and is equal the capacity of the channel.

The InDel channel is a special case of the above model with $\mc{X} = \mc{Y} = \{0,1 \}$, and the length of the output corresponding to any input symbol has  a maximum value of two and expected value equal to $(1-d+i)$, which is non-zero for all $d<1$. Hence the claim is a direct consequence of Dobrushin's result.
\end{IEEEproof}

In this paper, we fix the input process  to be the class of binary symmetric first-order Markov processes and focus on evaluating the mutual information. This will give us a lower bound on the capacity.The input process $\mathbf{X}=\{X_n\}_{n \geq 1}$ is characterized by the following distribution for all $n$:
\[ P(X_1, \ldots, X_n) =  P(X_1) \prod_{j=2}^n P(X_j|X_{j-1}), \]
where for $x \in \{0,1\}$,  $P(X_1=x)=\tfrac{1}{2}$ and for  $j  > 1$
\be \label{eq:inp_def}
\begin{split}
 P(X_j=x|X_{j-1}=x)=\gamma, \  \  P(X_j=\bar{x}|X_{j-1}=x)=\bar{\gamma}.
\end{split}
\ee
A binary  sequence may be represented by a sequence of positive integers representing the lengths of its runs, and the value of the first bit (to indicate whether the first run has zeros or ones). For example, the sequence $0001100000$ can be represented as $(3,2,5)$ if we know that the first bit is $0$. The value of the first bit of $\mathbf{X}$ can  be communicated to the decoder with vanishing rate, and we will assume this has been done at the outset.  Hence, denoting the length of the $j$th run of $\mathbf{X}$ by $L^{X}_j$ we have the following equivalence:  $\mathbf{X} \leftrightarrow (L^{X}_1, L^{X}_2, \ldots)$. For a first-order Markov binary source of \eqref{eq:inp_def}, the run-lengths are independent and geometrically distributed, i.e.,
\be \label{eq:run_dist}
\text{Pr}(L^{X}_j =r)= \gamma^{r-1} (1-\gamma), \qquad r=1, 2, \ldots
\ee
The average length of a run in $\mathbf{X}$ is $\tfrac{1}{1-\gamma}$, so the number of runs in a sequence  of length $n$ is close to $n(1-\gamma)$ for large $n$. Our bounding techniques aim to establish  a one-to-one correspondence between input runs and output runs. The independence of run-lengths of $\mathbf{X}$ enables us to obtain analytical bounds on the capacity. We denote by $I_P({X}^n; Y^{M_n}), H_P(X^n), H_P(X^n|Y^{M_n})$  the mutual information and entropies computed with the channel input sequence ${X^n}$ distributed as in \eqref{eq:inp_def}. For all $n$, we have
\be
\begin{split}
C_n= \sup_{P_{X^n}} \frac{1}{n} I(X^n;Y^{M_n}) & > \frac{1}{n} I_P(X^n;Y^{M_n}).
\end{split}
\ee
Therefore
\be \label{eq:mut_decomp}
\begin{split}
C & > \liminf_{n \to \infty} \frac{1}{n} I_P(X^n;Y^{M_n}) \\
&= h(\gamma) - \limsup_{n \to \infty} \frac{1}{n}H_P(X^n|Y^{M_n})
\end{split}
\ee
where $h(\gamma)$ is the entropy rate of the Markov process $\mathbf{X}$ \cite{CoverThomas}. We will derive upper bounds on $ \limsup_{n \to \infty} \frac{1}{n}H_P(X^n|Y^{M_n})$ and use it in \eqref{eq:mut_decomp} to obtain a lower bound on the capacity.
\subsection{Technical Lemmas}
To formally prove our results, we will use a framework similar to \cite{DrineaK10}. The notion of uniform integrability will play an important role.
We list the relevant definitions and technical lemmas below. The reader is referred to \cite[Appendix I]{DrineaK10} for a good overview of the concept of uniform integrability.
\begin{defi}
A family of random variables $\{Z_n\}_{n\geq 1}$ is uniformly integrable  if
\[ \lim_{a \to \infty} \sup_{n}  \expec[|Z_n| \mathbf{1}_{\{|Z_n| \geq a\}}] =0.\]
\end{defi}
\begin{lem} \cite[Lemma 7.10.6]{GrimStir}
\label{lem:unif_equiv}
A family of random variables $\{Z_n\}_{n\geq 1}$ is uniformly integrable  if and only if both the following conditions hold:
\begin{enumerate}
\item $\sup_n \expec[|Z_n|] < \infty$, and

\item  For any $\e >0$, there exists some $\delta >0$ such that for all $n$ and any event $\mc{A}$ with $\text{Pr}(\mc{A}) <\delta$,
we have $\expec[|Z_n| \ \mathbf{1}_{\mc{A}}] < \e$.
\end{enumerate}
\end{lem}

Let Supp$(W|Z)$ denote the random variable whose value is the size of the support of the conditional distribution of $W$ given $Z$.
\begin{lem} \cite[Lemma $4$]{DrineaK10} \label{lem:support}
Let $\{W_n, Z_n\}_{n \geq 1}$ be a sequence of pairs of discrete random variables with Supp$(W_n|Z_n) \leq c^n$ for some constant $c \geq 1$.
Then $\sup_n \expec \left[ \left(\frac{1}{n} \log \text{Pr}(W_n|Z_n)\right)^2 \right] < \infty$. In particular, the sequence
$\left\{ -\frac{1}{n} \log \text{Pr}(W_n|Z_n)\right\}_{n\geq 1}$ is uniformly integrable.
\end{lem}
\begin{lem}\cite[Thm. 7.10.3]{GrimStir}
\label{lem:exchange_lim}
Suppose that $\{Z_n: n\geq 1\}$ is a sequence of random variables that converges to $Z$ in probability. Then the following are equivalent.
\begin{enumerate}
\item $\{Z_n: n\geq 1\}$ is uniformly integrable.
\item $\expec[|Z_n|] < \infty$ for all $n$, and $\expec[|Z_n|] \stackrel{n \to \infty}{\longrightarrow} \expec[|Z|]$.
\end{enumerate}
\end{lem}
\begin{lem} \label{lem:entropyrate}
Let $\mathbf{Z}=\{Z_n\}_{n \geq 1}$ be a process for which the asymptotic equipartition property (AEP) holds, i.e.,
\[ \lim_{n \to \infty} -\frac{1}{n} \log \text{Pr}(Z_1, \ldots, Z_{n}) = H(Z) \quad a.s.\]
where $H(Z)$ is the (finite) entropy rate of the process $\mathbf{Z}$.  Let $\{M_n\}_{n\geq 1}$ be a sequence of positive integer valued random variables defined on the same probability space as the $Z_n$'s, and suppose that $\lim_{n \to \infty} \frac{M_n}{n}=x$ almost surely for some constant $x$. Then
\[ \lim_{n \to \infty} -\frac{1}{n} \log \text{Pr}(Z_1, \ldots, Z_{M_n}) = H(Z)x \quad a.s. \]
\end{lem}
\begin{IEEEproof}
Fix  $\epsilon> 0$ and define $a(n, \epsilon) \triangleq \lceil n(x-\epsilon) \rceil$ and $b(n, \epsilon) \triangleq \lceil n(x+\epsilon)\rceil$.
Since $\lim_{n \to \infty} \frac{M_n}{n} =x$ a.s., there exists an $L(\epsilon)$ such that for all $n > L(\epsilon)$,
\[ a(n, \epsilon)  \leq M_n \leq   b(n, \epsilon) \quad a.s. \]
It follows that for all $n>L(\epsilon)$,
\be
\begin{split}
& -\frac{1}{n} \log \text{Pr}(Z_1, \ldots, Z_{M_n})  \geq -\frac{1}{n} \log \text{Pr}(Z_1, \ldots, Z_{a(n,\epsilon)}) \\
& = - \frac{a(n,\epsilon)}{n} \cdot \frac{ \log \text{Pr}(Z_1, \ldots, Z_{a(n,\epsilon)})}{a(n,\epsilon)} \quad a.s.
\end{split}
\ee
Hence,
\be \label{eq:inf_hz}
\begin{split}
& \liminf_{n \to \infty} -\frac{1}{n} \log \text{Pr}(Z_1, \ldots, Z_{M_n})  \\
&\geq \liminf_{n \to \infty} -\frac{a(n,\epsilon)}{n} \cdot \frac{ \log \text{Pr}(Z_1, \ldots, Z_{a(n,\epsilon)})}{a(n,\epsilon)}.  \\
&{=} \lim_{n \to \infty} \frac{a(n,\epsilon)}{n}  \cdot \lim_{n \to \infty}\frac{-\log \text{Pr}(Z_1, \ldots, Z_{a(n,\epsilon)})}{a(n,\epsilon)}\\
&= (x-\epsilon)H(Z) \quad  a.s.
\end{split}
\ee
Similarly one can show that
\be \label{eq:sup_hz}
\limsup_{n \to \infty} -\frac{1}{n} \log \text{Pr}(Z_1, \ldots, Z_{M_n}) \leq (x+\epsilon)H(Z) \quad a.s.
\ee
Since $\epsilon >0$ is arbitrary, combining \eqref{eq:inf_hz} and \eqref{eq:sup_hz}, we get the result of the lemma.
\end{IEEEproof}

\section{Coding Schemes} \label{sec:coding_scheme}
In this section, we describe coding schemes to give intuition about the auxiliary sequences used to obtain the bounds. The discussion here is informal. The capacity bounds are rigorously proved in the following sections where the auxiliary sequences are used to directly decompose $\frac{1}{n} I(X^n; Y^{M_n})$ and the limiting behavior is bounded using information-theoretic inequalities and elementary tools from analysis.

\subsection{Insertion Channel} \label{subsec:insertion_scheme}
Consider the insertion channel with parameters $(i,\alpha)$. For $0<\alpha<1$, the inserted bits may create new runs, so we cannot associate each run of $\un{Y}$ with a run in $\un{X}$. For example, let \be \label{eq:ins_example} \un{X}={000111000} \; \  \text{ and } \; \ \un{Y}=00{\Large\emph{1}}0111 {\Large\emph{0}}000 {\Large\emph{0}},\ee where the inserted bits are indicated in large italics. There is one duplication (in the third run), and two complementary insertions (in the first and second runs). While a  duplication never introduces a new run, a complementary
insertion introduces a new run, except when it occurs at the end of a run of $\un{X}$ (e.g., the $0$ inserted at the end of the second run in \eqref{eq:ins_example}).
For any input-pair $(X^n, Y^{M_n})$, define an auxiliary sequence $T^{M_n}=(T_1, \ldots, T_{M_n})$ where $T_j=1$ if $Y_j$ is a
\emph{complementary} insertion, and $T_j=0$ otherwise. The sequence $T^{M_n}$ indicates the positions of the complementary insertions in $Y^{M_n}$.
In the example of \eqref{eq:ins_example}, $T^{M_n}=(0,0,1,0,0,0,0,1,0,0,0,0)$.

 Consider the following coding scheme. Construct a codebook of $2^{nR}$ codewords of length $n$, each chosen independently
according to the first-order Markov distribution \eqref{eq:inp_def}. Let $X^n$ denote the transmitted codeword, and $Y^{M_n}$ the channel output. From $Y^{M_n}$, the decoder decodes (using joint typicality) the positions of the complementary insertions, in addition to the input sequence.
The joint distribution of these sequences is determined by the input distribution \eqref{eq:inp_def} and the channel parameters $(i,\alpha)$.

Such a decoder is sub-optimal since the complementary insertion pattern $T^{M_n}$ is not unique given an input-output pair $(X^n, Y^{M_n})$. For example, the pair $\un{X}=01, \ \un{Y}=011$ can  either correspond to a complementary insertion in the second bit or a duplication in the third bit. The maximum rate achievable by this decoder  is obtained by analyzing the probability of error. Assuming all sequences satisfy the asymptotic equipartition property \cite{CoverThomas}, we have for sufficiently large $n$
\be \text{Pr(error) } \leq 2^{n(R + H(T^{M_n}|X^n))} \cdot 2^{-nI(X^n T^{M_n};Y^{M_n})}.  \ee
The second term above is the probability that  $(X^n,T^{M_n}, Y^{M_n})$ are jointly typical when $Y^{M_n}$ is picked independently from $(X^n,T^{M_n})$.
The first term is obtained by taking a union bound over all the codewords and all the typical complementary insertion patterns for each codeword.
Hence the probability of error goes to zero if
\be
\begin{split}
R  & < \frac{1}{n} \left( I(X^n \,T^{M_n};Y^{M_n}) - H(T^{M_n}|X^n) \right) \\
&= \frac{1}{n} \left( H(X^n) - H(X^n, T^{M_n}|Y^{M_n})  \right).
\end{split}
\ee
We can decompose the mutual information $I(X^n; Y^{M_n})$ as
\be
\begin{split}
\frac{1}{n}I(X^n; Y^{M_n})  =  &\underbrace{ \frac{1}{n}\left(H(X^n) - H(X^n,T^{M_n}|Y^{M_n}) \right)}_{\text{rate of sub-optimal decoder}}   \\
& + \underbrace{\frac{1}{n}H(T^{M_n}| X^n, Y^{M_n})}_{\text{penalty term}}.
\end{split}
\label{eq:ins_tpenalty}
\ee
The first part above is the rate achieved by the decoder described above and the second  `penalty' term represents the rate-loss due to its sub-optimality. We obtain a lower bound on the insertion capacity in Section \ref{subsec:lb2} by obtaining good single-letter lower bounds on the limiting behavior of both terms in
\eqref{eq:ins_tpenalty}.

\subsection{Deletion Channel} \label{subsec:deletion_scheme}
Consider the following pair of input and output sequences for the deletion channel:
$\un{X} = 000 111 000 , \ \un{Y}=  00 1 0$.
For this pair, we can associate each {run} of $\un{Y}$ uniquely with a run in $\un{X}$.
Therefore, we can write
\ben
\begin{split}
& P(\un{Y}=0010|\un{X}=000111000)  =\\
& P(L^Y_1=2|L^X_1=3)  P(L^Y_2=1|L^X_2=3)  P(L^Y_3=1|L^X_3=3)
\end{split}
\een
where $L^X_j, L^Y_j$ denote the lengths of the $j$th runs of $X$ and $Y$, respectively.
We observe that if {no} runs in $\un{X}$ are completely deleted, then the conditional distribution
of $\un{Y}$ given $\un{X}$ may be written as a product distribution of run-length transformations:
\be
\begin{split}
P(\un{Y}|\un{X})= P(L^Y_1|L^X_1)   P(L^Y_2|L^X_2)  P(L^Y_3|L^X_3) \ldots
\end{split}
\ee
where for all runs $j$, $P(L^Y_j=s|L^X_j=r) =  {r \choose s}d^{r-s} (1-d)^s$ for $1\leq s\leq r$. In general, there \emph{are} runs of $\un{X}$ that are completely deleted. For  example, if  $\un{X}=000111000$ and $\un{Y}= 000$, we cannot associate the single run in $\un{Y}$  uniquely with a run in $\un{X}$.

For any input-output pair $(X^n, Y^{M_n})$, define an auxiliary sequence $S^{M_n+1}=(S_1, S_2,\ldots, S_{M_n+1})$, where $S_j \in \mathbb{N}_0$ is the number of {runs} \emph{completely} deleted in ${X^n}$ between the bits corresponding to $Y_{j-1}$ and $Y_{j}$. ($S_1$ is the number of runs deleted before the first output symbol $Y_1$, and $S_{M_n+1}$ is the number of runs deleted after the last output symbol $Y_{M_n}$.) For example, if $ \un{X} = 00\underbrace{\emph{011100}}0 $ and the bits shown in italics were deleted to give $\un{Y}= 000$, then $\un{S}=(0,0,1,0)$. On the other hand, if the last six bits were all deleted, i.e.,
$\un{X} = 000\underbrace{\emph{111000}}$, then $\un{S}=(0, 0,0,2)$. Thus $\un{S}$ is not uniquely determined given $(\un{X}, \un{Y})$. The auxiliary sequence $\un{S}$ enables us to augment $\un{Y}$ with the positions of missing runs. As will be explained in Section \ref{sec:deletion}, the runs of this augmented output sequence are in one-to-one correspondence with the runs
of the input sequence.

Consider the following coding scheme. Construct a codebook of $2^{nR}$ codewords of length $n$, each chosen independently according to \eqref{eq:inp_def}. The decoder receives $Y^{M_n}$, and decodes (using joint typicality) both the auxiliary sequence and the input sequence. Such a decoder is sub-optimal since the auxiliary sequence $S^{M_n+1}$ is not unique given a codeword $X^n$ and the output $Y^{M_n}$.  Assuming all sequences satisfy the asymptotic equipartition property, we have for sufficiently large $n$
\be \text{Pr(error) } \leq 2^{n(R + H(S^{M_n+1}|X^n))} \cdot 2^{-nI(X^n S^{M_n+1};Y^{M_n})}.  \ee
The second term above is the probability that  $(X^n,S^{M_n+1}, Y^{M_n})$ are jointly typical when $Y^{M_n}$ is picked independently from $(X^n,S^{M_n+1})$.
The first term is obtained by taking a union bound over all the codewords and all the typical auxiliary sequences for each codeword.
Hence the probability of error goes to zero if
\be
\begin{split}
R  & < \frac{1}{n} \left( I(X^n S^{M_n+1};Y^{M_n}) - H(S^{M_n+1}|X^n) \right) \\
& = \frac{1}{n} \left( H(X^n) - H(X^n,S^{M_n+1}|Y^{M_n})  \right)
\end{split}
\ee
We can decompose the mutual information $I(X^n; Y^{M_n})$ as
\be
\begin{split}
\frac{1}{n} I(X^n; Y^{M_n}) = & \underbrace{ \frac{1}{n} \left(H(X^n) - H(X^n,S^{M_n+1}|Y^{M_n}) \right)}_{\text{rate of sub-optimal decoder}}  \\
&  +  \underbrace{\frac{1}{n} H(S^{M_n +1}| X^n, Y^{M_n})}_{\text{penallty term}}.
\end{split}
\label{eq:del_spenalty}
\ee
In Section \ref{sec:deletion}, we obtain an exact expression for the limit of the first term as $n \to \infty$ and a lower bound for the penalty term. These together yield a lower bound on the deletion capacity.

\subsection{InDel Channel}
For the InDel channel,  we use both auxiliary sequences $T^{M_n}$ and $S^{M_n+1}$. The sub-optimal decoder decodes both these sequences in addition to the
codeword $X^n$. The mutual information  decomposition in this case is
\be \label{eq:delins_coding_rate}
\begin{split}
\frac{1}{n} I(X^n; Y^{M_n}) = & \underbrace{ \frac{1}{n} \left( H(X^n) - H(X^n,S^{M_n+1}, T^{M_n}|Y^{M_n}) \right)}_{\text{rate of sub-optimal decoder}}  \\
& +  \underbrace{ \frac{1}{n}H(S^{M_n +1},  T^{M_n}| X^n, Y^{M_n})}_{\text{penalty term}}.
\end{split}
\ee
In Section \ref{sec:delins_channel}, we establish a lower bound on the capacity of the InDel channel by obtaining lower bounds for  both parts of\eqref{eq:delins_coding_rate} . As seen from \eqref{eq:ins_tpenalty}, \eqref{eq:del_spenalty} and \eqref{eq:delins_coding_rate}, the rate penalty for using the sub-optimal decoder is the conditional entropy of the auxiliary sequences given both the input and output sequences; this is essentially the extra information decoded compared to a maximum-likelihood decoder. In the following sections, we bound this conditional entropy by identifying insertion/deletion patterns that lead to different auxiliary sequences for the same $(\un{X}, \un{Y})$  pair.

 \section{Insertion Channel}\label{sec:insertion}
In this channel, an extra bit may be inserted after each bit of $\un{X}$ with probability $i \in (0,1)$. When a bit is inserted after $X_j$, the inserted bit is equal to ${X}_j$ (a duplication)  with probability $\alpha$, and equal to $\bar{X}_j$ (a complementary insertion) with probability $1-\alpha$. When $\alpha=1$, we have only duplications -- this is the elementary sticky channel studied in \cite{Mitz_sticky}. In this case,  we can associate each run of $\un{Y}$ with a unique run in $\un{X}$, which leads to a computable single-letter characterization of the best achievable rates with a first-order Markov distribution. We  derive two lower bounds on the capacity of the insertion channel, each using a different auxiliary sequence.

\subsection{Lower Bound $1$} \label{subsec:lb1}
For any input-pair $(X^n, Y^{M_n})$, define an auxiliary sequence $I^{M_n}=(I_1, \ldots, I_{M_n})$ where $I_j=1$ if $Y_j$ is an inserted bit, and $I_j=0$ otherwise. The sequence $I^{M_n}$ indicates the positions of all the inserted bits in $Y^{M_n}$, and is not unique  for a given $(X^n, Y^{M_n})$. Using $I^{M_n}$, we can decompose $H_P(X^n|Y^{M_n})$ as
\begin{equation*}
\begin{split}
H_P(X^n|Y^{M_n}) & = H_P(X^n, I^{M_n}|Y^{M_n}) - H_P(I^{M_n}| X^n, Y^{M_n}) \\
&  = H_P(I^{M_n}|Y^{M_n}) - H_P(I^{M_n}| X^n, Y^{M_n})
\end{split}
\end{equation*}
since $H(X^n|Y^{M_n}, I^{M_n})=0$. Therefore
\be  \begin{split}
&\liminf_{n \to \infty} \frac{1}{n} I_P(X^n ; Y^{M_n}) = h(\gamma) - \limsup_{n \to \infty} \frac{1}{n}H_P(X^n|Y^{M_n})  \\
& \geq \underbrace{h(\gamma) - \limsup_{n \to \infty} \frac{1}{n} H_P(I^{M_n}|Y^{M_n})}_{\text{rate of sub-optimal decoder}} \\
 & \ + \  \underbrace{\liminf_{n \to \infty} \frac{1}{n}H_P(I^{M_n}| X^n, Y^{M_n})}_{\text{penalty term}}.
\end{split}
\label{eq:HXY_HIY}
\ee
The last term above represents the rate loss of a sub-optimal decoder that decodes the transmitted codeword by first determining the positions of all the insertions.
We obtain a lower bound on the insertion capacity by deriving an upper bound on the  limsup  and a lower bound on the liminf in \eqref{eq:HXY_HIY}.
\begin{prop} \label{prop:ins_sy}
The process $\{\mathbf{I}, \mathbf{Y}\} \triangleq \{(I_1,Y_1), (I_2,Y_2), \ldots\}$ is a second-order Markov process characterized by the following joint distribution for all  $m \in \mathbb{N}$:
\ben
\begin{split}
& P(I^{m}, Y^m) = \\
& P(I_1, Y_1)  P(I_2,Y_2|I_1, Y_1) \prod_{j=3}^{m}  P(I_j, Y_j|I_{j-1}, Y_{j-1}, Y_{j-2})
\end{split} \een
where  for $x,y \in \{0,1\}$ and $j \geq 3$:
\be \label{eq:IY_dist}
\begin{split}
& P(I_j=1, Y_j=y \mid (I_{j-1}, Y_{j-1} , Y_{j-2})  = (0,y,x))= i\alpha, \\
& P(I_j=1, Y_j=\bar{y} \mid  (I_{j-1}, Y_{j-1}, Y_{j-2})=(0,y,x))= i\bar{\alpha} \\
& P(I_j=0, Y_j=y \mid (I_{j-1}, Y_{j-1}, Y_{j-2})=  (0,y,x))= \bar{i} \gamma,  \\
& P(I_j=0, Y_j=\bar{y} \mid   (I_{j-1}, Y_{j-1}, Y_{j-2})=  (0,y,x) )= \bar{i} \bar{\gamma}\\
& P(I_j=0, Y_j=x \mid   (I_{j-1}, Y_{j-1}, Y_{j-2})=  (1,y,x) )=  \gamma, \\
& P(I_j=0, Y_j=\bar{x} \mid  (I_{j-1}, Y_{j-1}, Y_{j-2})=  (1,y,x) )= \bar{\gamma}.
\end{split}
\ee
\end{prop}
\begin{IEEEproof}
We need to show that for all $j\geq 3$, the following Markov relation holds: $(I_j, Y_j) - (I_{j-1}, Y_{j-1}, Y_{j-2}) - (I^{j-2}, Y^{j-3})$.
First consider $P(I_j, Y_j \mid I_{j-1}=0, Y_{j-1}=y, I^{j-2}, Y^{j-2})$. Since $I_{j-1}=0$, $Y_{j-1}$ is the most recent input bit (say $X_a$) before $Y_j$.

$P(I_j=0, Y_j=y|I_{j-1}=0, Y_{j-1}=y, I^{j-2}, Y^{j-2})$ is the probability that the following independent events both occur: $1$) the input bit $X_{a+1}$ equals $X_a$ and $2$) there was no insertion after input bit $X_a$. Since the insertion process is i.i.d and independent of  the first-order Markov  process $\mathbf{X}$,  we have
\[P(I_j=0, Y_j=y | I_{j-1}=0, Y_{j-1}=y, I^{j-2}, Y^{j-2}) = \bar{i} \gamma. \]
Similarly, we obtain
\ben
\begin{split}
P(I_j=0, Y_j=\bar{y} \mid I_{j-1}=0, Y_{j-1}=y, I^{j-2}, Y^{j-2}) &= \bar{i} \bar{\gamma}, \\
P(I_j=1, Y_j=y \mid I_{j-1}=0, Y_{j-1}=y, I^{j-2}, Y^{j-2}) &= i\alpha, \\
P(I_j=1, Y_j=\bar{y} \mid I_{j-1}=0, Y_{j-1}=y, I^{j-2}, Y^{j-2}) &= i \bar{\alpha}. \\
\end{split}
\een
Next consider \[ P(I_j, Y_j \mid (I_{j-1}, Y_{j-1}, Y_{j-2}) = (1,y,x), \  I^{j-2}, Y^{j-3}).\] Since $I_{j-1}=1$, $Y_{j-2}$ is the most recent input bit (say, $X_a$) before
$Y_j$. Also note that $Y_j$ is the input bit $X_{a+1}$ since $Y_{j-1}$ is an insertion. (At most one insertion can occur after each input bit.) Hence
{\small{ \[ P(I_j=0, Y_j=x \mid | (I_{j-1}, Y_{j-1}, Y_{j-2}) = (1,y,x), \ I^{j-2}, Y^{j-2}) \] }}
is just the probability that $X_{a+1}=X_a$, which is equal to $\gamma$. Similarly,
{\small{ \[ P(I_j=0, Y_j=\bar{x} \mid  (I_{j-1}, Y_{j-1}, Y_{j-2}) = (1,y,x), \ I^{j-2}, Y^{j-2})\] }} equals  $1-\gamma$.
\end{IEEEproof}
\emph{\textbf{Remark}}: Proposition \ref{prop:ins_sy} implies that the process $\{\mathbf{I}, \mathbf{Y}\}$ can be characterized as a Markov chain with
state at time $j$ given by $(I_{j}, Y_{j}, Y_{j-1})$. This is an aperiodic, irreducible Markov chain. Hence a stationary distribution $\pi$
exists, which for $y\in\{0,1\}$ can be verified to be
\be \label{eq:stat_pi}
\begin{split}
& \pi(I_{j}=1, Y_{j}=y, Y_{j-1}=y)= \frac{i\alpha}{2(1+i)},  \\
& \pi(I_{j}=1, Y_{j}=\bar{y}, Y_{j-1}=y)= \frac{i\bar{\alpha}}{2(1+i)}, \\
& \pi(I_{j}=0, Y_{j}=y, Y_{j-1}=y)= \frac{\bar{i}\gamma+ i \alpha \gamma +  i \bar{\alpha} \bar{\gamma}}{2(1+i)}, \\
&\pi(I_{j}=0, Y_{j}=\bar{y}, Y_{j-1}=y)= \frac{\bar{i} \bar{\gamma} + i \alpha \bar{\gamma} + i \bar{\alpha} \gamma }{2(1+i)}.
\end{split}
\ee
\begin{lem}
\label{lem:ins_H_I_Y}
\[ \limsup_{n \to \infty} \frac{1}{n} H_P(I^{M_n}|Y^{M_n}) = (1+i) \limsup_{m \to \infty} \frac{1}{m} H_P(I^{m}|Y^{m}). \]
 \end{lem}
\begin{IEEEproof}
See Appendix \ref{proof:ins_H_I_Y}.
\end{IEEEproof}
\begin{lem}
\label{lem:lim_HI_Y}
\[ \limsup_{m \to \infty} \frac{1}{m} H_P(I^{m}|Y^{m}) \leq \lim_{j \to \infty} H_P(I_j|I_{j-1}, Y_{j}, Y_{j-1}, Y_{j-2}), \]
and
\be
\begin{split}
& \lim_{j \to \infty} H_P(I_j|I_{j-1}, Y_{j}, Y_{j-1}, Y_{j-2}) \\
& = \frac{(i\alpha + \bar{i}\gamma)}{1+i} h\left(\frac{i\alpha}{i\alpha + \bar{i}\gamma}\right) +  \frac{(i\bar{\alpha} + \bar{i}\bar{\gamma})}{1+i} h\left(\frac{i\bar{\alpha}}{i\bar{\alpha} + \bar{i}\bar{\gamma}}\right).
\end{split}  \label{eq:hI_limit} \ee
\end{lem}
\begin{IEEEproof}
See Appendix \ref{proof:lim_HI_Y}.
\end{IEEEproof}
\emph{Bounding the penalty term} 
We next focus on the penalty term $H(I^{M_n}|Y^{M_n}, X^n)$ which is the uncertainty in the positions of the insertions given both the channel input and output sequences.
Consider the following example.
\be
\begin{split}
 \un{X} & = \ldots \ 0  \overbrace{1 1 1 1 1 1 }^{k_1 \text{ bits}} 0 \ \ldots \\
\un{Y}  & = \ldots \ 0  \underbrace{1 1 1 1 1 1 1  1 1}_{k_1 + k_2 \text{ bits}} 0 \ \ldots
\end{split}
\label{eq:ins_penalty_example}
\ee
Assume that the value of $I_j$ is known for all bits in $\un{Y}$ except the run of $k_1 + k_2$ ones shown above.  Further
Suppose that it is known that the $\un{X}$-bits shown in the first line  of \eqref{eq:ins_penalty_example} exactly correspond to the $\un{Y}$-bits in the second line.
 For any $k_1 \geq 1$ and $1 \leq k_2 \leq k_1$, the following are all the insertion patterns that are consistent with the shown $(\un{X}, \un{Y})$ pair:
\begin{itemize}
\item[-] The $0$ preceding the $\un{X}$-run undergoes a complementary insertion leading to the first $1$ in the $\un{Y}$-run. Then $(k_2-1)$ out of the $k_1 \  1$'s in the $\un{X}$-run undergo duplications, the remaining $1$'s are transmitted without any insertions.
\item[-]  The $0$ preceding the $\un{X}$-run is transmitted without any insertions.  $k_2$ of the $k_1 \ 1$'s in the $\un{X}$-run undergo duplications, the remaining are transmitted without insertions.
\end{itemize}
For the same  $(\un{X},\un{Y})$ pair, the first scenario above leads to $\tbinom{k_1}{k_2-1}$ different $\un{I}$ sequences, and  the second leads to another $\tbinom{k_1}{k_2}$ $\un{I}$'s. Calculating the entropy associated with these patterns yields a lower bound on the penalty term. This intuition is made rigorous in the following lemma.
\begin{lem}
$ \liminf_{n \to \infty}\frac{1}{n}H_P(I^{M_n}|Y^{M_n}, X^n)  \geq  \Pi(i,\alpha, \gamma)$ where
{\small{
\be
\begin{split}
& \Pi (i,\alpha, \gamma) = {\bar{\gamma}}^2     \sum_{k_1=1}^\infty \sum_{k_2=1}^{k_1} \binom{k_1}{k_2} \gamma^{k_1-1} (i\alpha)^{k_2}  (1-i)^{k_1-k_2+1}  \\
& \  \cdot \left( 1 + \frac{\bar{\alpha}k_2}{\alpha(k_1-k_2+1)} \right)  \left[ \binom{k_1}{k_2-1} \frac{\bar{\alpha}}{\kappa} \log \frac{\kappa}{\bar{\alpha}}
+ \dbinom{k_1}{k_2} \frac{\alpha}{\kappa} \log \frac{\kappa}{\alpha} \right]
\end{split}
\label{eq:pi_def} \ee }}
with $ \kappa \triangleq \tbinom{k_1}{k_2-1} \bar{\alpha}  +  \tbinom{k_1}{k_2} \alpha$.
\label{lem:lb1_liminf}
\end{lem}
\begin{IEEEproof} See Appendix \ref{app:lb1_liminf_proof}. \end{IEEEproof}

\begin{thm}
\label{thm:ins_lb1}
(LB $1$) The capacity of the insertion channel with parameters $(i,\alpha)$ can be lower bounded as
\ben
\begin{split} C(i,\alpha)  \geq \max_{0 < \gamma <1} \Big[ & h(\gamma)
- (i\alpha + \bar{i}\gamma) h\left(\frac{i\alpha}{i\alpha + \bar{i}\gamma}\right)  \\
& - (i\bar{\alpha} + \bar{i} \bar{\gamma}) h\left(\frac{i\bar{\alpha}}{i\bar{\alpha} + \bar{i}\bar{\gamma}}\right)
+ \Pi (i,\alpha, \gamma)\Big]  \end{split} \een
where  $\Pi (i,\alpha, \gamma)$ is defined in \eqref{eq:pi_def}.
\end{thm}
\begin{IEEEproof}
Using Lemmas \ref{lem:ins_H_I_Y}, \ref{lem:lim_HI_Y} and \ref{lem:lb1_liminf} in  \eqref{eq:HXY_HIY} we obtain the RHS above, which is a lower bound on the insertion capacity due to \eqref{eq:mut_decomp}. We optimize the lower bound by maximizing over the Markov parameter $\gamma \in (0,1)$.
\end{IEEEproof}
\subsection{Lower Bound $2$} \label{subsec:lb2}
For any input-pair $(X^n, Y^{M_n})$, define an auxiliary sequence $T^{M_n}=(T_1, \ldots, T_{M_n})$ where $T_j=1$ if $Y_j$ is a
\emph{complementary} insertion, and $T_j=0$ otherwise. The sequence $T^{M_n}$ indicates the positions of the complementary insertions in $Y^{M_n}$.
Note that $T^{M_n}$ is different from the sequence $I^{M_n}$, which indicates the positions of \emph{all} the insertions.
Using $T^{M_n}$, we can decompose $H_P(X^n|Y^{M_n})$ as
\be
\begin{split}
& H_P(X^n|Y^{M_n}) = H_P(X^n, T^{M_n}|Y^{M_n}) - H_P(T^{M_n}| X^n, Y^{M_n}) \\
&=  H_P(T^{M_n}|Y^{M_n}) + H_P(X^n|T^{M_n}, Y^{M_n}) \\
 &  \quad - H_P(T^{M_n}| X^n, Y^{M_n}).
\end{split}
\label{eq:ins_lb2_decomp}
\ee
Using this we have
\be
\begin{split}
& \liminf_{n \to \infty} \frac{1}{n} I_P(X^n ; Y^{M_n}) = h(\gamma) - \limsup_{n \to \infty} \frac{1}{n} H_P(X^n|Y^{M_n})   \\
&\geq h(\gamma) - \limsup_{n \to \infty} \frac{1}{n} H_P(T^{M_n}|Y^{M_n}) \\
& \  - \limsup_{n \to \infty} \frac{1}{n} H_P(X^n|T^{M_n}, Y^{M_n}) \\
&  \ + \underbrace{\liminf_{n \to \infty} \frac{1}{n}H_P(T^{M_n}| X^n, Y^{M_n})}_{\text{penalty term}}.
\end{split}
\label{eq:ins_lb2a}
\ee
We  obtain a lower bound on the insertion capacity by bounding each of the limiting terms in \eqref{eq:ins_lb2a}.
\begin{lem}
\label{lem:ins_H_T_Y}
\[ \limsup_{n \to \infty} \frac{1}{n}H_P(T^{M_n}|Y^{M_n}) = (1+i) \limsup_{m \to \infty} \frac{1}{m} H_P(T^{m}|Y^{m}). \]
 \end{lem}
\begin{IEEEproof}
The proof of this lemma is identical to that of Lemma \ref{lem:ins_H_I_Y}, and can be obtained by replacing $I^{M_n}$ with $T^{M_n}$.
\end{IEEEproof}
\begin{lem}
\label{lem:ins_lim_HT_Y}
\[ \limsup_{m \to \infty} \frac{1}{m} H_P(T^{m}|Y^{m}) \leq \lim_{j \to \infty} H_P(T_j|T_{j-1}, Y_{j}, Y_{j-1}) \]
and
\be
\begin{split}
& \lim_{j \to \infty} H_P(T_j|T_{j-1}, Y_{j}, Y_{j-1}) \\
& = \frac{(1-\gamma+ \gamma i \bar{\alpha})}{(1+i)} h\left(\frac{i \bar{\alpha}}{1-\gamma+ \gamma i \bar{\alpha}}\right ).
\end{split}
 \label{eq:hT_limit} \ee
\end{lem}
\begin{IEEEproof}
See Appendix \ref{proof:ins_lim_HT_Y}.
\end{IEEEproof}
We now derive two upper bounds on  the limiting behavior of $\frac{1}{n}H({X}^n | T^{M_n}, Y^{M_n})$.  Define $\tilde{Y}^{M_n}$  as the sequence obtained from
 $(Y^{M_n}, T^{M_n})$ by flipping the complementary insertions in $Y^{M_n}$, i.e., flip bit $Y_j$ if $T_j=1$. $\tilde{Y}^{M_n}$ has insertions in the same locations as $Y^{M_n}$, but the insertions are all duplications. Hence $\tilde{Y}^{M_n}$ has the same number of runs as $X^n$. Recall from Section \ref{sec:prelim} that we can represent both binary sequences in terms of their run-lengths as
\[
X^n  \leftrightarrow (L^X_1, \ldots, L^X_{R_n}), \qquad \tilde{Y}^{M_n}  \leftrightarrow (L^{\tilde{Y}}_1, \ldots, L^{\tilde{Y}}_{R_n}),
\]
where $R_n$, the number of runs in $X^n$ (and $\tilde{Y}^n$) is a random variable.  Therefore, for all $n$ we have the upper bound
\begin{equation}
\begin{split}
H_P({X}^n| {Y}^{M_n}, T^{M_n}) & \leq H_P({X}^n| \tilde{Y}^{M_n}) \\
& = H_P(L^X_1, \ldots, L^X_{R_n}|L^{\tilde{Y}}_1, \ldots, L^{\tilde{Y}}_{R_n})
\label{eq:ins_bits_to_runs}
\end{split}
\end{equation}
where the  inequality holds because $\tilde{Y}^{M_n}$ is a function of  $(Y^{M_n}, T^{M_n})$.

We can obtain another upper bound by removing the complementary insertions from $Y^{M_n}$. Define $\un{\hat{Y}}$ as the sequence obtained from $(Y^{M_n}, T^{M_n})$ by deleting the complementary insertions. Let $\hat{M}_n$ denote the length of $\un{\hat{Y}}$. Since all the complementary insertions have been removed, $\un{\hat{Y}} = \hat{Y}^{\hat{M}_n}$ has the same number of runs as $X^n$. We therefore have the bound
\begin{equation}
\begin{split}
H_P({X}^n | {Y}^{M_n}, T^{M_n}) & \leq H_P({X}^n| \hat{Y}^{\hat{M}_n})\\
&  = H_P(L^X_1, \ldots, L^X_{R_n}|L^{\hat{Y}}_1, \ldots, L^{\hat{Y}}_{R_n}).
\label{eq:ins_bits_to_runs2}
\end{split}
\end{equation}
\begin{prop} \label{prop:lxytil}
The processes \[ \{\mathbf{L^X, L^{\tilde{Y}}}\} \triangleq \{(L^X_1, L^{\tilde{Y}}_1), (L^X_2, L^{\tilde{Y}}_2), \ldots \} \]
and \[ \{\mathbf{L^X, L^{\hat{Y}}}\} \triangleq \{(L^X_1, L^{\hat{Y}}_1), (L^X_2, L^{\hat{Y}}_2), \ldots \} \]
are both i.i.d processes characterized by the following joint distributions for all $j \geq 1$.  For $r \geq 1$ and $ r\leq s \leq 2r$,
\be \begin{split}
& P(L^X_j=r, L^{\tilde{Y}}_j=s) \\
&  =   \gamma^{r-1} (1-\gamma) \cdot {r \choose {s-r}} i^{s-r} (1-i)^{2r-s},
\label{eq:joint_xytil}
\end{split} \ee
\be
\begin{split}
& P(L^X_j=r, L^{\hat{Y}}_j=s)  \\
& = \gamma^{r-1} (1-\gamma) \cdot {r \choose {s-r}} (i\alpha)^{s-r} (1-i\alpha)^{2r-s}.
\end{split} \label{eq:joint_xyhat} \ee
\end{prop}
\begin{IEEEproof} Since $\mathbf{X}$ is a Markov process, $\{L^X_j\}_{j\geq 1}$ are independent with
 \[ P(L^X_j=r)= \gamma^{r-1} (1-\gamma), \:r=1,2,\ldots\]
 $\tilde{Y}^{M_n}$ is generated from $X^n$ by independently duplicating each bit with probability $i$. Hence $L^{\tilde{Y}}_j$ can be thought of being obtained by passing a run of length $L^X_j$ through a discrete memoryless channel with transition probability
\[ P(L^{\tilde{Y}}_j=s|L^X_j=r) =  {r \choose {s-r}} i^{s-r} (1-i)^{2r-s}, \ r\leq s \leq 2r.  \]
\eqref{eq:joint_xyhat} can be obtained in a similar fashion by observing that  $\hat{Y}^{\hat{M}_n}$ is generated from $X^n$ by independently duplicating each bit with probability $i\alpha$.
\end{IEEEproof}
\begin{lem} \label{lem:ins_x_ytil}
\begin{align*}
\limsup_{n \to \infty} \frac{1}{n} H_P(X^n|T^{M_n}, Y^{M_n}) & \leq \lim_{n \to \infty} \frac{1}{n} H_P({X}^n| \tilde{Y}^{M_n})  \\
& = (1-\gamma) H_P(L^X_1|L^{\tilde{Y}}_1), \\
\limsup_{n \to \infty} \frac{1}{n} H_P(X^n|T^{M_n}, Y^{M_n}) &  \leq \lim_{n \to \infty} \frac{1}{n} H_P({X}^n| \hat{Y}^{\hat{M}_n}) \\
&  = (1-\gamma) H_P(L^X_1|L^{\hat{Y}}_1),
\end{align*} where the joint distributions  of $(L^X_1, L^{\tilde{Y}}_1)$ and $(L^X_1, L^{\hat{Y}}_1)$  are given by Proposition \ref{prop:lxytil}.
\end{lem}
\begin{IEEEproof}
See Appendix \ref{proof:ins_x_ytil}.
\end{IEEEproof}

The penalty term $\tfrac{1}{n} H(T^{M_n}|Y^{M_n}, X^n)$ is the uncertainty in the positions of the complementary insertions given both the channel input and output sequences. We lower bound this using a technique very similar  to the one used for the penalty term in Section \ref{subsec:lb1}. Consider again the $(\un{X}, \un{Y})$ pair shown in \eqref{eq:ins_penalty_example} with the knowledge that the $\un{X}$-bits shown in the first line  of \eqref{eq:ins_penalty_example} yielded the $\un{Y}$-bits in the second line. Further assume that the value of $I_j$ is known for all bits in $\un{Y}$ except the run of $k_1 + k_2$ ones shown above.
Denoting the first bit of the run of ones in $\un{Y}$ by $Y_l$, the only remaining uncertainty in $\un{T}$ is in $T_l$. Indeed,
\begin{itemize}
\item[-] $T_l=1$ if the $0$ preceding the $\un{X}$-run undergoes a complementary insertion leading to the first $1$ in the $\un{Y}$-run. Then $(k_2-1)$ out of the $k_1 \  1$'s in the $\un{X}$-run undergo duplications, the remaining $1$'s are transmitted without any insertions.
\item[-]  $T_l=0$ if the $0$ preceding the $\un{X}$-run is transmitted without any insertions.  $k_2$ out of the $k_1 \ 1$'s in the $\un{X}$-run undergo duplications, the remaining are transmitted without insertions.
\end{itemize}
Calculating the binary entropy associated with the two cases above yields a lower bound on the penalty term. This intuition is made rigorous in the following lemma.
\begin{lem} \label{lem:lb2_liminf}
$\liminf_{n \to \infty} \frac{1}{n}H_P(T^{M_n}|Y^{M_n}, X^n)  \geq  \Gamma(i,\alpha, \gamma)$ where
\ben
\begin{split}
&  \Gamma(i,\alpha, \gamma)  =  {\bar{\gamma}}^2   \sum_{k_1=1}^\infty \sum_{k_2=1}^{k_1} \binom{k_1}{k_2} \gamma^{k_1-1} (i\alpha)^{k_2}  (1-i)^{k_1-k_2+1} \\
& \quad  \cdot  \left( 1 + \frac{\bar{\alpha}k_2}{\alpha(k_1-k_2+1)} \right)  h \left( \frac{\bar{\alpha}k_2}{ \bar{\alpha}k_2 + \alpha(k_1-k_2+1)} \right).
\end{split}
\een
\end{lem}
\begin{IEEEproof}
See Appendix \ref{proof:lb2_liminf}.
\end{IEEEproof}
\begin{figure}
\centering
\includegraphics[width=3.3in]{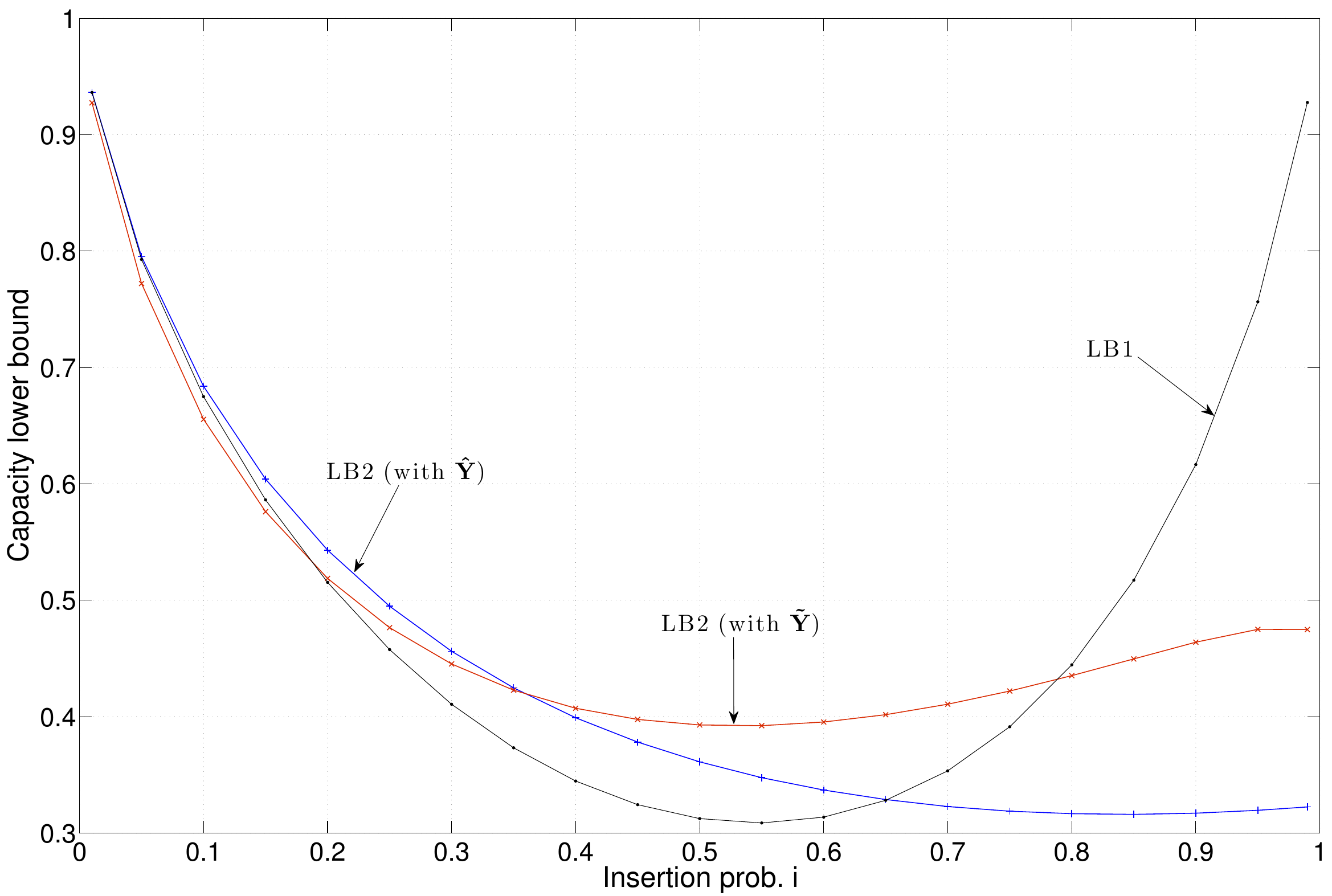}
\caption{Comparison of the two lower bounds on $C(i,\alpha)$ with $\alpha=0.8$}
\vspace{-2pt}
\label{fig:ins_lbs_compare}
\end{figure}
\begin{figure}
\centering
\includegraphics[width=3.5in]{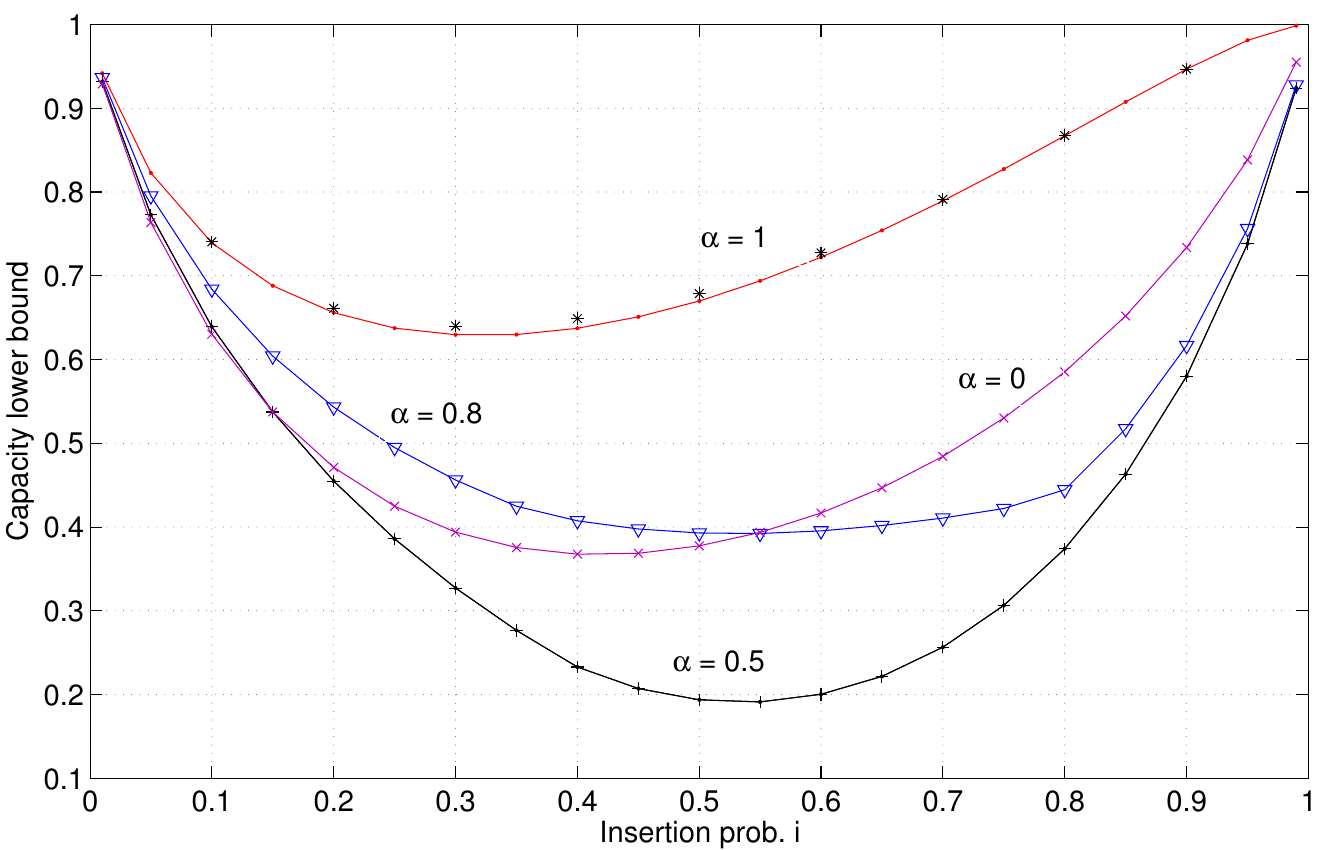}
\caption{Lower bound  on the insertion capacity $C(i,\alpha)$ for  $\alpha =1, 0.8$. For $\alpha=1$, the lower bound of \cite{Mitz_sticky} is shown using *.}
\vspace{-1pt}
\label{fig:ins_lb_combined}
\end{figure}
\begin{thm} \label{thm:ins_lb2}  (LB $2$)
The capacity of the insertion channel with parameters $(i,\alpha)$ can be lower bounded as
\ben
\begin{split}
& C(i,\alpha) \geq \max_{0 < \gamma <1} \  \Big[ h(\gamma) - \left(\bar{\gamma} +
\gamma i \bar{\alpha}\right) h\left(\frac{i \bar{\alpha}}{\bar{\gamma} + \gamma i \bar{\alpha}}\right) \\
& \quad  - \bar{\gamma} \min\{H({L_{X}}_1|{L_{\tilde{Y}}}_1), H({L_{X}}_1|{L_{\hat{Y}}}_1) \} + \Gamma(i,\alpha, \gamma) \Big]
\end{split} \een
where $H({L_{X}}_1|{L_{\tilde{Y}}}_1),  H({L_{X}}_1|{L_{\hat{Y}}}_1)$ are computed using the joint distributions given in Proposition \ref{prop:lxytil}.
\end{thm}
\begin{IEEEproof}
Using Lemmas \ref{lem:ins_H_T_Y}, \ref{lem:ins_lim_HT_Y}, \ref{lem:ins_x_ytil} and \ref{lem:lb2_liminf}
in  \eqref{eq:ins_lb2a} we obtain the RHS above, which is a lower bound on the insertion capacity due to \eqref{eq:mut_decomp}.
We optimize the lower bound by maximizing over the Markov parameter $\gamma \in (0,1)$.
\end{IEEEproof}
We note that the infinite sums in the formulas for $\Pi(i,\alpha,\gamma)$ and $\Gamma(i,\alpha,\gamma)$  can be truncated  to compute the capacity lower bounds  since they appear with positive signs in Theorems \ref{thm:ins_lb1} and \ref{thm:ins_lb2}.  Figure \ref{fig:ins_lbs_compare} compares $LB 1$ with  $LB 2$ for different values of $i$ with $\alpha$ fixed at $0.8$. We show the bounds obtained by evaluating $LB \ 2$ separately with $\hat{Y}$ (deleting the complementary insertions) and $\tilde{Y}$ (flipping the complementary insertions). For smaller insertion probabilities we observe the $\hat{Y}$-bound is better, i.e., $H({L_{X}}_1|{L_{\hat{Y}}}_1) < H({L_{X}}_1|{L_{\tilde{Y}}}_1)$.
We also observe that $LB \ 2$  is generally a better bound than $LB \ 1$, except when $i$ is large. For large $i$, it is more efficient to decode  the positions of all the insertions  rather than just the complementary insertions. Specifically, comparing Lemmas \ref{lem:lim_HI_Y} and \ref{lem:ins_lim_HT_Y},
\[ \lim_{j \to \infty} H(I_j|I_{j-1}, Y_j, Y_{j-1},  Y_{j-2}) \leq \lim_{j \to \infty} H(T_j|T_{j-1}, Y_j, Y_{j-1})\] for large values of $i$ because $I_j$ is very likely to be $1$ if $I_{j-1}=0$. (Recall that $I_j=0$ whenever $I_{j-1}=1$.)
Combining the bounds of Theorems \ref{thm:ins_lb1} and \ref{thm:ins_lb2}, we observe that $\max\{LB\ 1, LB\ 2\}$ is a lower bound to the insertion capacity. This is plotted in  Figure \ref{fig:ins_lb_combined} for various values of $i$, for $\alpha= 0, 0.5,0.8,1$. The lower bound is not monotonic in $\alpha$ for a given $i$. This is because the curve is the maximum of three different bounds ($LB \ 1$, $LB \ 2$ with $\tilde{Y}$ and $\hat{Y}$), each of which has a different behavior if we vary $i, \alpha$.

For $\alpha=1$, the bound is very close to the near-optimal lower bound in \cite{Mitz_sticky}. The difference is entirely due to using a first order Markov input distribution rather than the numerically optimized input distribution in \cite{Mitz_sticky}.

\section{Deletion Channel}\label{sec:deletion}
In this channel, each input bit is deleted with probability $d$, or retained with probability $1-d$. For any input-output pair $(X^n, Y^{M_n})$, define the auxiliary sequence $S^{M_n+1}$, where $S_j \in \mathbb{N}_0$ is the number of {runs} \emph{completely} deleted in ${X^n}$ between the bits corresponding to $Y_{j-1}$ and $Y_{j}$. ($S_1$ is the number of runs deleted before the first output symbol $Y_1$, and $S_{M_n+1}$ is the number of runs deleted after the last output symbol $Y_{M_n}$.) Examples of $S^{M_n}$ for the input-output pair  $(\un{X} = 000111000, \ \un{Y}=000)$ were given in Section \ref{subsec:deletion_scheme}.

The auxiliary sequence $\un{S}$ lets us augment $\un{Y}$ with the positions of missing runs. Consider $\un{X} = 00{{011100}}0$. If the decoder is given  $\un{Y}=000$ and $\un{S}=(0,0,0,2)$, it can form the augmented sequence $ \un{Y}'=000- -$, where a $-$ denotes a missing run, or equivalently a run of length $0$ in $\un{Y}$.
With the ``$-$'' markers indicating deleted runs, we can associate each run of the augmented  sequence $\un{Y}'$ uniquely with a run in $\un{X}$. Denote by
 $L^{Y'}_1, L^{Y'}_2, \ldots $ the run-lengths of the augmented sequence $\un{Y}'$, where $L^{Y'}_j = 0$ if run $j$ is a $-$.
 Then we have
 \be
P(\un{X}, \un{Y}') = P(L^X_1)P(L^{Y'}_1|L^X_1) \cdot P(L^X_2) P(L^{Y'}_2|L^X_2) \ldots
 \ee
where $\forall j$:
\be \label{eq:aug_joint_dist}
\begin{split}
P(L^X_j=r)& = \gamma^{r-1}(1-\gamma), \quad r=1,2,\ldots \\
P(L^{Y'}_j=s|L^X_j=r) & =  {r \choose s}d^{r-s} (1-d)^s,  \quad 0\leq s\leq r.
\end{split}
\ee
Using the auxiliary sequence $S^{M_n +1}$, we can write
\be
\begin{split}
H_P({X}^n | {Y}^{M_n}) = & H_P({X}^n , {S}^{M_n+1}| {Y}^{M_n})  \\
& - H_P({S}^{M_n+1}|{X}^n, {Y}^{M_n}).
\end{split}
 \label{eq:del_decomp}
\ee
We therefore have
\be
\begin{split}
 & \liminf_{n \to \infty}  \frac{1}{n} I_P({X}^n ;  {Y}^{M_n})   = h(\gamma) -  \limsup_{n \to \infty}  \frac{1}{n} H_P({X}^n | {Y}^{M_n})  \\
 & \geq h(\gamma) -  \limsup_{n \to \infty} \frac{1}{n} H_P({X}^n , {S}^{M_n+1}| {Y}^{M_n})  \\
 & \quad  + \underbrace{\liminf_{n \to \infty} \frac{1}{n} H_P({S}^{M_n+1}|{X}^n, {Y}^{M_n})}_{\text{penalty term}}.
\end{split}
 \label{eq:del_lb1}
\ee
We will show that $\lim_{n \to \infty} \frac{1}{n} H_P({X}^n,  {S}^{M_n+1}| {Y}^{M_n})$ exists, and obtain an analytical expression for this limit.  We also derive a lower bound on the penalty term, thereby obtaining a lower bound on the deletion capacity. We remark that it has been shown in \cite{DrineaK10} that for any input distribution with independent runs, $\lim_{n \to \infty} \frac{1}{n}H(X^n|Y^{M_n})$ exists for the deletion channel. Hence the $\liminf$ on the left hand side of \eqref{eq:del_lb1} is actually a limit.

\begin{prop} \label{prop:y}
The process $\mathbf{Y} =\{Y_1,Y_2, \ldots\}$ is a first-order Markov process characterized by the following joint distribution for
all $m \in \mathbb{N}$.
\[P(Y^m) = P(Y_1) \prod_{j=2}^m P(Y_j|Y_{j-1})\]
where for $y \in \{0,1\}$, $P(Y_1=1)=0.5$ and
\be \begin{split}  P(Y_j=y|Y_{j-1}=y) & =1-\ P(Y_j=\bar{y}|Y_{j-1}=y) \\
& = \frac{\gamma+d-2\gamma d}{1+d-2\gamma d}. \end{split}   \label{eq:del_ydist}   \ee
\end{prop}
\proof The proof of this proposition can be found in \cite{Mitzenmacher09}.
\begin{prop} \label{prop:sy}
The process $\{\mathbf{S}, \mathbf{Y}\} \triangleq \{(S_1,Y_1), (S_2,Y_2), \ldots\}$ is a first-order Markov process characterized by the following joint distribution for all  $m \in \mathbb{N}$:
\[ P(S^{m}, Y^m) = P(Y_1, S_1) \prod_{j=2}^{m}  P(Y_j, S_j|Y_{j-1}),
\]
where  for $y \in \{0,1\}$ and $j \geq 2$:
\begin{align}
& P(Y_j={y}, S_j=k|Y_{j-1}=y) \nonumber\\
&= \left\{
\begin{array}{ll}
\frac{\gamma(1-d)}{(1-\gamma d)}, & k=0\\
\frac{(1-d)(1-\gamma)}{(1-\gamma d)^2} \left(\frac{d(1-\gamma)}{1-\gamma d}\right)^k, & k=1,3, \ldots\\
0, & \text{ otherwise}
\end{array}
\right.  \label{eq:del_pys_y}\\
& \text{and }  \nonumber \\
& P(Y_j=\bar{y}, S_j=k|Y_{j-1}=y) \nonumber\\
&= \left\{
\begin{array}{ll}
\frac{(1-d)(1-\gamma)}{ (1-\gamma d)^2}\left(\frac{d(1-\gamma)}{1-\gamma d}\right)^k, & k=0,2, \ldots\\
0, & \text{ otherwise}
\end{array}
\right.  \label{eq:del_pbarys_y}
\end{align}
\end{prop}
\begin{IEEEproof}
In the sequel, we use the shorthand notation $\un{v}$ to denote the sequence $(s_{j-1}, y_{j-2}, s_{j-2}, y_{j-3}, s_{j-3},  \ldots)$. We need to show that
\ben
 \begin{split}
 & P(Y_j=y, S_j=k  \mid  (Y_{j-1}, S_{j-1}, Y_{j-2}, S_{j-2}, \ldots) = (y, \un{v}) ) \\
 & = P(Y_j=y, S_j=k \mid Y_{j-1}=y_{j-1}),
\end{split}
\een
for all $y,y_{j-1},y_{j-2}, \ldots \in \{0,1\}$ { and } $k,s_{j-1}, s_{j-2} \ldots \in \mathbb{N}_0$.

Let the output symbols $Y_j, Y_{j-1}, Y_{j-2}, \ldots$ correspond to input symbols $X_{a_j}, X_{a_{j-1}}, X_{a_{j-2}}, \ldots$ for some positive integers
$a_j>  a_{j-1} > a_{j-2}> \ldots$. $S_{j-1}$ is the number of runs between the input symbols $X_{a_{j-2}}$ and $X_{a_{j-1}}$, not counting the runs containing
$X_{a_{j-2}}$ and $X_{a_{j-1}}$. Similarly,  $S_{j-2}$ is the number of runs between the input symbols $X_{a_{j-3}}$ and $X_{a_{j-2}}$, not counting the runs containing
$X_{a_{j-3}}$ and $X_{a_{j-2}}$ etc.

First consider the case where $Y_{j}=Y_{j-1}=y$. When $Y_j=X_{a_j}=y$ and $Y_{j-1}=X_{a_{j-1}}=y$,  note that $S_j$, the number of completely deleted runs between $X_{a_{j-1}}$ and $X_{a_j}$, is either zero or an odd number. We have
\be
\begin{split}
& P(Y_j=y, S_j=0  \mid (Y_{j-1}, S_{j-1}, Y_{j-2}, S_{j-2}, \ldots) = (y, \un{v})) \\
& \stackrel{(a)}{=} \sum_{m=1}^{\infty} \gamma^{m} (1-\gamma) (1-d^m) = \frac{\gamma (1-d)}{(1-\gamma d)}
\end{split}
\ee
where $(a)$ is obtained as follows. $\gamma^{m} (1-\gamma)$ is the probability that the input run containing $X_{a_{j-1}}$ contains $m$ bits after $a_{j-1}$, and $(1-d^m)$ is the probability that at least one of them is not deleted. This needs to hold for some $m\geq 1$ in order to have $S_j=0$ and $Y_j=Y_{j-1}$.
By reasoning similar to the above, we have for $k=1,3,5,\ldots$:
\be
\begin{split}
& P(Y_j=y, S_j=k  \mid (Y_{j-1}, S_{j-1}, Y_{j-2}, S_{j-2}, \ldots) = (y, \un{v})) \\
&\stackrel{(b)}{=} \left(\sum_{m=0}^{\infty} \gamma^{m} (1-\gamma) d^m\right) \ \left(\sum_{m=1}^{\infty} \gamma^{m-1} (1-\gamma) d^m \right)^k \\
& \quad \  \left(\sum_{m=1}^{\infty} \gamma^{m-1} (1-\gamma) (1-d^m)\right)\\
&=\frac{(1-\gamma)(1-d)} {(1-\gamma d)^2} \left[\frac{d (1-\gamma)}{(1-\gamma d)}\right]^k
\end{split}
\ee
where the first term in $(b)$ is the probability that the remainder of the run containing $X_{a_{j-1}}$ is completely deleted, the second term is the probability that the next $k$ runs are deleted, and the last term  is the probability that the subsequent run is \emph{not} completely deleted.

When $Y_{j}=y$ and $Y_{j-1}=\bar{y}$, the number of deleted runs $S_j$ is either zero or an even number.
For $k=0,2,4, \ldots$ we have
\be
\begin{split}
&P(Y_j=y, S_j=k |(Y_{j-1}, S_{j-1}, Y_{j-2}, S_{j-2}, \ldots) = (\bar{y},\un{v}))\\
&\stackrel{(c)}{=} \left(\sum_{m=0}^{\infty} \gamma^{m} (1-\gamma) d^m \right) \left(\sum_{m=1}^{\infty} \gamma^{m-1} (1-\gamma) d^m\right)^k \\
& \quad \ \left(\sum_{m=1}^{\infty} \gamma^{m-1} (1-\gamma) (1-d^m) \right)\\
& =\frac{(1-\gamma)(1-d)} {(1-\gamma d)^2} \left[\frac{d (1-\gamma)}{(1-\gamma d)}\right]^k.
\end{split}
\ee
In the above, the first term in $(c)$ is the probability that the remainder of the run containing $X_{a_{j-1}}$ is completely deleted,
the second term is the probability that the next
$k$ runs are deleted ($k$ may be equal to zero), and the third term is the probability that the subsequent run is not completely deleted.
This completes the proof of the lemma.
\end{IEEEproof}

We now show that $\lim_{n \to \infty} \frac{1}{n} H_P({S}^{M_n +1}|{Y}^{M_n})$ and
$\lim_{n \to \infty} \frac{1}{n} H_P(X^n|Y^{M_n}, S^{M_n+1})$ each exist, thereby proving the existence of
$\lim_{n \to \infty} \frac{1}{n} H_P(X^n, {S}^{M_n +1}|{Y}^{M_n})$.
\begin{lem}
$\lim_{n \to \infty} \frac{1}{n} H_P(S^{M_n+1}|Y^{M_n}) = \bar{d} H_P(S_2|Y_1 Y_2)$ where the joint distribution of $(Y_1,Y_2, S_2)$ is given by \eqref{eq:del_ydist}, \eqref{eq:del_pys_y}, and \eqref{eq:del_pbarys_y}.
\label{lem:del_s_y}
\end{lem}
\begin{IEEEproof}
See Appendix \ref{proof:del_s_y}.
\end{IEEEproof}

To determine the limiting behavior of $\frac{1}{n}H({X}^n| {S}^{M_n +1}, Y^{M_n})$, we recall that $X^n$ can be equivalently represented in terms of its
run-lengths as $(L^X_1, \ldots, L^X_{R_n})$, where $R_n$, the number of runs in $X^n$, is a random variable. Also recall from the discussion at the beginning
of this section that the pair of sequences $({S}^{M_n +1}, Y^{M_n})$ is equivalent to an augmented sequence $\un{Y}'$ formed by adding the positions of the deleted runs to $\un{Y}=Y^{M_n}$.  $\un{Y}'$ can be equivalently represented in terms of its  run-lengths as $(L^{Y'}_1, \ldots, L^{Y'}_{R_n})$, where we emphasize that $L^{Y'}_1,L^{Y'}_2, \ldots$ can take value $0$ as well.  To summarize, we have
\be
\begin{split}
X^n  \leftrightarrow (L^X_1, \ldots, L^X_{R_n}), \quad  ({S}^{M_n +1}, Y^{M_n})  \leftrightarrow (L^{Y'}_1, \ldots, L^{Y'}_{R_n}).
\end{split}
\label{eq:run_corresp}
\ee
Thus, for all $n$
\begin{equation}
H_P({X}^n| {S}^{M_n +1}, Y^{M_n}) = H_P(L^X_1, \ldots, L^X_{R_n}|L^{Y'}_1, \ldots, L^{Y'}_{R_n}).
\label{eq:bits_to_runs}
\end{equation}
\begin{prop} \label{prop:lxy'}
The process $\{\mathbf{L^X, L^{Y'}}\} \triangleq \{(L^X_1, L^{Y'}_1), (L^X_2, L^{Y'}_2), \ldots \}$ is an i.i.d process characterized by the
following joint distribution for all $j \geq 1$:
\be \begin{split}
P(L^X_j=r, L^{Y'}_j=s) = \gamma^{r-1} (1-\gamma) \cdot {r \choose s} d^{r-s} (1-d)^s,& \\
0\leq s \leq r, \  \ r=1,2,\ldots  &
\end{split}
\label{eq:del_xy'_joint}
\ee
\end{prop}
\begin{IEEEproof} Since $\mathbf{X}$ is a Markov process, $\{L^X_j\}_{j\geq 1}$ are independent with
 \[ P(L^X_j=r)= \gamma^{r-1} (1-\gamma), \:r=1,2,\ldots\]
 Since the deletion process is i.i.d, each $L^{Y'}_j$ can be thought of being obtained by passing a run of length $L^X_j$ through a  discrete memoryless channel with transition probability
\[ P(L^{Y'}_j=s|L^X_j=r) =  {r \choose s} d^{r-s} (1-d)^s, \ 0\leq s \leq r.  \]
\end{IEEEproof}
\begin{lem}
\[ \lim_{n \to \infty} \frac{1}{n} H_P({X}^n| {S}^{M_n +1}, Y^{M_n}) = \bar{\gamma} H_P(L^X|L^{Y'}) \]
where the joint distribution of $(L^X, L^{Y'})$ is given by \eqref{eq:del_xy'_joint}.
\label{lem:del_x_sy}
\end{lem}
\begin{IEEEproof}
See Appendix \ref{proof:del_x_sy}.
\end{IEEEproof}

\subsection{Bounding the penalty term} \label{subsec:del_penalty}
The penalty term $H(S^{M_n + 1} | Y^{M_n}, X^n)$ is the uncertainty in the positions of the deleted runs given both the channel input and output sequences. To get some intuition about this term, consider the following example.
\be
\begin{split}
\un{X} =  \overbrace{0 0 0 0 0}^{z \text{ bits}} 111 \overbrace{0 0 0 0 0}^{r \text{ bits}}  \quad  \longrightarrow \quad
\un{Y} =  \overbrace{0 0 0}^{s \text{ bits}}
\end{split}
\label{eq:del_exXY}
\ee
Given $(\un{X}, \un{Y})$ the uncertainty in $\un{S}$ corresponds to how many of the $s$ output bits  came from the first run of zeros in $\un{X}$, and how many came from the second. In \eqref{eq:del_exXY}, $\un{S}$ can be one of four sequences: $(2,0,0,0), (0,0,0,2)$, $(0,1,0,0)$  and $(0,0,1,0)$. The first case corresponds to all the output bits coming from the second run of zeros; in the second case  all the output bits come from the first run. The third and fourth cases correspond to the output bits coming from both input runs of zeros. The probability of the deletion patterns resulting in each of these possibilities can be calculated. We can thus compute
$H(\un{S} | \un{X}, \un{Y})$ precisely for this example.  For general $(\un{X}, \un{Y})$, we lower bound $H(\un{S} | \un{X}, \un{Y})$ by considering  patterns in $(\un{X}, \un{Y})$ of the form shown in \eqref{eq:del_exXY}. This is done in the following lemma.
\begin{lem}
$ \liminf_{n \to \infty}\frac{1}{n}H_P(S^{M_n+1}|Y^{M_n}, X^n)  \geq  \Phi(d, \gamma)$ where
\be
\begin{split}
\Phi(d, \gamma)= & \frac{ \bar{d} \  \bar{q} \  \bar{\gamma}^3 \  d}{\gamma^2 \  (1 - \gamma d)} \sum_{z, r =1}^{\infty} (\gamma d)^{z+r} \
 \sum_{s=1}^{z+r}  \left( \frac{\bar{d}}{d} \right)^s \binom{z+r}{s} \\
 & \cdot H\left( \left\{ \frac{ \tbinom{z}{l} \tbinom{r}{s-l} }{\tbinom{z+r}{s}} \right \}_{l=0, \ldots, s} \right)
\label{eq:Phi_def}
\end{split}
\ee
where $q = \tfrac{\gamma + d - 2\gamma d}{1 + d - 2\gamma d}$ and $H(\{p_i \})$ is the entropy of the pmf $\{p_i \}$. (In
\eqref{eq:Phi_def} is assumed that $\tbinom{n}{k} = 0$ for $k > n$.)
\label{lem:lb_delpenalty}
\end{lem}
\begin{IEEEproof}
See Appendix \ref{app:lb_delpenalty}.
\end{IEEEproof}
\begin{table}[]
\caption{Capacity lower bound for the deletion channel}
\begin{center}
\begin{tabular}{|c|c|c|c|}
\hline
$d$ & LB  of Thm.\ref{thm:del_thm}  & Optimal $\gamma$ & LB of \cite{DrineaM07}  \\
\hline
$0.05$ & $\mathbf{0.7291}$ &  $0.535$ & $0.7283$ \\
\hline
$0.10$ & $\mathbf{0.5638}$ & $0.575$ & $0.5620$ \\
\hline
$0.15$ &  $\mathbf{0.4414}$ &  $0.62$  &  $0.4392$ \\
\hline
$0.20$ & $\mathbf{0.3482}$ & $0.67$  & $0.3467$ \\
\hline
 $0.25$ &  $\mathbf{0.2770}$ & $0.72$  & $0.2759$\\
 \hline
 $0.30$ & $\mathbf{0.2225}$ &  $0.77$  & $0.2224$  \\
 \hline
 $0.35$ &  $0.1805$  &  $0.81$ & $0.1810$  \\
 \hline
 $0.40$ & $0.1478$ & $0.84$  & $0.1484$ \\
 \hline
  $0.45$ & $0.1217$ & $0.87$  & $0.1229$ \\
\hline
 $0.50$ & $0.1005$ & $0.89$  & $0.1019$ \\
\hline
 $0.55$ & $0.0830$ & $0.91$  & $0.0843$ \\
\hline
 $0.60$ & $0.0682$ & $0.925$  & $0.0696$ \\
 \hline
 $0.65$ & $0.0556$ & $0.94$  & $0.0566$ \\
 \hline
  $0.70$ & $0.0446$ & $0.95$  & $0.0453$ \\
  \hline
\end{tabular}
\end{center}
\label{tab:del_comparison}
\vspace{-3pt}
\end{table}
\begin{thm}
The deletion channel  capacity $C(d)$ can be lower bounded as
\ben
\begin{split}
C(d) \geq \max_{0<\gamma<1} & \Big[ h(\gamma) - (1-d) H(S_2|Y_1 Y_2) \\
& - (1-\gamma) H(L^{X}|L^{Y'}) + \Phi(d, \gamma) \Big]
\end{split}
\een
where
\be  \label{eq:del_HS2_y1y2}
\begin{split}
H(S_2|Y_1 Y_2) = & \   {\gamma \bar{\theta}} \log \frac{q}{\gamma \bar{\theta}}
 + \frac{\beta \theta }{(1-\theta)^2} \log \frac{1}{\theta} \\
& +  \frac{\beta \theta }{1-\theta^2} \log \frac{q}{\beta}   +\frac{\beta }{1-\theta^2}\log \frac{\bar{q}}{\beta},
\end{split}
\ee
\[ q = \frac{\gamma + d - 2\gamma d}{1 + d - 2\gamma d}, \quad \theta =  \frac{\bar{\gamma} d}{1-\gamma d} , \quad
\beta = \frac{\bar{\gamma} \bar{d}}{(1-\gamma d)^2}\]
and
\be
\begin{split}
& H(L^{X}|L^{Y'})= \left(\frac{d}{\bar{\gamma}} -\frac{d \bar{\gamma}}{(1-\gamma d)^2}\right)\log \frac{1}{\gamma d} + \frac{d \bar{\gamma} h(d\gamma)}{(1-d\gamma)^2 } \\
& \quad -\frac{\bar{d}(2-\gamma - \gamma d)\log ({1-\gamma d})}{\bar{\gamma}(1-\gamma d)}\\
& \quad - \frac{\bar{\gamma}}{\gamma} \sum_{k=1}^\infty \sum_{j=1}^\infty (\bar{d}\gamma)^k \, (d\gamma)^j\,  {j+k \choose k} \, \log {j+k \choose k}.
\end{split}
\label{eq:del_HLXLY}
\ee
\label{thm:del_thm}
\end{thm}
\begin{IEEEproof}
We obtain the lower bound on the deletion capacity by using Lemmas \ref{lem:del_s_y}, \ref{lem:del_x_sy} and \ref{lem:lb_delpenalty} in \eqref{eq:del_lb1}.
$ H(S_2|Y_1 Y_2)$ is then be computed using the joint distribution given by \eqref{eq:del_ydist}, \eqref{eq:del_pys_y}, and \eqref{eq:del_pbarys_y}. $H(L^{X}|L^{Y'})$ can be computed using the joint distribution given in Proposition  \ref{prop:lxy'}. Finally, we optimize the lower bound by maximizing over the Markov parameter $\gamma$.
\end{IEEEproof}

Since the summations in \eqref{eq:del_HLXLY} and \eqref{eq:Phi_def} appear with positive signs in Theorem \ref{thm:del_thm}, they can be truncated to compute a capacity lower bound.  Table  \ref{tab:del_comparison} shows the capacity lower bound of Theorem \ref{thm:del_thm} for various values of $d$ together with $\gamma \in (0,1)$ optimized with a resolution of  $0.005$. We notice that the maximizing $\gamma$ increases with $d$, i.e., input runs get longer and are hence less likely to be deleted completely. We observe that for $d \leq 0.3$ (values shown in bold) our lower bound improves on that of \cite{DrineaM07}, the best previous lower bound on the deletion capacity.

A sharper lower bound on the penalty term will improve the capacity bound of Theorem \ref{thm:del_thm}.  In deriving $\Phi(d, \gamma)$ in Lemma \ref{lem:lb_delpenalty}, we considered $\un{Y}$-runs obtained from  either one or three adjacent $\un{X}$-runs and lower bounded the conditional entropy
$H(\un{S}| \un{X}, \un{Y})$ assuming all but three $\un{X}$-runs giving rise to the $\un{Y}$-run were known. The lower bound can be refined by additionally considering the cases where $\un{Y}$-run arose from $5/7/\ldots$ adjacent $\un{X}$-runs.  However, this will imply a more complicated formula for $\Phi(d, \gamma)$ in \eqref{eq:Phi_def}.

\subsection{Comparison with Jigsaw Decoding} \label{subsec:jigsaw_comp}
The jigsaw decoder decodes the type of each run in the output sequence $\un{Y}$. The type of a $\un{Y}$-run is the set of input runs that gave rise to it, with the first input run in the set contributing at least one bit. The penalty for decoding the codeword by first decoding the sequence of types is $\tfrac{1}{n}H( \text{types of $\un{Y}$ } | \un{X},\un{Y})$. A characterization of this conditional entropy in terms of the joint distribution of $(\un{X},\un{Y})$ is derived in \cite{DrineaK10}, but we will not need the precise expression for the discussion below.

Given a pair  $(\un{X},\un{Y})$, observe that knowledge of $\un{S}$ uniquely determines the sequence of types of $\un{Y}$, but not vice versa. For example, consider the pair
\be
\un{X} = 1010101, \qquad \un{Y}= 1101
\label{eq:jig_ex}
\ee
Suppose we know that $\un{S}=(0,3,0,0,0)$, i.e., $\un{Y}$ can be augmented with deleted runs as $1 - - - 1 0 1$. Then the types of the three $\un{Y}$ runs are
\be
\{10101\}  \rightarrow {11}, \quad \{ 0 \} \rightarrow 0, \quad \{ 1 \} \rightarrow  1.
\label{eq:type_example}
\ee
In contrast, suppose we know that the set of types for the $(\un{X}, \un{Y})$ pair in \eqref{eq:jig_ex} is as shown in \eqref{eq:type_example}. Then $\un{S}=(0,3,0,0,0)$  and  $\un{S} = (0,1,2,0,0)$   (corresponding to deletion patterns $1 - - - 1 0 1$ and $1 -  1 - -0 1$, respectively) are both consistent $\un{S}$-sequences with the given set of types. In summary, since the set of types is a function of $(\un{X}, \un{Y}, \un{S})$ we have
\[ \frac{1}{n}H( \text{types of $\un{Y}$ } | \un{X},\un{Y}) \leq  \frac{1}{n}H(\un{S}|\un{X},\un{Y}). \]
In other words, the rate penalty incurred by the jigsaw decoder is \emph{smaller} than the penalty of the sub-optimal decoder considered here. However, the penalty term for our decoder can be lower bounded analytically, which leads to improved lower bounds on the deletion capacity for $d \leq 0.3$. The jigsaw penalty term is harder to lower bound and is estimated via simulation for a few values of $d$ in \cite{DrineaK10}.

We  note that the analysis of the jigsaw decoder in \cite{DrineaM07, DrineaK10} relies on two conditions being satisfied: 1) output runs are independent,  2) each output run arises from a set of \emph{complete} input runs.   In a channel with only deletions and duplications, the second condition is always true,  and the first is guaranteed  by choosing the input distribution to be i.i.d across runs.  With complementary insertions, the output runs are dependent even when the input distribution is i.i.d. Further,  we cannot associate each output run with a set of complete input runs. For example, if $ \un{X} = 0 0 0 $  and $\un{Y} = 0  0 1  0$, each output run of zeros corresponds to only a {part} of the input run. It is therefore hard to extend jigsaw decoding to channels where complementary insertions occur.  For such channels, the rate achieved by a  decoder which decodes auxiliary sequences to synchronize the output runs with the runs of the transmitted codeword can still be lower bounded analytically. Though the rate of such a decoder may not be close to capacity, the final capacity bound is higher since it also includes the  lower bound to the penalty term.

\section{InDel Channel} \label{sec:delins_channel}
The InDel channel is defined by three parameters $(d,i,\alpha)$ with $d+i <1$. Each input bit undergoes a deletion with probability $d$, a duplication with probability
$i\alpha$, a complementary insertion with probability $i\bar{\alpha}$.
\begin{figure}
\centering
\includegraphics[width=3.5in]{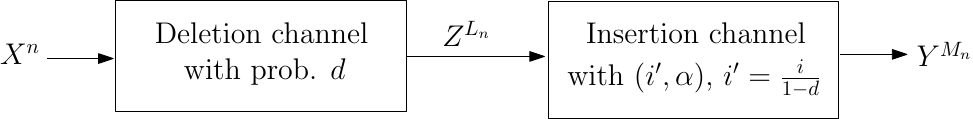}
\caption{Cascade channel equivalent to the InDel channel.}
\label{fig:two_ch_decomp}
\vspace{-3pt}
\end{figure}
Each input bit is deleted with probability $d$; given that a particular bit is \emph{not} deleted, the probability that it undergoes an insertion is $\frac{i}{1-d}$. Therefore, one can think of the channel as a cascade of two channels, as shown in  Figure \ref{fig:two_ch_decomp}. The first channel is a deletion channel that deletes each bit independently with probability $d$. The second channel is an insertion channel with parameters $(i', \alpha)$, where $i' \triangleq \frac{i}{1-d}$. We prove the equivalence of this cascade decomposition below.

\emph{Claim}: The InDel channel is equivalent to the cascade channel in the sense that both have the same transition probability $P(\un{Y}|\un{X})$, and hence the same capacity.
\begin{IEEEproof}
For an $n$-bit input sequence, define the deletion-insertion pattern $\Lambda^n=(\Lambda_1, \Lambda_2,\ldots, \Lambda_n)$ of the channel as the sequence where $\Lambda_i$ indicates whether the channel introduces a deletion/duplication/complementary insertion/no modification in bit $i$ of the input.
 Note that if the underlying probability space is $(\Omega, \mc{F}, P)$, the realization $\omega \in \Omega$ determines the deletion-insertion pattern $\Lambda^n(\omega)$. We calculate the probability of any specified pattern occurring in a)the InDel channel, and b)the cascade channel.

Consider a deletion-insertion pattern $\lambda^n$ with $k$ deletions at positions
$a_1, a_2, \ldots, a_k$, $l$ duplications at positions $b_1, \ldots, b_l$, and $m$ complementary insertions at positions $c_1,\ldots,c_m$.
The probability of this pattern occurring in the InDel channel is
\[P_{delins}(\Lambda^n(\omega)=\lambda^n)= d^k (i\alpha)^l (i\bar{\alpha})^m (1-d-i)^{n-k-l-m}. \]
The probability of this pattern occurring in the cascade channel of Fig. \ref{fig:two_ch_decomp} is
\be
\begin{split}
& P_{casc}(\Lambda^n(\omega)=\lambda^n) \\
& \stackrel{(a)}  = \left[d^k (\bar{d})^{n-k} \right] \left[(i'\alpha)^l (i'\bar{\alpha})^m (1-i')^{(n-k)-l-m}\right]\\
&= \left[ d^k (\bar{d})^{n-k} \right] \left[\left(\frac{i\alpha}{\bar{d}}\right)^l \left(\frac{i\bar{\alpha}}{\bar{d}}\right)^m \left(\frac{1-d-i}{\bar{d}}\right)^{n-k-l-m}\right] \\ &=d^k (i\alpha)^l (i\bar{\alpha})^m (1-d-i)^{n-k-l-m}.
\end{split}
\ee
where the first term in $(a)$ is the probability of deletions occurring in the specified positions in the first channel, and the second term is the probability of the insertions occurring in the specified positions in the second channel. Thus for any fixed pair ($\un{X}, \un{Y}$),  every deletion-insertion pattern  that produces $\un{Y}$ from $\un{X}$ has the same probability in both the InDel channel and the cascade channel. This implies that the two channels have the same transition probability.
\end{IEEEproof}

To obtain a lower bound on the capacity, we work with the cascade channel and use two auxiliary sequences, $T^{M_n}=(T_1, \ldots T_{M_n})$ and
$S^{M_n +1} = (S_1, \ldots, S_{M_n+1})$. As in Section \ref{subsec:lb2}, $T^{M_n}$  indicates  the complementary insertions in $Y^{M_n}$: $T_j=1$ if $Y_j$
is a complementary insertion, and $T_j=0$ otherwise. As in Section \ref{sec:deletion}, $S^{M_n+1}$ indicates the positions of the missing runs: $S_j=k$, if $k$ runs
were  completely deleted between $Y_{j-1}$ and $Y_j$. We decompose $H_P(X^n|Y^{M_n})$ as
\be \label{eq:xyts}
\begin{split}
 & H_P(X^n|Y^{M_n})   =  H_P(X^n, T^{M_n}, S^{M_n +1}|Y^{M_n}) \\
& \qquad \qquad \qquad \quad  - H_P(T^{M_n}, S^{M_n +1}| X^n, Y^{M_n})\\
&=  H_P(X^n|S^{M_n +1}, T^{M_n}, Y^{M_n})   + H_P(S^{M_n +1}| T^{M_n}, Y^{M_n}) \\
& \quad    + H_P(T^{M_n}|Y^{M_n}) - H_P(T^{M_n}, S^{M_n +1}| X^n, Y^{M_n}).
\end{split}
\ee
We therefore have
\be
\begin{split}
 \liminf_{n \to \infty}   \frac{1}{n} I_P  (X^n; &  Y^{M_n} )    \geq   h(\gamma) - \limsup_{n \to \infty}  \frac{1}{n} H_P(T^{M_n}|Y^{M_n})\\
& \ -  \limsup_{n \to \infty} \frac{1}{n} H_P(S^{M_n +1}| T^{M_n}, Y^{M_n}) \\
&  \ - \limsup_{n \to \infty} \frac{1}{n}H_P(X^n|S^{M_n +1}, T^{M_n}, Y^{M_n})   \\
 & \ + \underbrace{\liminf_{n \to \infty} \frac{1}{n}  H_P(T^{M_n}, S^{M_n +1}| X^n, Y^{M_n})}_{\text{penalty term}}.
\end{split}
\label{eq:delins_ineq}
\ee
Using the techniques developed in the previous two sections, we bound each of the limiting terms above to obtain a lower bound on the InDel capacity.

\begin{lem} \label{lem:delins_2}
\[ \limsup_{m \to \infty} \frac{1}{m}H_P(T^{M_n}|Y^{M_n}) \leq  (\bar{q}\bar{d} + q i \bar{\alpha}) h\left(\frac{i \bar{\alpha}}{\bar{q}\bar{d} + q i \bar{\alpha}}\right)\]  where  $q= \tfrac{\gamma+d-2\gamma d}{1+d-2\gamma d}$.
\end{lem}
\begin{IEEEproof}
We first note that
\be \limsup_{n \to \infty}  \frac{ H_P(T^{M_n}|Y^{M_n})}{n} =  (1-d+i) \limsup_{m \to \infty} \frac{H_P(T^{m}|Y^{m})}{m}. \label{eq:indel_mn} \ee
The proof of \eqref{eq:indel_mn} is essentially the same as that of Lemma \ref{lem:ins_H_I_Y}, with two changes: $T^{M_n}$ replaces $I^{M_n}$, and $\frac{M_n}{n}$ converges to $(1-d+i)$ for the InDel channel.
We then have
\be
\begin{split}
\frac{1}{m} H_P(T^{m}|Y^{m}) & = \frac{1} {m}  \sum_{j=1}^m H_P(T_j|T^{j-1}, Y^m) \\
& \leq  \frac{1}{m} \sum_{j=1}^m H_P(T_j|T_{j-1}, Y_{j}, Y_{j-1}).
\end{split}
\ee
Therefore
\be
\begin{split}
& \limsup_{m \to \infty} \frac{1}{m}H_P(T^{m}|Y^{m})  \\
& \leq \limsup_{m \to \infty} \frac{1} {m} \sum_{j=1}^m H_P(T_j | T_{j-1}, Y_{j}, Y_{j-1}) \\
& = \lim_{j \to \infty} H_P(T_j | T_{j-1}, Y_{j}, Y_{j-1}),
\end{split}
\label{eq:indel_ty}
\ee
provided the limit exists. From the cascade representation in Fig. \ref{fig:two_ch_decomp}, we see that the insertions are introduced by the second channel in the cascade. The input to this insertion channel is a process $\mathbf{Z}=\{Z_m\}_{m \geq 1}$ which is the output of the first channel in the cascade. From Propositon \ref{prop:y}, $\mathbf{Z}$ is a first-order Markov process with parameter
$q=\tfrac{\gamma+d-2\gamma d}{1+d-2\gamma d}$. We therefore need to calculate $\lim_{j \to \infty} H(T_j|T_{j-1}, Y_{j}, Y_{j-1})$ where $\mathbf{Y}$ is the output when a first-order Markov process with parameter $q$ is transmitted through an insertion channel with parameters $(i',\alpha)$. But we have already computed
$\lim_{j \to \infty} H(T_j|T_{j-1}, Y_{j}, Y_{j-1})$ in Lemma \ref{lem:ins_lim_HT_Y} for an insertion channel with parameters $(i,\alpha)$ with a
first-order Markov input with parameter $\gamma$. Hence, in Lemma \ref{lem:ins_lim_HT_Y} we can replace  $\gamma$ by $q$ and $i$ by $i'$  to obtain
\ben  \begin{split}
& \lim_{j \to \infty} H_P(T_j|T_{j-1}, Y_{j}, Y_{j-1}) \\
& = \frac{(1-q+ q i'\bar{\alpha})}{(1+i')} h\left(\frac{i' \bar{\alpha}}{1-q+ q i' \bar{\alpha}}\right).
\end{split} \een
Substituting $i'=\tfrac{i}{1-d}$ and simplifying yields a lower bound for $\limsup_{m \to \infty} \tfrac{1}{m}H_P(T^{m}|Y^{m})$ from \eqref{eq:indel_ty}. Combining with \eqref{eq:indel_mn} gives the statement of the lemma.
\end{IEEEproof}
\begin{lem} \label{lem:delins_4}
\ben
\begin{split}
 & \limsup_{n \to \infty}  \frac{1}{n}  H_P(S^{M_n +1}| T^{M_n}, Y^{M_n}) \\
 & \leq (1-d+i) \lim_{j \to \infty} H_P(S_j|Y_{j-1}, Y_j, T_j) \\
 & = (1-d)\left( A_1 + A_2 - \frac{\theta \beta}{(1-\theta)^2} \log \theta \right),
\end{split}
\een
where
\ben
\begin{split}
& i' = \frac{i}{1-d}, \quad   q= \frac{\gamma+d-2\gamma d}{1+d-2\gamma d},  \\
& \theta =  \frac{(1-\gamma) d}{1-\gamma d}, \quad  \beta = \frac{(1-\gamma)(1-d)}{(1-\gamma d)^2},
\end{split} \een
{\small{
\ben
\begin{split}
 &A_1 =  \frac{\theta \beta(1-i' \bar{\alpha})}{1-\theta^2} \log\left(\frac{i'\alpha + (1-i'\bar{\alpha})q + i'\bar{\alpha}\bar{q}}{\beta (1-i'\bar{\alpha})}\right) \\
&  +\frac{\theta^2 \beta i' \bar{\alpha}}{1-\theta^2} \log \left(\frac{i'\alpha + (1-i'\bar{\alpha})q + i'\bar{\alpha}\bar{q}}{\beta i'\bar{\alpha}}\right)\\
& +(i'\alpha  + (1-i'\bar{\alpha}) \gamma \bar{\theta} + i' \bar{\alpha} \beta)
\log \left(\frac{i'\alpha + (1-i'\bar{\alpha})q + i'\bar{\alpha}\bar{q}}{i'\alpha + (1-i'\bar{\alpha}) \gamma \bar{\theta} + i' \bar{\alpha} \beta}\right),\\
\end{split}
\een }}
\ben
\begin{split}
 A_2=&  \frac{\theta^2 \beta(1-i' \bar{\alpha})}{1-\theta^2} \log \left(\frac{(1-i'\bar{\alpha})\bar{q} + i'\bar{\alpha}{q}}{\beta (1-i'\bar{\alpha})}\right)\\ &  + \frac{\theta \beta i' \bar{\alpha}}{1-\theta^2} \log \left(\frac{(1-i'\bar{\alpha})\bar{q} + i'\bar{\alpha}{q}}{\beta i'\bar{\alpha}}\right) \\
& + (i' \bar{\alpha} \gamma \bar{\theta} +  (1-i'\bar{\alpha})\beta)
\log \left(\frac{ (1-i'\bar{\alpha})\bar{q} + i'\bar{\alpha}{q}}{i' \bar{\alpha} \gamma \bar{\theta} +  (1-i'\bar{\alpha})\beta}\right).
\end{split}
\een
\end{lem}
\begin{IEEEproof}
See Appendix \ref{proof:delins_lem4}
\end{IEEEproof}
As in Section \ref{subsec:lb2}, we  upper bound  $\frac{1}{n} H(X^n|S^{M_n +1},  T^{M_n}, Y^{M_n})$ in two ways. The first involves flipping the complementary insertions in $Y^{M_n}$ to
obtain $\tilde{Y}^{M_n}$; the second bound in obtained by deleting the complementary insertions to obtain $\hat{Y}^{\hat{M}_n}$. In both cases, the {extra} runs introduced by the channel are removed. Using $S^{M_n+1}$, we can augment $\tilde{Y}^{M_n}$ (resp. $\hat{Y}^{\hat{M}_n}$)  by adding the positions of the {deleted} runs to obtain a sequence $\tilde{Y}'^{M_n}$ (resp. $\hat{Y}'^{\hat{M}_n}$) which contains the same number of runs as $X^n$. $\tilde{Y}'^{M_n}$ can be represented in terms of its run-lengths as $(L^{\tilde{Y}'}_1, \ldots, L^{\tilde{Y}'}_{R_n})$, where we emphasize that $L^{Y'}_1,L^{Y'}_2, \ldots$ can take value $0$ as well.
To summarize, we have
\be
\begin{split}
& X^n  \leftrightarrow (L^X_1, \ldots, L^X_{R_n}),   \
({S}^{M_n +1}, \tilde{Y}^{M_n})  \leftrightarrow (L^{\tilde{Y}'}_1, \ldots, L^{\tilde{Y}'}_{R_n}),\\
& ({S}^{M_n +1}, \hat{Y}^{\hat{M}_n}) \leftrightarrow (L^{\hat{Y}'}_1, \ldots, L^{\hat{Y}'}_{R_n}).
\end{split}
\ee
\begin{prop} \label{prop:delins_lxy'}
1) The process $\{\mathbf{L^X, L^{\tilde{Y}'}}\} \triangleq \{(L^X_1, L^{\tilde{Y}'}_1), (L^X_2, L^{\tilde{Y}'}_2), \ldots \}$ is an i.i.d process characterized by the following joint distribution for all $j \geq 1$:
\be P(L^X_j=r) = \gamma^{r-1} (1-\gamma), \quad r=1,2,\ldots  \ee
 \be
\begin{split}
& P(L^{\tilde{Y}'}_j=s|L^X_j=r) = \\
&  \sum_{n_i\in \mathcal I} \binom{r} {n_i , r+n_i-s, s-2 n_i}  i^{n_i} d^{r+n_i-s} (1-d-i)^{s-2n_i},\\
& \qquad \qquad 0 \leq s \leq 2r
\end{split}
\label{eq:LY'_LX}
\ee
where $\mathcal{I}$, the set of possible values for the number of insertions $n_i$, is given by
\ben
\mathcal{I} =
\{0,\ldots, \lfloor{\frac{s}{2}}\rfloor \}  \text{ for } s \leq r \text{ and }
\{s-r,\ldots, \lfloor{\frac{s}{2}}\rfloor \}  \text{ for } s>r.
\een

2) The process $\{\mathbf{L^X, L^{\hat{Y}'}}\} \triangleq \{(L^X_1, L^{\hat{Y}'}_1), (L^X_2, L^{\hat{Y}'}_2), \ldots \}$ is an i.i.d process who joint distribution is obtained by replacing $i$ in \eqref{eq:LY'_LX} with $i\alpha$.
\end{prop}
\begin{IEEEproof} Since $\mathbf{X}$ is a Markov process, $\{L^X_j\}_{j\geq 1}$ are independent with
 \[ P(L^X_j=r)= \gamma^{r-1} (1-\gamma), \:r=1,2,\ldots\]
 Since there is a one-to-one correspondence between the runs of $\mathbf{X}$ and the runs of $\mathbf{\tilde{Y}'}$, we can think of each $L^{\tilde{Y}'}_j$ being obtained by  passing a run of length $L^X_j$ through a discrete memoryless channel. For a pair $(L^{X}_j=r, L^{\tilde{Y}'}_j=s)$, if the number of insertions is $n_i$, the number of deletions is easily seen to be $r+n_i-s$. Since there can be at most one insertion after each input bit, no more than half the bits in an output run can be insertions; hence the maximum value of $n_i$ is  $\lfloor{\frac{s}{2}}\rfloor$. The minimum value of $n_i$ is zero for $s\leq r$, and $s-r$ for $s>r$. Using these together with the fact that each bit can independently undergo an insertion with probability $i$, a deletion with probability $d$, or no change with probability $1-d-i$, the  transition probability of the memoryless run-length channel is  given by  \eqref{eq:LY'_LX}.

 The proof for the second part is identical except that the effective insertion probability is now $i\alpha$ since the complementary insertions have been removed.
\end{IEEEproof}
\begin{lem}
\be \begin{split}
& \limsup_{n \to \infty} \frac{1}{n} H_P({X}^n| {S}^{M_n +1}, T^{M_n}, Y^{M_n})\\
& \leq \limsup_{n \to \infty} \frac{1}{n} H_P({X}^n| {S}^{M_n +1}, \tilde{Y}^{M_n})  = \bar{\gamma} H_P(L^X|L^{\tilde{Y}'}) \\
& \limsup_{n \to \infty} \frac{1}{n} H_P({X}^n| {S}^{M_n +1}, T^{M_n}, Y^{M_n})\\
& \leq \limsup_{n \to \infty} \frac{1}{n} H_P({X}^n| {S}^{M_n +1}, \hat{Y}^{\hat{M}_n})  = \bar{\gamma} H_P(L^X|L^{\tilde{Y}'})
\end{split}
\ee
where the joint distributions of $(L^X,L^{\tilde{Y}'})$ and $(L^X,L^{\hat{Y}'})$ are given by Proposition \ref{prop:delins_lxy'}.
\label{lem:delins_5}
\end{lem}
\begin{IEEEproof}
The proof is identical to that of Lemma \ref{lem:del_x_sy}.
\end{IEEEproof}
Using the cascade representation, the penalty term in \eqref{eq:delins_ineq} can be bounded as follows.
\begin{lem}
\ben \begin{split} & \liminf_{n \to \infty} \frac{1}{n}  H_P(T^{M_n}, S^{M_n +1} \mid  X^n, Y^{M_n})\\
 & \geq \   (1-d) \Gamma(i', \alpha, q) + \Phi(d, \gamma). \end{split} \een
\label{lem:delins_penalty}
\end{lem}
\begin{IEEEproof}
We have
\be
\begin{split}
& H_P(T^{M_n}, S^{M_n +1}| X^n, Y^{M_n}) \\
& =  H_P(T^{M_n} | X^n, Y^{M_n}) +  H_P(S^{M_n +1}| X^n, Y^{M_n}, T^{M_n}).
\end{split}
\label{eq:HTS_decomp}
\ee
The first term in \eqref{eq:HTS_decomp} can be lower bounded as
\be
\begin{split}
H_P(T^{M_n} | X^n, Y^{M_n}) & \geq H_P(T^{M_n} | X^n, Y^{M_n}, Z^{L_n}) \\
& {=}  H_P(T^{M_n} | Y^{M_n}, Z^{L_n})
\end{split}
\label{eq:HTS_term0}
\ee
The equality above holds due to  the Markov chain
\be
(\mathbf{X}, \mathbf{\Lambda}_{del} ) - \mathbf{Z} - (\mathbf{\Lambda}_{ins}, \mathbf{Y})
\label{eq:casc_mc}
\ee
where $\mathbf{\Lambda}_{del}$ and $\mathbf{\Lambda}_{ins}$ denote the deletion and insertion processes  of the first and second channels in the cascade, respectively\footnote{$\mathbf{\Lambda}_{del}$ is a process where $\Lambda_{del,j}$ indicates if the $j$th input bit was deleted or not. Similarly, $\mathbf{\Lambda}_{ins}$ specifies which input bits to the second channel undergo duplications and which undergo complementary insertions.}. The equality in  \eqref{eq:HTS_term0} is due to the fact that the process $\mathbf{T}$ is a function of $\mathbf{\Lambda}_{ins}$. We then have
\be
\liminf_{n \to \infty} \frac{1}{n} H_P(T^{M_n} | Y^{M_n}, Z^{L_n}) \geq (1-d) \Gamma(i', \alpha, q)
\label{eq:HTS_term1}
\ee
which is obtained as follows. The second channel in the cascade has insertion probability $i'$ and input $Z^{L_n}$ which is first-order Markov with parameter $q$. We then obtain \eqref{eq:HTS_term1} via steps very similar to the proof of  Lemma \ref{lem:lb2_liminf}, accounting for the fact that $\expec[L_n] =n \bar{d}$. The second term in \eqref{eq:HTS_decomp} is bounded  as follows.
\be
\begin{split}
& H_P(S^{M_n +1}| X^n, Y^{M_n}, T^{M_n})\\
 & \stackrel{(a)}{\geq} H_P(S^{M_n +1}| X^n, Y^{M_n}, I^{M_n}) \\
 & \stackrel{(b)}{=}  H_P(\un{S}_Z | X^n, Y^{M_n}, I^{M_n}) \stackrel{(c)}{=} H_P(\un{S}_Z | X^n, Z^{L_n}).
\end{split}
\label{eq:HTS_term2}
\ee
In \eqref{eq:HTS_term2}, $(a)$ holds because $T^{M_n}$ is a function of $(X^n, Y^{M_n}, I^{M_n})$. (Recall $I_j=1$ if $Y_j$ is an inserted bit and $0$ otherwise.) To obtain $(b)$, first note that $Z^{L_n }$ can be determined from $(Y^{M_n}, I^{M_n})$ (by deleting the inserted bits from $Y^{M_n}$). $\un{S}_Z$ is the length $(L_n+1)$ $S$-sequence corresponding to just the first channel in the cascade: for $j=1,\ldots, L_n +1$,  $S_{Z,j}$ is the number of runs completely deleted between bits $Z_{j-1}$ and $Z_j$. In contrast, $S^{M_n}$ is the $S$-sequence for the overall channel: $S_j$ is the number of deleted runs betwen bits $Y_{j-1}$ and $Y_j$. $(b)$ holds because given  $(Y^{M_n}, I^{M_n})$, knowledge of $\un{S}_Z$ is sufficient to reconstruct $S^{M_n}$ and vice versa. To obtain $(c)$, we observe that $\un{S}_Z$ is a function of $(\un{X}, \un{\Lambda}_{del} )$.
Then it follows from the Markov chain  \eqref{eq:casc_mc} that $\un{S}_Z$ is conditionally independent of
$(Y^{M_n}, I^{M_n})$ given $(Z^{L_n}, X^n)$, resulting in $(c)$. Finally, applying Lemma \ref{lem:lb_delpenalty} to the first channel in the cascade we get
\be
\liminf_{n \to \infty} \frac{1}{n} H_P(\un{S}_Z | X^n, Z^{L_n})  \geq \Phi(d, \gamma).
\label{eq:HTS_finalterm}
\ee
Using \eqref{eq:HTS_term1},  \eqref{eq:HTS_term2} and \eqref{eq:HTS_finalterm} in \eqref{eq:HTS_decomp} completes the proof.
\end{IEEEproof}
\begin{figure}
\centering
\includegraphics[width=3.5in]{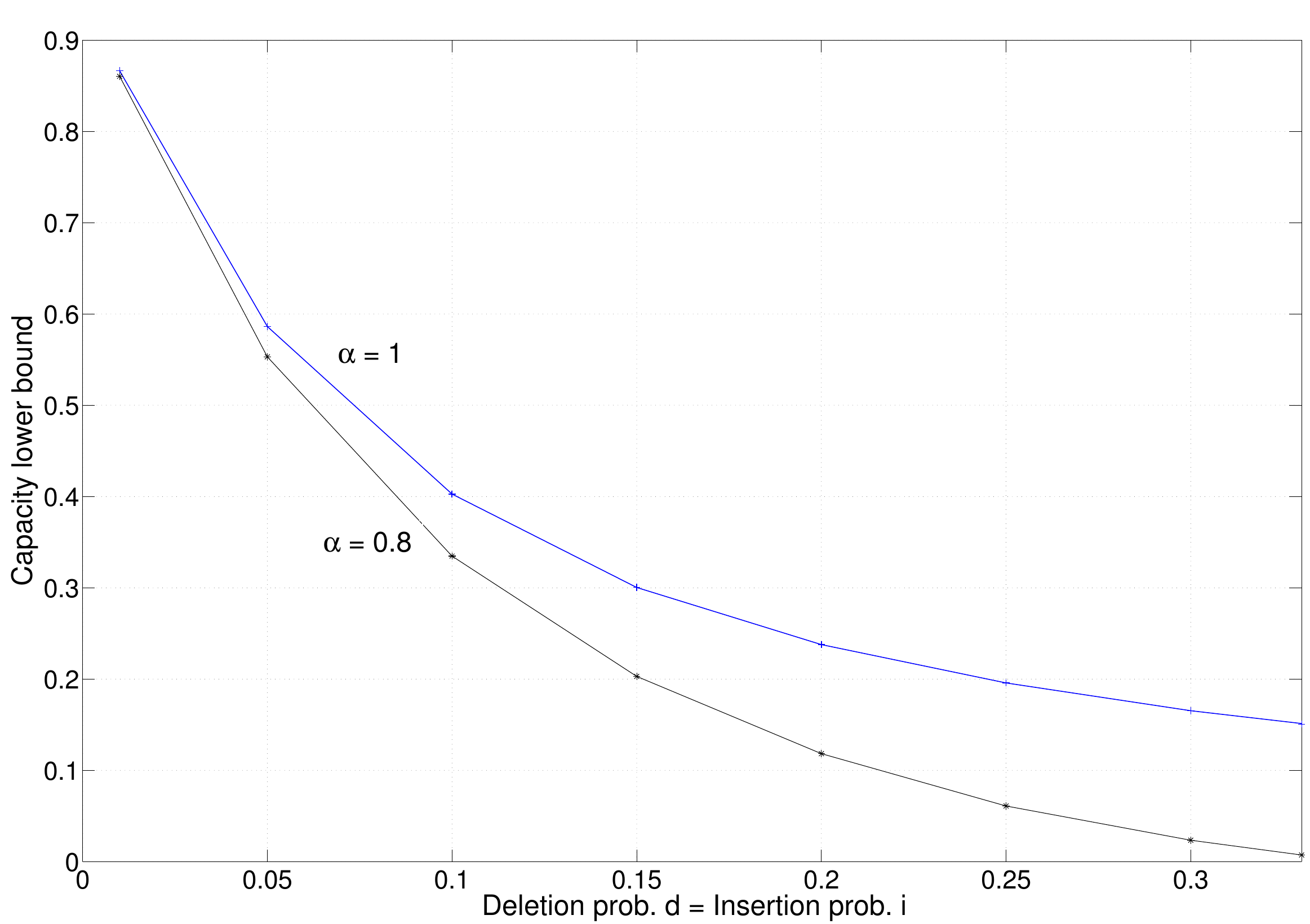}
\caption{Lower bound on the InDel capacity $C(d,i,\alpha)$ for $d=i$.}
\label{fig:delins_fig}
\vspace{-2pt}
\end{figure}
\begin{thm} \label{thm:delins_thm}
The capacity of the InDel channel can be lower bounded as
\ben
\begin{split}
& C(d, i, \alpha) \geq \\
& \max_{0 < \gamma <1}  \Big[ h(\gamma)    -   (\bar{q}(1-d) + q i \bar{\alpha}) h\left(\frac{i \bar{\alpha}}{\bar{q}(1-d) + q i \bar{\alpha}}\right)\\
& - (1-d) \left(A_1 + A_2 - \frac{\theta \beta}{(1-\theta)^2} \log \theta\right)   + (1-d)  \Gamma(i', \alpha, q)  \\
& + \Phi(d, \gamma)  - \bar{\gamma} \min\{ H_P(L^X_{1} | L^{\tilde{Y}'_1}),  H_P(L^X_{1}| L^{\hat{Y}'_1)}  \}  \Big].
 \end{split} \een
where $q, \beta, \theta, A_1, A_2$ are defined in Lemma \ref{lem:delins_4}, and $H_P(L^X_{1}| L^{\tilde{Y}'}_1), H_P(L^X_{1}| L^{\hat{Y}'}_1)$
are computed using the joint distributions given in Proposition \ref{prop:delins_lxy'}.
\end{thm}
\begin{IEEEproof}
The result is obtained by using Lemmas \ref{lem:delins_2}-\ref{lem:delins_penalty} in \eqref{eq:delins_ineq}.
\end{IEEEproof}

The lower bound is plotted in Figure \ref{fig:delins_fig} for various values of $d=i$, for  $\alpha=0.8$ and for $\alpha=1$. Like the deletion channel, the maximizing $\gamma$ was found to increase as $d=i$ was increased; this has the effect of making runs longer and thus  less likely to be deleted completely.

For Theorem \ref{thm:delins_thm}, we used the sequence $T^{M_n}$ to indicate the positions of complementary insertions together with $S^{M_n}$ to indicate deleted runs. We can obtain another lower bound on the InDel capacity by using the sequence $I^{M_n}$ instead of $T^{M_n}$. This bound can be derived in a straightforward  manner by combining the techniques of Sections \ref{subsec:lb1} and \ref{sec:deletion}, and is omitted. As in Section \ref{subsec:lb1}, we expect that a lower bound with $\un{I}$ improve on a bound with $\un{T}$ when $i$ is large. Such a bound would be useful for InDel channels with large $i$ and small $d$.

\section{Discussion} \label{sec:conc}
The approach we have used to obtain capacity lower bounds consists of two parts: analyzing a decoder that synchronizes input and output runs by decoding auxiliary sequences, and bounding the rate penalty due to the sub-optimality of such a decoder. The rate penalty measures the extra information that is decoded compared to maximum-likelihood decoding, and can be lower bounded by identifying patterns  of deletions and insertions that give rise to multiple auxiliary sequences for the same $(\un{X}, \un{Y})$ pair.

We now discuss some directions for future work.

 \emph{Improving the Lower Bounds}:  The lower bound for the penalty terms can be sharpened  by identifying additional  deletion/insertion patterns that lead to different auxiliary sequences for the same pair of input and output sequences.  For the deletion channel, such patterns were briefly described  at the end of Section \ref{subsec:del_penalty}.
There are also a few other ways to sharpen the bounds for the insertion and InDel channels. Creating the sequences $\un{\tilde{Y}}$ and $\un{\hat{Y}}$ using  the positions of complementary insertions synchronises the input and output runs, but is not an optimal way to the sequence $\un{T}$. Is there a better way to use the knowledge of $(\un{T}, \un{Y})$?
Another observation is that the presence of insertions results in an output process that is not Markov, which is the reason an exact expression for the limiting behavior of $\frac{1}{n} H(\un{T}|\un{Y})$ could not be obtained in Lemma \ref{lem:ins_lim_HT_Y}. A better bound for this term would improve the capacity lower bound.

 Another direction is to investigate the performance of more general input distributions with i.i.d runs. In \cite{KanoriaM10}, a specific input distribution with i.i.d runs was shown to be optimal for the deletion channel as $d \to 0$. Further, it was shown that a first-order Markov input distribution incurs a rate loss of roughly $0.1d^2$ with respect to the optimal distribution as $d \to 0$.  A similar result on the structure of the optimal input distribution for the InDel channel with small values of $i$ and $d$  would be very useful, and could be combined with the approach used here to obtain good estimates of the capacity for small insertion and deletion probabilities.

\emph{Upper Bounds}: Obtaining upper bounds for insertion and InDel channels is an important topic  to address. For the InDel channel with $\alpha=1$, the upper bounding techniques of \cite{DMP07} and \cite{Mitz_sticky} can be used. Here an augmented decoder which knows the positions of the deleted runs is considered, and the capacity per unit cost of the resulting synchronized channel is upper bounded. For channels where complementary insertions occur $(\alpha < 1)$, the computational approach of \cite{FertonaniD10, FDE11} can be used to obtain upper bounds. Here it is assumed that the decoder is supplied with markers distinguishing blocks of the output which arose from successive input blocks of length $L$. This augmented channel is then equivalent to a memoryless channel with input alphabet size $2^L$, whose capacity can be evaluated numerically using the Blahut-Arimoto algorithm.

It is interesting to explore how the insertion capacity varies with $\alpha$ for fixed $i$.  Figure $2$ shows that the insertion lower bound is not a monotonic function of $\alpha$, but is the highest for $\alpha=1$. Does the actual  capacity also behave in a similar way? Further, what is the `worst' $\alpha$ for a given $i$?

 \emph{Extension to larger alphabets}: The capacity lower bounds can be extended to channels with larger alphabet. Though some modifications arise, the steps leading up to the result are similar. For an alphabet of size $\abs{\Sigma} >2$, fix a symmetric first-order Markov input distribution such that the probability that $X_{j}=X_{j-1}$ is $\gamma$ and the probability that $X_{j}$ takes any of the other values in the alphabet is $\tfrac{1 - \gamma}{\abs{\Sigma}-1}$. Consider the deletion channel where each symbol is deleted with probability $d$. The output sequence is still first-order Markov. Proposition \ref{prop:sy} gets modified so that the right side of \eqref{eq:del_pbarys_y} is multiplied by  $\tfrac{1}{\abs{\Sigma}-1}$ to account for the multiple possibilities when $Y_{j} \neq Y_{j}$. Lemmas \ref{lem:del_s_y} and \ref{lem:lb_delpenalty} remain unchanged.
The input entropy rate is computed with the new input distribution, and the conditional entropy rate calculation in Lemma  \ref{lem:del_x_sy} now has to account for the fact that when runs are deleted, the symbol corresponding to the run also needs to be indicated. In other words, \eqref{eq:run_corresp} and \eqref{eq:bits_to_runs} are modified  so that $\un{X}$ corresponds to a sequence of (symbol, run-length) pairs, as does $(\un{S}, \un{Y})$.
The insertion  and InDel capacity lower bounds can be similarly extended for larger alphabet channels that are symmetric.

The framework used here can be extended to derive bounds for channels with substitution errors in addition to deletions and insertions. For this, we would need an additional auxiliary sequence, e.g., one that indicates the positions of  substitutions.

The problem of synchronization also appears in file backup and file sharing \cite{MaKTse11,ITASynch11}. In the basic model, we have two terminals,  with the first terminal having source a $\un{X}$ and the second having  $\un{Y}$, which is an edited version of $\un{X}$. The edits may include deletions, insertions, and substitutions.  A basic question is: To update $\un{Y}$ to $\un{X}$, what is the minimum communication rate needed from the first node to the second? It can be shown that regardless of whether the first terminal knows $\un{Y}$ or not, the optimal rate is given by the limiting behavior of $H(\un{X}|\un{Y})$. The results derived in this paper provide bounds on this optimal  rate for the case where $\un{X}$ is Markov, and the edit model $P(\un{Y}|\un{X})$ is one with i.i.d deletions and insertions. Extension of these results to edit models with substitution errors would yield rates to benchmark  the performance of  practical file synchronization tools such as rsync \cite{rsync}.
\begin{figure*}[b]
\vspace*{0pt} \hrulefill
\normalsize
\begin{equation}
\begin{split}
&\frac{1}{n} H_P(I^{M_n}|Y^{M_n}) = \expec\left[ -\frac{1}{n} \log P(I^{M_n}|Y^{M_n})
\cdot \left(\mathbf{1}_{\left\{\frac{M_n}{n} \in (1+i-\e, 1+i+\e)\right\}}  +  \mathbf{1}_{\left\{\frac{M_n}{n} \notin (1+i-\e, 1+i+\e)\right\}}\right) \right]\\
& = \expec\left[  -\frac{1}{n} \log \frac{P(I^{M_n}, Y^{M_n})}{P(Y^{M_n})} \cdot  \mathbf{1}_{\left\{\frac{M_n}{n} \in (1+i-\e, 1+i+\e)\right\}}\right]  +\expec\left[-\frac{1}{n} \log P(I^{M_n}|Y^{M_n}) \cdot  \mathbf{1}_{\left\{\frac{M_n}{n} \notin (1+i-\e, 1+i+\e)\right\}} \right] \\
&\leq \expec\left[  -\frac{1}{n} \log \frac{P(I^{n(1+i+\e)}, Y^{n(1+i+\e)})}{P(Y^{n(1+i-\e)})} \cdot \mathbf{1}_{\left\{\frac{M_n}{n} \in (1+i-\e, 1+i+\e) \right\}}\right]  +\expec\left[-\frac{1}{n} \log P(I^{M_n}|Y^{M_n}) \cdot \mathbf{1}_{\left\{ \frac{M_n}{n} \notin (1+i-\e, 1+i+\e) \right\}} \right] \\
&= \expec\left[  -\frac{1}{n} \log \frac{P(I^{n(1+i+\e)}, Y^{n(1+i+\e)})}{P(Y^{n(1+i-\e)})} \right] - \expec\left[  -\frac{1}{n} \log \frac{P(I^{n(1+i+\e)}, Y^{n(1+i+\e)})}{P(Y^{n(1+i-\e)})} \cdot \mathbf{1}_{\left\{\frac{M_n}{n} \notin (1+i-\e, 1+i+\e) \right\}} \right]\\
& \quad  + \expec\left[-\frac{1}{n} \log P(I^{M_n}|Y^{M_n}) \cdot \mathbf{1}_{\left\{ \frac{M_n}{n} \notin (1+i-\e, 1+i+\e) \right\}} \right].
\end{split}
\label{eq:split_pIY_ub}
\end{equation}
\begin{equation}
\begin{split}
&\frac{1}{n} H_P(I^{M_n}|Y^{M_n}) = \expec\left[-\frac{1}{n} \log \frac{P(I^{M_n}, Y^{M_n})}{P(Y^{M_n})} \cdot  \mathbf{1}_{\left\{\frac{M_n}{n} \in (1+i-\e, 1+i+\e)\right\}}\right]  +\expec \left[-\frac{1}{n} \log P(I^{M_n}|Y^{M_n}) \cdot  \mathbf{1}_{\left\{\frac{M_n}{n} \notin (1+i-\e, 1+i+\e)\right\}} \right] \\
&\geq \expec\left[  -\frac{1}{n} \log \frac{P(I^{n(1+i-\e)}, Y^{n(1+i-\e)})}{P(Y^{n(1+i+\e)})} \cdot \mathbf{1}_{\left\{\frac{M_n}{n} \in (1+i-\e, 1+i+\e) \right\}}\right]  +\expec\left[-\frac{1}{n} \log P(I^{M_n}|Y^{M_n}) \cdot \mathbf{1}_{\left\{ \frac{M_n}{n} \notin (1+i-\e, 1+i+\e) \right\}} \right].\\
&= \expec\left[  -\frac{1}{n} \log \frac{P(I^{n(1+i-\e)}, Y^{n(1+i-\e)})}{P(Y^{n(1+i+\e)})} \right]  - \expec\left[ -\frac{1}{n} \log \frac{P(I^{n(1+i-\e)}, Y^{n(1+i-\e)})}{P(Y^{n(1+i+\e)})} \cdot \mathbf{1}_{\left\{\frac{M_n}{n} \notin (1+i-\e, 1+i+\e) \right\}} \right]  \\
& \quad  + \expec\left[-\frac{1}{n} \log P(I^{M_n}|Y^{M_n}) \cdot \mathbf{1}_{\left\{ \frac{M_n}{n} \notin (1+i-\e, 1+i+\e) \right\}} \right] .
\end{split}
\label{eq:split_pIY_lb}
\end{equation}
\end{figure*}
%

\appendices
\section{Insertion Channel}
\subsection{Proof of Lemma \ref{lem:ins_H_I_Y}} \label{proof:ins_H_I_Y}
We begin by noting that $\frac{M_n}{n} \to (1+i)$ almost surely, due to the strong law of large numbers.

We decompose and upper bound $\tfrac{1}{n} H_P(I^{M_n}|Y^{M_n})$ as shown in \eqref{eq:split_pIY_ub} at the bottom of the page. First examine the third term in the last line of \eqref{eq:split_pIY_ub}. The size of the support of $-\frac{1}{n} \log P(I^{M_n}|Y^{M_n})$ is at most
$2^{2n}$, since $I^{M_n}$
is a binary sequence of length at most $2n$. Hence, from Lemma \ref{lem:support}, $\left\{ -\frac{1}{n} \log P(I^{M_n}|Y^{M_n}) \right\}_{n \geq 1}$
is uniformly integrable. From Lemma \ref{lem:unif_equiv}, for any $\epsilon >0$, there exists some $\delta>0$
\be \label{eq:pIY_small}
\expec\Big[-\tfrac{1}{n} \log P(I^{M_n}|Y^{M_n}) \cdot \mathbf{1}_{\left\{ \frac{M_n}{n} \notin (1+i-\e, 1+i+\e) \right\}} \Big]  < \epsilon
\ee
whenever $P\left(  \frac{M_n}{n} \notin (1+i-\e, 1+i+\e)  \right) < \delta$. Since $\frac{M_n}{n} \to (1+i)$ almost surely,
$P\left( \left\{ \frac{M_n}{n} \notin (1+i-\e, 1+i+\e) \right\} \right)$ is less than $\delta$ for all sufficiently large $n$.
Thus \eqref{eq:pIY_small} is true for all sufficiently large $n$. Similarly, the third term can be shown to be smaller than $\e$ for all sufficiently large $n$.
Therefore, for all sufficiently large $n$, \eqref{eq:split_pIY_ub} becomes
\be
\begin{split}
& \frac{1}{n} H_P(I^{M_n}|Y^{M_n}) \\
& \leq \frac{1}{n}\left(  H_P(I^{n(1+i+\e)}, Y^{n(1+i+\e)})- H_P(Y^{n(1+i-\e)}) \right)  + \e \\
&= (1+i- \e) \frac{H_P(I^{n(1+i-\e)} \mid Y^{n(1+i-\e)} )}{n(1+i-\e)}  + \e  \\
& \  + \frac{1}{n}H_P(I_{n(1+i-\e)+1}^{n(1+i+\e)}, Y_{n(1+i-\e)+1}^{n(1+i+\e)} \mid I^{n(1+i-\e)}, Y^{n(1+i-\e)} ) \\
& \stackrel{(a)}{\leq} (1+i- \e) \frac{H_P(I^{n(1+i-\e)}|  Y^{n(1+i-\e)} )}{n(1+i-\e)} + 4\epsilon + \epsilon.
\end{split}
\ee
where $(a)$ holds because $I_{n(1+i-\e)+1}^{n(1+i+\e)}$ and $Y_{n(1+i-\e)+1}^{n(1+i+\e)}$ can each take on  at most $2^{2n\e}$ different values.
Hence
\ben
\begin{split}
& \limsup_{n \to \infty} \frac{1}{n} H_P(I^{M_n}|Y^{M_n}) \\
& \leq  (1+i+\e)\limsup_{m \to \infty} \frac{1}{m} H_P(I^m|Y^m) + 5\epsilon.
\end{split}
\een
Since $\epsilon>0$ is arbitrary, we let $\epsilon \to 0$ to obtain
\be \label{eq:limsup_ub}
\limsup_{n \to \infty} \frac{1}{n} H_P(I^{M_n}|Y^{M_n}) \leq (1+i) \limsup_{m \to \infty} \frac{1}{m} H_P(I^m|Y^m).
\ee

\begin{figure*}[b]
\vspace*{0pt} \hrulefill
\normalsize
\setcounter{mytempeqncnt}{\value{equation}} 
\setcounter{equation}{77}
\begin{equation}
\begin{split}
\mc{E}_{j, k_1, k_2 }  = & \left\{ (\un{X}, \un{Y}, \un{I} \backslash \un{I}(j)) : (X_{a_j}, \un{X}(j), X_{b_j}) = (\bar{c}, \underbrace{c, c, \ldots,c}_{k_1 \text{ bits}}, \bar{c})  \longrightarrow
(Y_{\tilde{a}_j}, \un{Y}(j), Y_{\tilde{b}_j}) = (\bar{c}, \underbrace{c, c,\ldots,c}_{k_1 + k_2 \text{ bits}}, \bar{c}), \quad c \in\{0,1\} \right\}
\end{split}
\label{eq:Ejk1k2}
\end{equation}
\setcounter{equation}{\value{mytempeqncnt}}
\end{figure*}

Using steps similar to \eqref{eq:split_pIY_ub}, we can decompose and lower bound $\tfrac{1}{n} H_P(I^{M_n}|Y^{M_n})$ as shown in \eqref{eq:split_pIY_lb} at the bottom of the page. Using arguments  identical to the one used for the upper bound, we can show that the last two terms in \eqref{eq:split_pIY_lb} are smaller than $\epsilon$ in absolute value for all sufficiently large $n$, leading to
\be
\label{eq:limsup_lb}
\limsup_{n \to \infty} \frac{1}{n} H_P(I^{M_n}|Y^{M_n}) \geq  (1+i) \limsup_{m \to \infty} \frac{1}{m} H_P(I^m|Y^m).
\ee
Combining \eqref{eq:limsup_ub} and \eqref{eq:limsup_lb} completes the proof.
\qedfilled
\subsection{Proof of Lemma \ref{lem:lim_HI_Y}} \label{proof:lim_HI_Y}
 We have
\be
\begin{split}
\frac{1}{m} H_P(I^{m}|Y^{m}) &  = \frac{1}{m} \sum_{j=1}^m H_P(I_j|I^{j-1}, Y^m) \\
&  \leq  \frac{1}{m} \sum_{j=1}^m  H_P(I_j|I_{j-1}, Y_{j}, Y_{j-1}, Y_{j-2})
\end{split}
\ee
where the inequality holds because conditioning cannot increase entropy. Therefore
\be
\begin{split}
& \limsup_{m \to \infty} \frac{1}{m} H_P(I^{m}|Y^{m}) \\
& \leq \limsup_{m \to \infty} \frac{1}{m} \sum_{j=1}^m H_P(I_j|I_{j-1}, Y_{j}, Y_{j-1}, Y_{j-2}).
\end{split}
\label{eq:iylim}
\ee

From Proposition \ref{prop:ins_sy}, the process $\{\mathbf{I}, \mathbf{Y}\}$ is characterized by a Markov chain with state at time $j$ given by $(I_{j}, Y_{j}, Y_{j-1})$.  The conditional distribution $P(I_j, Y_j | I_{j-1}, Y_{j-2}, Y_{j-1})$ is given by \eqref{eq:IY_dist} and the stationary joint distribution $\pi(I_{j}, Y_{j}, Y_{j-1})$  is given by \eqref{eq:stat_pi}.  The right side of \eqref{eq:iylim} can then be computed from the entropy rate of this Markov chain as
\be \begin{split}
&\limsup_{m \to \infty} \frac{1}{m} \sum_{j=1}^m H_P(I_j|I_{j-1}, Y_{j}, Y_{j-1}, Y_{j-2})\\
& = \lim_{j \to \infty} H_\pi(I_j|I_{j-1}, Y_{j}, Y_{j-1}, Y_{j-2}).
\end{split}
\ee
where the subscript $\pi$ refers to the stationary joint distribution on $(Y_{j-2}, Y_{j-1}, I_{j-1}, I_{j}, Y_j)$, given by \eqref{eq:stat_pi} and \eqref{eq:IY_dist}.

$H_\pi(I_j|I_{j-1}, Y_{j}, Y_{j-1}, Y_{j-2})$ can be computed as follows. First, we note that
$I_j=0$ whenever $I_{j-1}=1$. Therefore
\be \label{eq:ins_HIj_cond}
\begin{split}
& H_\pi(I_j|I_{j-1}, Y_{j}, Y_{j-1}, Y_{j-2}) = \\
& \sum_{y=0}^1 \Big\{ H(I_j | (I_{j-1}, Y_{j}, Y_{j-1}, Y_{j-2})=(0,y,y,y)) \  \pi(0,y,y,y) \\
&+  H(I_j | (I_{j-1}, Y_{j}, Y_{j-1}, Y_{j-2}) = (0,y,y,\bar{y})) \  \pi(  0,y,y,\bar{y} ) \\
&+  H(I_j | (I_{j-1}, Y_{j}, Y_{j-1}, Y_{j-2}) = (0,y, \bar{y}, \bar{y})) \  \pi(0,y, \bar{y}, \bar{y}) \\
&+  H(I_j|  (I_{j-1}, Y_{j}, Y_{j-1}, Y_{j-2})= (0,y,\bar{y}, y)) \  \pi( 0,y,\bar{y}, y) \Big\}.
\end{split}
\ee
From \eqref{eq:stat_pi} and \eqref{eq:IY_dist}, we have
\be
\begin{split}
& \pi((I_{j-1}, Y_{j}, Y_{j-1}, Y_{j-2})=(0,y,y,y))  \\
&  = \pi( (I_{j-1}, Y_{j-1}, Y_{j-2})=(0,y,y) ) \cdot \\
& \quad \Big[ P(Y_{j}=y, I_j=1| (I_{j-1}, Y_{j-1}, Y_{j-2})=(0,y,y) ) \\
& \quad + P(Y_{j}=y, I_j=0| (I_{j-1}, Y_{j-1}, Y_{j-2})=(0,y,y)) \Big]\\
&= \frac{\bar{i}\gamma+ i \alpha \gamma +  i \bar{\alpha} \bar{\gamma}}{2(1+i)} \cdot (i\alpha + \bar{i} \gamma),
\end{split}
\ee
and
\be H(I_j|I_{j-1}=0, Y_{j}=y, Y_{j-1}=y, Y_{j-2}=y) = h\left( \frac{i\alpha}{i\alpha + \bar{i}\gamma}\right)\ee
The remaining terms in \eqref{eq:ins_HIj_cond} can be similarly calculated to obtain \eqref{eq:hI_limit}.
\qedfilled
\subsection{Proof of Lemma \ref{lem:lb1_liminf}} \label{app:lb1_liminf_proof}
In the following, all entropies are with respect to the distribution $P$, so we drop the subscript for brevity.
Let $\un{Y}(j)$ denote the $j$th run of $\un{Y}$, and   $\un{I}(j)$  denote the sequence of $I$-bits corresponding to the $j$th run of $\un{Y}$. (Note that $\un{I}(j)$ is distinct from the $j$th run of $\un{I}$.) Let    $R(\un{a})$ denote the number of runs in sequence $\un{a}$.
Using this notation, we can write $\un{I}  =  \ \un{I}(1), \ldots,\un{I}(R(\un{Y}))$. We have
{\small{
\be
\begin{split}
& H(I^{M_n} \mid Y^{M_n}, X^n)\\
& = \sum_{ \un{x}, \un{y} }P(\un{X} = \un{x}, \un{Y} = \un{y}) H( \un{I} \mid \un{X}=\un{x}, \un{Y}=\un{y} ) \\
& = \sum_{ \un{x}, \un{y} } P(\un{X} = \un{x},  \un{Y} = \un{y})  \\
&\quad \cdot \sum_{j=1}^{R(\un{y})} H(\un{I}(j) \mid \un{I}(1), \ldots, \un{I}(j-1), \un{X}=\un{x}, \un{Y}=\un{y} )  \\
& \geq \sum_{ \un{x}, \un{y}} P( \un{X} = \un{x}, \un{Y} = \un{y}) \sum_{j=1}^{R(\un{y})} H(\un{I}(j) | \un{I}(1) \backslash \un{I}(j), \ \un{X}=\un{x}, \un{Y}=\un{y} ).
\end{split}
\label{eq:i_runs}
\ee}}

\begin{figure*}[b]
\vspace*{0pt} \hrulefill
\normalsize
\setcounter{mytempeqncnt}{\value{equation}} 
\setcounter{equation}{81}
\be
\begin{split}
&  \sum_{\un{x}, \un{y}} \sum_{j=1}^{R(\un{y})} \sum_{\stackrel{\un{i} : (\un{x}, \un{y}, \un{i})  \in}{\mc{E}_{j, k_1, k_2}} }
 P( (\un{X}, \un{Y},  \un{I} \backslash \un{I}(j)) = (\un{x}, \un{y}, \un{i})) \\
 & = \sum_{\un{x}}  P( \un{X} = \un{x}) \sum_{l=1}^{R(\un{x})}  \mathbf{1}( l\text{th run of $\un{x}$ has length }  k_1)   \ P( l\text{th run  of $\un{x} \ \longrightarrow  \ $ run of length $k_1+k_2$ in } \un{Y} ) \\
&=  \expec \Big[ \   \expec\Big[ \sum_{l=1}^{R(\un{X})}  \mathbf{1}( l\text{th run of $\un{X}$ has length }  k_1)   \
 \mathbf{1}( l\text{th run  of $\un{x} \ \longrightarrow  \ $ run of length $k_1+k_2$ in } \un{Y} )  \mid  R(\un{X}) \Big]  \Big]\\
&\stackrel{(a)}{=}
\expec \Big[  \sum_{l=1}^{R(\un{X})} \gamma^{k_1-1}\bar{\gamma} \cdot  \left( \binom{k_1}{k_2 -1}(i\bar{\alpha}) (i \alpha)^{k_2-1} (1-i)^{k_1 - k_2 +1}  + \binom{k_1}{k_2} (1-i)(i \alpha)^{k_2} (1-i)^{k_1 - k_2} \right)   \Big] \\
&\stackrel{(b)}{=}
n\bar{\gamma} \cdot \gamma^{k_1-1}\bar{\gamma} \cdot
\left( \binom{k_1}{k_2 -1}(i\bar{\alpha}) (i \alpha)^{k_2-1} (1-i)^{k_1 - k_2 +1}  + \binom{k_1}{k_2} (1-i)(i \alpha)^{k_2} (1-i)^{k_1 - k_2} \right).
\end{split}
\label{eq:RY_RX_lb1}
\ee
\setcounter{equation}{\value{mytempeqncnt}}
\end{figure*}

The inequality in the last line of \eqref{eq:i_runs} is obtained by conditioning on additional random variables: $\un{I} \backslash \un{I}(j)$ denotes all the bits in $\un{I}$ except $\un{I}(j)$. Given  $ \un{I} \backslash \un{I}(j) $,  we know exactly which bit in $\un{X}$ corresponds to each bit in
\[ (\un{Y}(1),\ldots, \un{Y}(j-1), \un{Y}(j+1) , \ldots, \un{Y}(R(\un{Y})).\]  Therefore the set of bits in $\un{X}$ that correspond to  the run
$\un{Y}(j)$ is also known. This set of bits is denoted by $\un{X}(j)$.\footnote{Note that $\un{X}(j)$ is not the $j$th run of $\un{X}$. In fact, $\un{X}(j)$ may not even be a full run of $\un{X}$.} To summarize, given $(\un{X}, \un{Y}, \un{I} \backslash \un{I}(j) )$,  the pair  $(\un{X}(j), \un{Y}(j))$  is determined and any remaining uncertainty may only be about which bits in $\un{X}(j)$ underwent insertions to yield $\un{Y}(j)$.

To obtain an analytical lower bound, we  only consider terms in  \eqref{eq:i_runs} for which  $(\un{X}(j), \un{Y}(j), \un{I} \backslash \un{I}(j))$  has a particular structure. Motivated by the discussion in Section \ref{subsec:lb1} (see \eqref{eq:ins_penalty_example}), we consider terms for which $\un{X}(j)$ is an $\un{X}$-run of length $k_1$ and $\un{Y}(j)$ is a $Y$-run of length $k_1 +k_2$ for some $k_1 \geq 1$ and $1 \leq k_2 \leq k_1$.
Let $X_{a_j}$ and $X_{b_j}$ be the $\un{X}$-bits just before and after $\un{X}(j)$, respectively; similarly denote by $Y_{\tilde{a}_j}$ and $Y_{\tilde{b}_j}$ the $\un{Y}$-bits just before and after the run $\un{Y}(j)$. Define  the set $\mc{E}_{j, k_1, k_2 }$ as in \eqref{eq:Ejk1k2} at the bottom of this page.
\addtocounter{equation}{1}
The arrow in the definition in \eqref{eq:Ejk1k2}  means that the  input bits $(X_{a_j}, \un{X}(j), X_{b_j})$ give rise to $(Y_{\tilde{a}_j}, \un{Y}(j), Y_{\tilde{b}_j})$ through some pattern of insertions. We lower bound \eqref{eq:i_runs} by considering only $(\un{X}, \un{Y}, \un{I} \backslash \un{I}(j))$ that belong to $\mc{E}_{j, k_1, k_2}$:
\be
\begin{split}
 & H(I^{M_n} \mid Y^{M_n}, X^n)   \geq \\
&  \  \sum_{k_1  \geq 1} \sum_{k_2 =1}^{k_1} \ \sum_{\un{x}, \un{y}} \sum_{j=1}^{R(\un{y})} \sum_{\stackrel{\un{i} : (\un{x}, \un{y}, \un{i})  \in}{\mc{E}_{j, k_1, k_2}} }
 P( (\un{X}, \un{Y},  \un{I} \backslash \un{I}(j)) = (\un{x}, \un{y}, \un{i})  ) \\
 & \  \cdot  H(\un{I}(j) \mid (\un{X}, \un{Y},  \un{I} \backslash \un{I}(j)) = (\un{x}, \un{y}, \un{i}) ).
\end{split}
\label{eq:HI_YX_lb1}
\ee
Given $\mc{E}_{j, k_1, k_2 }$, the following are the different choices for insertion pattern $\un{I}(j)$:
\begin{enumerate}
\item The bit $X_{a_j} = \bar{c}$ undergoes a complementary insertion leading to the first $c$ in the $\un{Y}(j)$. Then $(k_2-1)$ out of the $k_1 \  c$'s in $\un{X}(j)$ undergo duplications, the remaining $c$'s are transmitted without any insertions. There are $\tbinom{k_1}{k_2 -1}$ such insertion patterns, each of which occurs with probability
$(i \bar{\alpha}) (i \alpha)^{k_2-1} (1-i)^{k_1 - k_2 +1}$.

\item  $X_{a_j} = \bar{c}$ is transmitted without any insertions.  $k_2$ out of the $k_1 \ c$'s in the $\un{X}(j)$ undergo duplications, the remaining are transmitted without insertions. There are $\tbinom{k_1}{k_2}$ such insertion patterns; each of which occurs with probability $(1-i)(i \alpha)^{k_2} (1-i)^{k_1 - k_2}$.
\end{enumerate}
We thus compute for all $(\un{x}, \un{y}, \un{i}) \in \mc{E}_{j, k_1, k_2 } $
\be
\begin{split}
& H( \un{I}(j) \mid (\un{X}, \un{Y},  \un{I} \backslash \un{I}(j)) = (\un{x}, \un{y}, \un{i})  ) \\
& = \binom{k_1}{k_2-1} \frac{\bar{\alpha}}{\kappa} \log_2 \frac{\kappa}{\bar{\alpha}}
+ \dbinom{k_1}{k_2} \frac{\alpha}{\kappa} \log_2 \frac{\kappa}{\alpha}
\end{split}
\label{eq:HI_YXE}
\ee
with $ \kappa \triangleq \tbinom{k_1}{k_2-1} \bar{\alpha} + \tbinom{k_1}{k_2} \alpha$. Substituting \eqref{eq:HI_YXE} in
\eqref{eq:HI_YX_lb1}, we obtain
\be
\begin{split}
& H(I^{M_n} \mid Y^{M_n}, {X}^n)  \geq \\
& \sum_{k_1 \geq 1} \sum_{k_2 =1}^{k_1}   \left( \binom{k_1}{k_2-1} \frac{\bar{\alpha}}{\kappa} \log_2 \frac{\kappa}{\bar{\alpha}}
+ \dbinom{k_1}{k_2} \frac{\alpha}{\kappa} \log_2 \frac{\kappa}{\alpha}\right) \\
& \  \cdot  \Big[ \sum_{\un{x}, \un{y}} \sum_{j=1}^{R(\un{y})} \sum_{\stackrel{\un{i} : (\un{x}, \un{y}, \un{i})  \in}{\mc{E}_{j, k_1, k_2}} }
 P( (\un{X}, \un{Y},  \un{I} \backslash \un{I}(j)) = (\un{x}, \un{y}, \un{i})) \Big].
 \end{split}
\label{eq:HIYX_simp}
\ee
The terms in square brackets is simply the expected number of $Y$-runs $\un{Y}(j)$ for which the event in \eqref{eq:Ejk1k2} occurs. By definition, the event occurs only when the run $\un{X}(j)$ is a \emph{full} $\un{X}$-run that yields $\un{Y}({j})$ without any complementary insertions in $\un{X}(j)$. (If there is a complementary insertion in a full run of $\un{X}$,  it will create a  $\un{Y}(j)$ for which $\un{X}(j)$ is only part of the full $\un{X}$-run.)

Hence in \eqref{eq:HIYX_simp}, the term in brackets can be written as shown in \eqref{eq:RY_RX_lb1} at the bottom of this page.
\addtocounter{equation}{1}
In \eqref{eq:RY_RX_lb1}, each term in $(a)$ is the probability of an $\un{X}$-run having length $k_1$ bits \footnote{Conditioned on $R(\un{X})$, the run-lengths of $\un{X}$ are not strictly independent as they have to sum to $n$. This can be handled by observing that $R(\un{X})/n$ concentrates around its mean $\bar{\gamma}$ and taking the  $k_1$ sum only over values much smaller than  $n$, e.g., $1 \leq k_1 \leq n^{0.5 + \epsilon}$. We then show that conditioned on $R(\un{X})$, the probability that a run of $\un{X}$ has length $k_1$ is very close to $\gamma^{k-1} \bar{\gamma}$. } multiplied by the probability of it generating a $\un{Y}$-run of length $k_2$ (given by points 1 and 2 above). $(b)$ is obtained by recognising that the expected number of runs in $\un{X}$ is $n \bar{\gamma}$.
Substituting \eqref{eq:RY_RX_lb1} in  \eqref{eq:HIYX_simp} and dividing throughout by $n$ yields  the lemma.
\qedfilled

\subsection{Proof of Lemma \ref{lem:ins_lim_HT_Y}} \label{proof:ins_lim_HT_Y}
We have
\be
\begin{split}
\frac{1}{m} H_P(T^{m}|Y^{m}) &= \frac{1}{m} \sum_{j=1}^m H_P(T_j|T^{j-1}, Y^m)\\
&  \leq  \frac{1}{m} \sum_{j=1}^m H_P(T_j|T_{j-1}, Y_{j}, Y_{j-1}).
\end{split}
\ee
Therefore
\be
\begin{split}
& \limsup_{m \to \infty} \frac{1}{m} H_P(T^{m}|Y^{m}) \\
& \leq \limsup_{m \to \infty} \frac{1}{m} \sum_{j=1}^m H_P(T_j|T_{j-1}, Y_{j}, Y_{j-1})
\end{split}
\label{eq:TmYm}
\ee
We now show that the $\limsup$ on the right side of  \eqref{eq:TmYm} is a limit and is given by \eqref{eq:hT_limit}.

Note that $T_j=0$ whenever $T_{j-1}=1$ since we cannot have two consecutive insertions. Also, $T_j=0$ whenever $Y_j=Y_{j-1}$ since $T_j=1$
only when $Y_j$ is a complementary insertion. Thus we have for all $j \geq 2$:
\be
\begin{split}
 H(T_j|T_{j-1}, Y_{j}, Y_{j-1})
&=  P(T_{j-1}=0, Y_{j}=1, Y_{j-1}=0)\\
& \quad \cdot H(T_j|T_{j-1}=0, Y_{j}=1, Y_{j-1}=0)\\
&+ P(T_{j-1}=0, Y_{j}=0, Y_{j-1}=1) \\
&\quad \cdot H(T_j|T_{j-1}=0, Y_{j}=0, Y_{j-1}=1).
\end{split}
\label{eq:HTjTj1YjYj1}
\ee
The remainder of the proof consists of showing that the quantities in \eqref{eq:HTjTj1YjYj1} all converge to well-defined limits which can be easily computed.

For all $j\geq 1$, $P(T_j=1) = P(I_j=1)\bar{\alpha}$, where $I_j=1$ if  $Y_j$ is an inserted bit, and $I_j=0$ otherwise.
Therefore,
\be \label{eq:IjTj} P(T_{j}=0) =1- P(I_j=1)\bar{\alpha}, \quad  j\geq 1.\ee
The binary-valued process $\{I_j\}_{j \geq 1}$ is a Markov chain with transition probabilities
\be
\begin{split}
& \text{Pr}(I_j=1|I_j=0)=1-\text{Pr}(I_j=0|I_j=0)=i, \\
& \text{Pr}(I_j=1|I_j=1)=1-\text{Pr}(I_j=0|I_j=1)=0.
\end{split}
\ee
For $i \in (0,1)$, this is an irreducible, aperiodic Markov chain with a stationary distribution $\pi$ given by
\be \pi(I_j=1)=1-\pi(I_j=0)=\frac{i}{1+i}. \ee
Hence for any $\e >0$,
\be
\left|P(I_j=1) - \frac{i}{1+i}\right|<\epsilon \text{ and } \left|\text{Pr}(I_j=0) - \frac{1}{1+i}\right|<\epsilon
\label{eq:i_stat}
\ee
 for all sufficiently large $j$.  Using this in \eqref{eq:IjTj}, for all sufficiently large $j$, the distribution $P(T_{j})$
is within total variation norm $\e$ of the following stationary distribution.
 \be
 \pi(T_{j}=0) = 1-\frac{i\bar{\alpha}}{1+i} = \frac{1+i\alpha}{1+i}, \quad
 \pi(T_{j}=1) = \frac{i\bar{\alpha}}{1+i}.
 \label{eq:Tj_asymp}
 \ee
  Further, we have $P(Y_{j}=1|T_{j}=0)= P(Y_{j}=0 |T_{j}=0)=0.5$ since the input distribution and the  insertion process are both symmetric in $0$ and $1$. Hence the stationary distribution for $(T_{j-1}, Y_{j-1})$ is
 \be
 \begin{split}
 & \pi(T_{j-1}=0, Y_{j-1}=y) =   \frac{1+i\alpha}{2(1+i)}, \\
 & \pi(T_{j-1}=1, Y_{j-1}=y) = \frac{i\bar{\alpha}}{2(1+i)}, \quad y \in \{0,1\}.
 \end{split}
 \label{eq:Tj1Yj1_asymp}
 \ee
 Next, we determine the conditional distribution $P(Y_j, T_j|Y_{j-1}=y, T_{j-1}=0)$.

 For $y \in \{0, 1\}$,
\be \label{eq:tjyj_cond1}
\begin{split}
&P(T_j=0, Y_j=y \mid Y_{j-1}=y, T_{j-1}=0)\\
&= P(T_j=0, Y_j=y, I_{j-1}=1 \mid Y_{j-1}=y, T_{j-1}=0) \\
& \quad  +  P(T_j=0, Y_j=y, I_{j-1}=0 \mid Y_{j-1}=y, T_{j-1}=0) \\
&= P(I_{j-1}=1| T_{j-1}=0) \\
& \qquad \cdot P(T_{j}=0, Y_j=y \mid I_{j-1}=1, T_{j-1}=0, Y_{j-1}=y) \\
& \quad + P(I_{j-1}=0| T_{j-1}=0) \\
& \qquad  \cdot P(T_{j}=0, Y_j=y \mid I_{j-1}=0, T_{j-1}=0, Y_{j-1}=y)\\
&\stackrel{(a)}{=} \frac{P(I_{j-1}=1) P(T_{j-1}=0 \mid I_{j-1}=1)} {P(T_{j-1}=0)} \gamma \\
& \quad + \frac{P(I_{j-1}=0) P(T_{j-1}=0 \mid I_{j-1}=0)} {P(T_{j-1}=0)} (\bar{i}\gamma + i \alpha)\\
&\stackrel{(b)}{=} \frac{P(I_{j-1}=1) \alpha}{1- \bar{\alpha}P(I_{j-1}=1)} \gamma
+ \frac{P(I_{j-1}=0)}{1- \bar{\alpha}P(I_{j-1}=1)} (\bar{i}\gamma + i \alpha)
\end{split}
\ee
In the chain above, $(b)$ is obtained using \eqref{eq:IjTj}. $(a)$ is obtained as follows. The event $(I_{j-1}=1, T_{j-1}=0, Y_{j-1}=y)$ implies
 $Y_{j-1}$ is a duplication, and hence $Y_{j-2}=y$ corresponds to an input bit (say $X_a$),
and $Y_{j}$ is the next input bit $X_{a+1}$. The probability that $X_{a+1}=X_a$ is $\gamma$.
 Hence $P(T_{j}=0, Y_j=y \mid I_{j-1}=1, T_{j-1}=0, Y_{j-1}=y)=\gamma$.
When $(I_{j-1}=0, T_{j-1}=0, Y_{j-1}=y)$, $Y_{j-1}$ corresponds to an input bit, say $X_b$. Conditioned on this, the event $(T_{j}=0, Y_j=y)$ can occur in two ways:
\begin{itemize}
\item $Y_j$ is the next input bit $X_{b+1}$ and is equal to $y$. This event has probability $(1-i)\gamma$.
\item $Y_j$ is a duplication of $Y_{j-1}$. This event has probability $i \alpha$.
\end{itemize}
Hence $P(T_{j}=0, Y_j=y \mid I_{j-1}=0, T_{j-1}=0, Y_{j-1}=y)= (\bar{i} \gamma + i \alpha)$.  We similarly calculate
\be \label{eq:tjyj_cond2}
\begin{split}
& P(T_j=0, Y_j=\bar{y} \mid Y_{j-1}=y, T_{j-1}=0) \\
& {=} \frac{P(I_{j-1}=1)}{1- \bar{\alpha}P(I_{j-1}=1)} \alpha \bar{\gamma}
+ \frac{P(I_{j-1}=0)}{1- \bar{\alpha}P(I_{j-1}=1)} (1-i)\bar{\gamma},
\end{split}
\ee
\be \label{eq:tjyj_cond3}
\begin{split}
& P(T_j=1, Y_j=\bar{y}|Y_{j-1}=y, T_{j-1}=0) \\
& {=}  \frac{P(I_{j-1}=0)}{1- \bar{\alpha}P(I_{j-1}=1)} i\bar{\alpha},
\end{split}
\ee
and
\be \label{eq:tjyj_cond4}
P(T_j=1, Y_j=y|Y_{j-1}=y, T_{j-1}=0)=0.
\ee
Using \eqref{eq:i_stat} in equations \eqref{eq:tjyj_cond1}-\eqref{eq:tjyj_cond4}, we see that for all sufficiently large $j$,
the distribution $P(T_j, Y_j|Y_{j-1}=y, T_{j-1}=0)$ is within a total variation norm $\e$ from the following stationary distribution
{\small{
\be
\begin{split}
& \pi(T_j=0, Y_j=y|Y_{j-1}=y, T_{j-1}=0) =\frac{i\alpha(1+ \gamma) + (1-i)\gamma}{1+i\alpha} , \\
&\pi(T_j=0, Y_j=\bar{y}|Y_{j-1}=y, T_{j-1}=0) = \frac{\bar{\gamma} (1-i \bar{\alpha})}{1+i\alpha},\\
& \pi(T_j=1, Y_j=\bar{y}|Y_{j-1}=y, T_{j-1}=0) = \frac{i\bar{\alpha}}{1+i\alpha},\\
& \pi(T_j=1, Y_j=y|Y_{j-1}=y, T_{j-1}=0) =0.
\end{split}
\label{eq:cond_tjyj_asymp}
\ee}}
We conclude that the distribution $P(T_{j-1}, Y_{j-1}, Y_j, T_j)$ converges to
$\pi(T_{j-1}, Y_{j-1}, Y_j, T_j)$ as $j \to \infty$, where the joint distribution $\pi$ is obtained by combining  \eqref{eq:Tj1Yj1_asymp} and \eqref{eq:cond_tjyj_asymp} .
Due to the continuity of the entropy function in the joint distribution, we also have
\[ \lim_{j \to \infty} H_P(T_j|T_{j-1}, Y_{j-1}, Y_j) =  H_\pi(T_j|T_{j-1}, Y_{j-1}, Y_j),\]
We can use these facts in \eqref{eq:HTjTj1YjYj1} and compute $H_\pi(T_j|T_{j-1}, Y_{j-1}, Y_j)$ to obtain the lemma.
\qedfilled
\subsection{Proof of Lemma \ref{lem:ins_x_ytil}}\label{proof:ins_x_ytil}
Due to \eqref{eq:ins_bits_to_runs}, it is enough to show that $\frac{1}{n} H_P(L^X_1, \ldots, L^X_{R_n}|L^{\tilde{Y}}_1, \ldots, L^{\tilde{Y}}_{R_n})$
converges to $ \bar{\gamma} H_P(L^X_1|L^{\tilde{Y}}_1)$.
Since $\{(L^X_1, L^{\tilde{Y}}_1), (L^X_2, L^{\tilde{Y}}_2), \ldots \}$ is an i.i.d process, from the strong law of large numbers, we have
\be
\begin{split}
&\lim_{m \to \infty} -\frac{1}{m}  \log \text{Pr}(L^X_1, \ldots, L^X_{m} | L^{\tilde{Y}}_1, \ldots, L^{\tilde{Y}}_{m}) \\
&= H_P(L^X_1|L^{\tilde{Y}}_1) \quad a.s.
\end{split}
\ee
Further,  the normalized number of input runs $\frac{R_n}{n} \to \bar{\gamma}$ almost surely. Using the above in Lemma \ref{lem:entropyrate},
we obtain
\be
\begin{split}
& \lim_{n \to \infty} -\frac{1}{n}  \log \text{Pr}(L^X_1, \ldots, L^X_{R_n} | L^{\tilde{Y}}_1, \ldots, L^{\tilde{Y}}_{R_n}) \\
& = \bar{\gamma} H_P(L^X_1|L^{\tilde{Y}}_1) \quad a.s.
\end{split}
\label{eq:ins_run_conv}
\ee
We now argue that $-\frac{1}{n} \log \text{Pr}(L^X_1, \ldots, L^X_{R_n} | L^{\tilde{Y}}_1, \ldots, L^{\tilde{Y}}_{R_n})$ is uniformly integrable.
Supp$(L^X_1, \ldots, L^X_{R_n} | L^{\tilde{Y}}_1, \ldots, L^{\tilde{Y}}_{R_n})$ can be upper bounded  by  $2^{n}$ since the random sequence
$(L^X_1, \ldots, L^X_{R_n})$ is equivalent to $X^n$, which can take on at most $2^n$ values. Hence, from Lemma \ref{lem:support}, $-\frac{1}{n} \log \text{Pr}(L^X_1, \ldots, L^X_{R_n} | L^{\tilde{Y}}_1, \ldots, L^{\tilde{Y}}_{R_n})$ is uniformly integrable.
Using this together with \eqref{eq:ins_run_conv} in Lemma \ref{lem:exchange_lim}, we conclude that
\be
\begin{split}
& \lim_{n \to \infty}\frac{1}{n} H_P(L^X_1, \ldots, L^X_{R_n} | L^{\tilde{Y}}_1, \ldots, L^{\tilde{Y}}_{R_n})\\
&= \lim_{n \to \infty} \expec\left[  -\frac{1}{n} \log \text{Pr}(L^X_1, \ldots, L^X_{R_n} | L^{\tilde{Y}}_1, \ldots, L^{\tilde{Y}}_{R_n}) \right]\\
&=\expec\left[ \lim_{n \to \infty} -\frac{1}{n} \log \text{Pr}(L^X_1, \ldots, L^X_{R_n} | L^{Y'}_1, \ldots, L^{\tilde{Y}}_{R_n}) \right]\\
&= \bar{\gamma} H_P(L^X_1|L^{\tilde{Y}}_1).
\end{split}
\ee
The proof that $\frac{1}{n} H_P({X}^n| \hat{Y}^{\hat{M}_n})$ converges to $ \bar{\gamma} H_P(L^X_1|L^{\hat{Y}}_1)$ is essentially identical.
\qedfilled

\subsection{Proof of Lemma \ref{lem:lb2_liminf}} \label{proof:lb2_liminf}
The proof is similar to that of Lemma \ref{lem:lb1_liminf} in Appendix \ref{app:lb1_liminf_proof}. As before, let $\un{Y}(j)$ denote the $j$th run of $\un{Y}$ and   $\un{T}(j)$ be the sequence of $T$'s corresponding to the $j$th run of $\un{Y}$. We can expand  $H(T^{M_n}|Y^{M_n}, X^n)$ as
{\small{
\be
\begin{split}
 & H(T^{M_n}|Y^{M_n}, X^n) \\
 & = \sum_{ \un{x}, \un{y}} P(\un{X} = \un{x},  \un{Y} = \un{y}) \\
 & \quad \cdot  \sum_{j=1}^{R(\un{y})}  H(\un{T}(j) \mid \un{T}(1), \ldots, \un{T}(j-1), \un{X}=\un{x}, \un{Y}=\un{y} ) \\
 & = \sum_{ \un{x}, \un{y}} P( \un{X} = \un{x}, \un{Y} = \un{y}) \sum_{j=1}^{R(\un{y})} H(\un{T}(j)| \un{I}(1) \backslash \un{I}(j), \ \un{X}=\un{x}, \un{Y}=\un{y})
\end{split}
\label{eq:T_YX_split}
\ee}}
The last inequality holds because $\un{T}(1),\ldots, \un{T}(j-1)$ can be determined from $\un{I} \backslash \un{I}(j)$ and extra conditioning cannot increase the entropy. As in Appendix \ref{app:lb1_liminf_proof}, we only consider terms in \eqref{eq:T_YX_split} for which
$(\un{Y}, \un{X}, \un{I} \backslash \un{I}(j))$ belongs to $\mc{E}_{j, k_1, k_2}$ for some $k_1 \geq 1$ and $1 \leq k_2 \leq k_1$.
(Please refer to \eqref{eq:Ejk1k2} for the definition of $\mc{E}_{j, k_1, k_2}$  and related notation.)
We lower bound the right side of \eqref{eq:T_YX_split} as follows.
\be
\begin{split}
& H(T^{M_n} \mid Y^{M_n}, X^n)  \\
& \geq  \sum_{k_1 \geq 1} \sum_{k_2 =1}^{k_1}   \sum_{\un{x}, \un{y}} \sum_{j=1}^{R(\un{y})} \sum_{\stackrel{\un{i} : (\un{x}, \un{y}, \un{i})  \in}{\mc{E}_{j, k_1, k_2}} }
 P( (\un{X}, \un{Y},  \un{I} \backslash \un{I}(j)) = (\un{x}, \un{y}, \un{i})  ) \\
& \qquad \cdot H(\un{T}(j) \mid (\un{X}, \un{Y},  \un{I} \backslash \un{I}(j)) = (\un{x}, \un{y}, \un{i}) ).
\end{split}
\label{eq:HTYXlb2}
\ee
Given $(\un{Y}, \un{X}, \un{I} \backslash \un{I}(j)) \in \mc{E}_{j, k_1, k_2 }$, the only uncertainty in $\un{T}$ is in the first bit of $\un{T}(j)$, which is denoted $T_{\tilde{a}_j+1}$ (using the notation introduced in \eqref{eq:Ejk1k2}). The different possibilities for $T_{\tilde{a}_j+1}$ are the following.
\begin{enumerate}
\item $T_{\tilde{a}_j+1}=0$ if $X_{a_j} = \bar{c}$ undergoes a complementary insertion leading to $Y_{\tilde{a}_j+1}=c$. In this case, $(k_2-1)$ out of the $k_1 \  c$'s in $\un{X}(j)$ undergo duplications and  the remaining $c$'s are transmitted without any insertions. There are $\tbinom{k_1}{k_2 -1}$ such insertion patterns, each of which occurs with probability
$(i \bar{\alpha}) (i \alpha)^{k_2-1} (1-i)^{k_1 - k_2 +1}$.

\item  $T_{\tilde{a}_j+1}=0$ if $X_{a_j} = \bar{c}$ is transmitted without any insertions. In this case, $k_2$ out of the $k_1 \ c$'s in the $\un{X}(j)$ undergo duplications, the remaining are transmitted without insertions. There are $\tbinom{k_1}{k_2}$ such insertion patterns, each of which occurs with probability
    $(1-i)(i \alpha)^{k_2} (1-i)^{k_1 - k_2}$.
\end{enumerate}
$H(\un{T}(j) \mid (\un{X}, \un{Y},  \un{I} \backslash \un{I}(j)) = (\un{x}, \un{y}, \un{i}) )$ is  the binary entropy associated with the two possibilities above and is given by
\be
\begin{split}
& H(\un{T}(j) \mid (\un{X}, \un{Y},  \un{I} \backslash \un{I}(j)) = (\un{x}, \un{y}, \un{i}) )\\
& = h \left( \frac{\tbinom{k_1}{k_2}  \bar{i}^{k_1 - k_2+1} (i \alpha)^{k_2} }
{ \tbinom{k_1}{k_2}  \bar{i}^{k_1 - k_2+1}(i \alpha)^{k_2}
+  \tbinom{k_1}{k_2 -1}  \bar{i}^{k_1 - k_2 +1} (i \bar{\alpha}) (i \alpha)^{k_2-1}}\right) \\
&= h \left( \frac{\bar{\alpha}k_2}{ \bar{\alpha}k_2 + \alpha(k_1-k_2+1)} \right).
\end{split}
\label{eq:HT_binent}
\ee
Substituting \eqref{eq:HT_binent} in \eqref{eq:HTYXlb2}, we obtain
\be
\begin{split}
& H(T^{M_n} \mid Y^{M_n}, X^n) \\
& \geq \sum_{k_1 \geq 1} \sum_{k_2 =1}^{k_1}    h \left( \frac{\bar{\alpha}k_2}{ \bar{\alpha}k_2 + \alpha(k_1-k_2+1)} \right) \\
& \quad \cdot \Big[ \sum_{\un{x}, \un{y}} \sum_{j=1}^{R(\un{y})} \sum_{\stackrel{\un{i} : (\un{x}, \un{y}, \un{i})  \in}{\mc{E}_{j, k_1, k_2}} }
 P( (\un{X}, \un{Y},  \un{I} \backslash \un{I}(j)) = (\un{x}, \un{y}, \un{i}) ) \Big].
\end{split}
\label{eq:HTY_almost}
\ee
The term in square brackets above was computed in \eqref{eq:RY_RX_lb1}. Substituting it in \eqref{eq:HTY_almost}  and dividing throughout by $n$ completes the proof.
\qedfilled

\section{Deletion Channel}
\subsection{Proof of Lemma \ref{lem:del_s_y}} \label{proof:del_s_y}
We first show that
\be
\lim_{n \to \infty} -\frac{1}{n} \log P(S^{M_n} | Y^{M_n}) =  \bar{d}  H_P(S_2 | Y_1 Y_2) \quad a.s.
\label{eq:lim_s1y1y2}
\ee
From Propositions \ref{prop:y} and \ref{prop:sy}, $\{Y_m\}_{m \geq 1}$ and $\{(S_m, Y_m)\}_{m \geq 1}$ are
both ergodic Markov chains with stationary transition probabilities. Therefore, from the Shannon-McMillan-Breiman theorem \cite{AlgoetCover88}
we have
\begin{align}
\lim_{m \to \infty} -\frac{1}{m} \log P(Y^m) & = H_P(Y_2|Y_1) \: a.s., \label{eq:hy2y1} \\
\lim_{m \to \infty} -\frac{1}{m} \log P(S^m, Y^m) & = H_P(S_2, Y_2|Y_1)\: a.s.\label{eq:s2y2y1}
\end{align}
 Subtracting \eqref{eq:s2y2y1} from \eqref{eq:hy2y1}, we get
\be \label{eq:hs2_y2y1}
\lim_{m \to \infty} -\frac{1}{m} \log P(S^m|Y^m)  = H_P(S_2|Y_2 Y_1) \: a.s.
\ee
Due to Lemma \ref{lem:entropyrate}, \eqref{eq:hs2_y2y1} and the fact that $\tfrac{M_n}{n} \to \bar{d}$ almost surely imply that \eqref{eq:lim_s1y1y2} holds.

The support of the random variable $(S^{M_n}|Y^{M_n})$ can be upper bounded  by representing
$S^{M_n}$ as
\be \underbrace{xx\ldots x}_{S_1} \ Y \ \underbrace{x\ldots x}_{S_2}\ldots  \ \ldots Y \ \underbrace{x\ldots x}_{S_{M_n}} \ Y
\underbrace{x\ldots x}_{S_{M_n+1}}  \label{eq:SMn_count}\ee
where the $Y$'s represent the bits of the sequence $Y^{M_n}$, and each $x$ represents a missing run. Since the maximum length of the binary sequence  in \eqref{eq:SMn_count} is $n$, we have Supp$(S^{M_n}|Y^{M_n}) \leq 2^n$. Hence, from Lemma \ref{lem:support}, $-\frac{1}{n} \log \text{Pr}(S^{M_n}|Y^{M_n})$ is uniformly integrable. Using this together with \eqref{eq:lim_s1y1y2} in Lemma \ref{lem:exchange_lim}, we conclude that
\be
\begin{split}
& \lim_{n \to \infty}\frac{1}{n} H_P({S}^{M_n}|{Y}^{M_n}) =  \lim_{n \to \infty} \expec\left[  -\frac{1}{n} \log P({S}^{M_n}|{Y}^{M_n}) \right]\\
& =\expec\left[ \lim_{n \to \infty} -\frac{1}{n} \log P({S}^{M_n}|{Y}^{M_n}) \right] = \bar{d} H_P(S_2|Y_1 Y_2).
\end{split}
\ee
We thus have
\ben
\begin{split}
& \lim_{n \to \infty}\frac{1}{n} H_P({S}^{M_n+1}|{Y}^{M_n})  \\
& = \lim_{n \to \infty}\frac{1}{n} H_P({S}^{M_n}|{Y}^{M_n})   +  \lim_{n\to \infty} \frac{1}{n} H_P({S}_{M_n+1}|{Y}^{M_n}, S^{M_n})\\
&=  \ \bar{d} H_P(S_2|Y_1 Y_2) + 0.
\end{split}
\een
\subsection{Proof of Lemma \ref{lem:del_x_sy}} \label{proof:del_x_sy}
Due to \eqref{eq:bits_to_runs}, it is enough to show that $\frac{1}{n} H_P(L^X_1, \ldots, L^X_{R_n}|L^{Y'}_1, \ldots, L^{Y'}_{R_n})$ converges to $\bar{\gamma} H_P(L^X_1|L^{Y'}_1)$. Since $\{(L^X_1, L^{Y'}_1), (L^X_2, L^{Y'}_2), \ldots \}$ is an i.i.d process, from the strong law of large numbers, we have
\be
\begin{split}
& \lim_{m \to \infty} -\frac{1}{m}  \text{Pr}(L^X_1, \ldots, L^X_{m} | L^{Y'}_1, \ldots, L^{Y'}_{m}) \\
& = H_P(L^X_1|L^{Y'}_1) \quad a.s.
\end{split}
\ee
Further, we have the normalized number of input runs $\frac{R_n}{n} \to \bar{\gamma}$ almost surely. Using the above in Lemma \ref{lem:entropyrate},
we obtain
\be
\begin{split}
& \lim_{n \to \infty} -\frac{1}{n}  \text{Pr}(L^X_1, \ldots, L^X_{R_n} | L^{Y'}_1, \ldots, L^{Y'}_{R_n}) \\
& = \bar{\gamma} H_P(L^X_1|L^{Y'}_1) \quad a.s.
\label{eq:run_conv}
\end{split}
\ee
Next, Supp$(L^X_1, \ldots, L^X_{R_n} | L^{Y'}_1, \ldots, L^{Y'}_{R_n})$ can be upper bounded  by  $2^n$ since since the random sequence
$(L^X_1, \ldots, L^X_{R_n})$ is equivalent to $X^n$, which can take on at most $2^n$ values. Hence, from Lemma \ref{lem:support}, $-\frac{1}{n} \log \text{Pr}(L^X_1, \ldots, L^X_{R_n} | L^{Y'}_1, \ldots, L^{Y'}_{R_n})$ is uniformly integrable.
Using this together with \eqref{eq:run_conv} in Lemma \ref{lem:exchange_lim}, we conclude that
\be
\begin{split}
& \lim_{n \to \infty}\frac{1}{n} H_P(L^X_1, \ldots, L^X_{R_n} | L^{Y'}_1, \ldots, L^{Y'}_{R_n})\\
&= \lim_{n \to \infty} \expec\left[  -\frac{1}{n} \log \text{Pr}(L^X_1, \ldots, L^X_{R_n} | L^{Y'}_1, \ldots, L^{Y'}_{R_n}) \right]\\
&=\expec\left[ \lim_{n \to \infty} -\frac{1}{n} \log \text{Pr}(L^X_1, \ldots, L^X_{R_n} | L^{Y'}_1, \ldots, L^{Y'}_{R_n}) \right]\\
&= \bar{\gamma} H_P(L^X_1|L^{Y'}_1).
\end{split}
\ee
\qedfilled
\subsection{Proof of Lemma \ref{lem:lb_delpenalty} } \label{app:lb_delpenalty}
We expand $H(\un{S} | \un{Y}, \un{X})$ in terms of the runs of $\un{Y}$. Following the notation used in Appendix \ref{app:lb1_liminf_proof}, we denote the number of runs in $\un{Y}$ by $R(\un{Y})$, the $j$th run of $\un{Y}$ by $\un{Y}(j)$ and the corresponding part of $\un{S}$ by
$\un{S}(j)$. We have
{\small{
\be
\begin{split}
& H(S^{M_n +1} \mid Y^{M_n}, X^n) \\
& = \sum_{ \un{x}, \un{y}} P(\un{X} = \un{x},  \un{Y} = \un{y}) \\
&\quad \cdot \sum_{j=1}^{R(\un{y})} H(\un{S}(j) \mid \un{S}(1), \ldots, \un{S}(j-1), \un{X}=\un{x}, \un{Y}=\un{y}) \\
& \geq  \sum_{ \un{x}, \un{y}} P(\un{X} = \un{x},  \un{Y} = \un{y}) \sum_{j=1}^{R(\un{y})} H(\un{S}(j) | \un{X}=\un{x}, \un{Y}=\un{y},
\Theta(\un{Y} \backslash \un{Y}(j))
\end{split}
\label{eq:HS_YX_expand}
\ee}}
\begin{figure*}[b]
\vspace*{0pt} \hrulefill
\normalsize
\be
\begin{split}
\mc{F}_{j, z,r,s}  = & \left\{ (\un{X}, \un{Y}, \Theta(\un{Y} \backslash \un{Y}(j)) ) :  ( \underbrace{ {c}, {c},, \ldots, {c}}_{z \text{ bits}}, \underbrace{\bar{c}, \bar{c}, \ldots,\bar{c}}_{k \text{ bits}},
\underbrace{ {c},  {c}, \ldots, {c}}_{r \text{ bits}})  \longrightarrow
\un{Y}(j) = ( \underbrace{c, c,\ldots,c}_{s  \text{ bits}}) \ \text{ for some } k \geq 1, \ c\in\{0,1\}  \right\}
\end{split}
\label{eq:Fj_zrs}
\ee
\be
\begin{split}
& \sum_{\un{x}, \un{y}} \sum_{j=1}^{R(\un{y})}  \sum_{ \stackrel{ \un{\theta}: (\un{x}, \un{y}, \un{\theta}) \in}{\mc{F}_{j, z,r,s}} }
\hspace{-8pt} P (  \un{X}, \un{Y},   \Theta(\un{Y} \backslash \un{Y}(j))  = \un{x}, \un{y}, \un{\theta} ) \\
& = \sum_{\un{y}} P( \un{Y} = \un{y}) \sum_{j=1}^{R(\un{y})}  \mathbf{1}( j\text{th $\un{y}$-run has length } s)   \
P( \text{runs of $\un{X}$ with lengths } (z,k,r) \  \rightarrow \  j\text{th $\un{y}$-run, for some } k\geq 1) \\
& = \expec \Big[ \expec \Big[ \sum_{j=1}^{R(\un{Y})}
\mathbf{1}( \text{runs of $\un{X}$ with lengths } (z,k,r) \  \rightarrow \  j\text{th $\un{Y}$-run, for some } k\geq 1)  \  \mathbf{1}( j\text{th $\un{Y}$-run has length } s)
 \mid R(\un{Y})\Big] \Big]  \\
  & \stackrel{(a)}{=} \expec \Big[ \sum_{j=1}^{R(\un{Y})} \gamma^{z-1}\bar{\gamma} \underbrace{\left(\sum_{k=1}^{\infty} \gamma^{k-1}\bar{\gamma} d^k\right)}_{\text{ middle run deleted}}
\gamma^{r-1}\bar{\gamma} \cdot  \binom{z+r}{s} (1-d)^s d^{z+r-s}  \Big]  \
\stackrel{(b)}{=} n \bar{d} \bar{q} \cdot \frac{ \bar{\gamma}^3   d}{\gamma^2   (1 - \gamma d)}  (\gamma d)^{z+r}   \binom{z+r}{s} \left( \frac{1-d}{d} \right)^s.
\end{split}
\label{eq:HS_YX_part2}
\ee
\end{figure*}
where $\un{Y} \backslash \un{Y}(j)$ is the sequence obtained by removing $\un{Y}(j)$ from $\un{Y}$.  $\Theta(\un{Y} \backslash \un{Y}(j) )$  denotes the exact deletion pattern corresponding to the output bits $\un{Y} \backslash \un{Y}(j)$, i.e.,
it  tells us which bit in $\un{X}$ corresponds to each bit in $\un{Y} \backslash \un{Y}(j)$. The inequality  in \eqref{eq:HS_YX_expand} holds since
$\un{S}(1),\ldots, \un{S}(j-1)$ is determined by $\Theta(\un{Y} \backslash \un{Y}(j))$.

Motivated by the discussion in Section \ref{subsec:del_penalty}, we obtain an analytical lower bound for the right side of \eqref{eq:HS_YX_expand}  by considering only those terms for which the run $\un{Y}(j)$ is generated from either one or three adjacent runs in $\un{X}$, as in \eqref{eq:del_exXY}.
For $z,r \geq 1$ and  $1 \leq s \leq z + r$, define the set $\mc{F}_{j, z,r,s}$ as in \eqref{eq:Fj_zrs} at the bottom of this page.
 We allow the possibility that all of $\un{Y}(j)$ is generated from just one of the three runs; further note that the $\bar{c}$-run in the middle is always deleted. The right side of \eqref{eq:HS_YX_expand} is lower bounded by considering only triples $(\un{X}, \un{Y}, \Theta(\un{Y} \backslash \un{Y}(j))  \in \mc{F}_{j, z,r,s}$, as follows.
 \be
\begin{split}
& H(S^{M_n +1} \mid Y^{M_n}, X^n) \\
& \geq \sum_{z,r \geq 1} \sum_{s=1}^{z+r} \sum_{\un{x}, \un{y}} \sum_{j=1}^{R(\un{y})}
\sum_{ \stackrel{ \un{\theta}: (\un{x}, \un{y}, \un{\theta}) \in}{\mc{F}_{j, z,r,s}} }
\hspace{-8pt} P (  \un{X}, \un{Y},   \Theta(\un{Y} \backslash \un{Y}(j))  = \un{x}, \un{y}, \un{\theta} ) \\
&\quad \cdot H(\un{S}(j) \mid  \un{X},  \un{Y},   \Theta(\un{Y} \backslash \un{Y}(j))  = \un{x}, \un{y}, \un{\theta}).
\end{split}
\label{eq:HS_YX_expand2}
\ee

 $H(\un{S}(j) \mid \un{Y}, \un{X},  \Theta(\un{Y} \backslash \un{Y}(j)) = \un{x}, \un{y}, \un{\theta} )$ can be computed as follows for  $(\un{x}, \un{y}, \un{\theta}) \in \mc{F}_{j, z,r,s}$.  Given $\Theta(\un{Y} \backslash \un{Y}(j))$, we exactly know the set of  adjacent runs in $\un{X}$ from gave rise to $\un{Y}(j)$.  Given $(\un{X}, \un{Y}(j), \Theta(\un{Y} \backslash \un{Y}(j))) \in \mc{F}_{j, z,r,s }$,  $\ \un{Y}(j)$ arises from  one or three adjacent runs of $\un{X}$. It is possible that additional runs may be deleted in $\un{X}$ on either side of the three adjacent runs shown in \eqref{eq:Fj_zrs}. To handle this case, we can assume that $\Theta(\un{Y} \backslash \un{Y}(j))$  gives enough information so that we know three adjacent input runs that correspond to   $\un{Y}(j)$. (Conditioning on additional random variables can only decrease the lower bound.) Then the length-$s$ vector $\un{S}(j)$ has at most one non-zero element:  For $l=0,\ldots,s-1$, if $\un{Y}(j)$ was formed with $l$ bits from the first length-$z$ run and $s-l$ bits from the third length-$r$ run, $\un{S}(j)$ will have a non-zero in position $l+1$. If $\un{Y}(j)$ was formed with all $s$ bits from the first length-$z$ run, then all the $s$ elements of $\un{S}(j)$ are zero and the symbol in $\un{S}$ immediately after $\un{S}(j)$  is non-zero. We thus have
\be
\begin{split}
& H(\un{S}(j) \mid \un{X},  \un{Y},   \Theta(\un{Y} \backslash \un{Y}(j))  = \un{x}, \un{y}, \un{\theta} ) \\
& = H\left( \left\{ \frac{ \tbinom{z}{l} \tbinom{r}{s-l} }{\tbinom{z+r}{s}} \right \}_{l=0}^{s} \right).
\end{split}
\label{eq:HS_YX_part1}
\ee
Next, for a fixed $(z,r,s)$, we compute the three innermost sums of $P ( \un{X}, \un{Y},   \Theta(\un{Y} \backslash \un{Y}(j))  = \un{x}, \un{y}, \un{\theta})$ in \eqref{eq:HS_YX_expand2}. This is done in \eqref{eq:HS_YX_part2} at the bottom of this page where $(a)$ is obtained as follows. In the third line of  \eqref{eq:HS_YX_part2}, each term of the inner expectation  is the probability of three successive $\un{X}$-runs having the specified lengths and giving rise to a  $\un{Y}$-run of length $s$. We use the fact that  $\un{X}$  is first-order Markov and thus  has independent runs.\footnote{Conditioned on $R(\un{Y})$, the run-lengths of $\un{X}$ are not strictly independent. This can be handled by observing that $R(\un{Y})/n$ concentrates around its mean $\bar{q}$ and taking the  $z,r$ sums only over values much smaller than  $n$, e.g., $1 \leq z,r \leq n^{0.5 + \epsilon}$. We can then show that conditioned on $R(\un{Y})$, the probability that a run of $\un{X}$ has length $z$ is very close to $\gamma^{z-1} \bar{\gamma}$.}  $(b)$ holds because $\un{Y}$ is  first-order Markov with  parameter $q$ and has expected length $n(1-d)$. The expected number of runs in
$\un{Y}$ equals $(1-q)$ times the expected length of $\un{Y}$.

Substituting \eqref{eq:HS_YX_part1} and \eqref{eq:HS_YX_part2} in \eqref{eq:HS_YX_expand2}, we obtain
\be
\begin{split}
&H(S^{M_n +1} \mid Y^{M_n}, X^n) \\
 & \geq n \sum_{z,r \geq 1} \sum_{s=1}^{z+r} \frac{ \bar{d} \ \bar{q} \ \bar{\gamma}^3  \  d  }{\gamma^2   (1 - \gamma d)}  (\gamma d)^{z+r} \
 \binom{z+r}{s} \left( \frac{1-d}{d} \right)^s \\
 &\quad \cdot H\Big( \Big\{ \frac{ \tbinom{z}{l} \tbinom{r}{s-l} }{\tbinom{z+r}{s}} \Big \}_{l=0}^{s} \Big).
\end{split}
\label{eq:HS_YX_expand3}
\ee
Dividing both sides by $n$ yields the result of the lemma.
\qedfilled

\section{InDel Channel}
\subsection{Proof of Lemma \ref{lem:delins_4}} \label{proof:delins_lem4}
We first note that
\be
\begin{split}
& \limsup_{n \to \infty}  \frac{1}{n} H_P(S^{M_n +1}| T^{M_n}, Y^{M_n}) \\
& = (1-d+i) \limsup_{m \to \infty} \frac{1}{m} H_P(S^{m}| T^{m}, Y^{m}).
\label{eq:indel_first_step}
\end{split}
\ee
The proof of \eqref{eq:indel_first_step} is along the same lines as that of Lemma \ref{lem:ins_H_I_Y}: we use the uniform integrability of the sequence
$\left\{ -\frac{1}{n} \log P(S^{M_n +1}| T^{M_n}, Y^{M_n}) \right\}$ along with the fact that $\tfrac{M_n}{n} \to (1-d+i)$. The uniform integrability follows from Lemma \ref{lem:support} since Supp$(S^{M_n+1}|T^{M_n}, Y^{M_n})$ is upper bounded by $2^n$ as explained in Section \ref{proof:del_s_y}.
We then have
\be
\begin{split}
& \frac{1}{m} H(S^{m}| T^{m}, Y^{m}) = \frac{1}{m} \sum_{j=1}^{m}  H(S_j|S^{j-1}, T^m, Y^m) \\
&  \leq \frac{1}{m} \sum_{j=1}^{m} H(S_j|Y_{j-1}, Y_j,  T_j).
\label{eq:lim_SYYT}
\end{split}
\ee
We will show that $\lim_{j \to \infty} H(S_j|Y_{j-1}, Y_j,  T_j)$ exists and obtain an analytical expression for it.
For all $j$,
{\small{
\be \label{eq:delins_sjyyt_split}
\begin{split}
& H(S_j|Y_{j-1}, Y_j,  T_j) \\
&= P(Y_{j-1}, Y_j, T_j=0) H(S_j|Y_{j-1}, Y_j,  T_j=0) = \\
& \hspace{-3pt} \sum_{y \in \{0,1\}} \hspace{-5pt} P(Y_{j-1}=Y_j=y, T_j=0) H(S_j|Y_{j-1}=Y_j=y,  T_j=0) \\
&\quad + P( (Y_{j-1}, Y_j, T_j)= (\bar{y}, y, 0)) H(S_j | (Y_{j-1}, Y_j,  T_j)= (\bar{y},y,0))
\end{split}
\ee}}
The first equality above holds since $T_j=1$ implies $Y_j$ is an inserted bit, and so no deleted runs occur between $Y_{j-1}$ and $Y_j$.
$P(Y_{j-1}, Y_j, T_j=0)$ can be computed as follows.
\be
\begin{split}
& P( Y_{j-1}={y}, Y_j=y, T_j=0) \\
& = P( (I_{j-1},  Y_{j-1}, T_j, Y_j)=(0,y,0,y) ) \\
& \quad + P((I_{j-1}, T_{j-1},  Y_{j-1}, T_j, Y_j)=(1, 0,y,0,y)) \\
& \quad + P((I_{j-1}, T_{j-1},  Y_{j-1}, T_j, Y_j)=(1, 1, y,0,y)) \\
&\stackrel{(a)}{=} \frac{1}{2} P(I_{j-1}=0) P(I_j=0|I_{j-1}=0) q \\
& \quad + \frac{1}{2} P(I_{j-1}=0)  P(I_j=1|I_{j-1}=0) \alpha \\
&\quad + \frac{1}{2} P(I_{j-1}=1) \alpha  q + \frac{1}{2} P(I_{j-1}=1) \bar{\alpha} \bar{q} \\
&\stackrel{j \to \infty}{\longrightarrow} \  \pi (Y_{j-1}={y}, Y_j=y, T_j=0)     \\
& \quad  = \frac{1}{2}\Big[\frac{1}{1+i'}((1-i')q + {i'}\alpha) +  \frac{i'}{1+i'} (\alpha q +  \bar{\alpha}\bar{q})  \Big].
\end{split}
\label{eq:delins_yyt}
\ee
The last two terms in (a) are obtained by noting that $I_{j-1}=1$ implies $Y_{j-1}$ is an insertion and hence $T_j=0$. In this case, $Y_{j-2}$ is the last non-inserted bit before $Y_j$, and the last two terms in (a) correspond to $Y_{j-1}$ being a duplication and a complementary insertion, respectively. The convergence in last line is due to the fact that $\{I_j\}_{j \geq 1}$ is a Markov chain with
stationary distribution
\be \pi (I_j=1)=\frac{i'}{1+i'},  \quad \pi(I_j=0)=\frac{1}{1+i'}.  \label{eq:ipr_stat} \ee

Similarly, as $j \to \infty$, $P( Y_{j-1}={y}, Y_j=\bar{y}, T_j=0)$ converges to
\be
\begin{split}
& \pi( Y_{j-1}={y}, Y_j=\bar{y}, T_j=0)  \\
&= \frac{1}{2}\left[\frac{1}{1+i'}(1-i')(1-q) + \frac{i'}{1+i'} \alpha(1-q) +  \frac{i'}{1+i'} \bar{\alpha}q \right].
\end{split}
\label{eq:delins_ybaryt}
\ee
The joint distributions $\pi(S_j, Y_{j-1}={y},  Y_j=y, T_j=0)$ and $\pi(S_j, Y_{j-1}=\bar{y},  Y_j=y, T_j=0)$ are next determined in order
to  compute the entropies in \eqref{eq:delins_sjyyt_split}. For $k=0,1,\ldots$, we have
\be
\begin{split}
& P((S_j, Y_{j-1},  Y_j, T_j)=(k,y,y,0)) \\
& = P( (I_{j-1}, Y_{j-1}, T_j, Y_j, S_j) = (0,y,0,y, k)) \\
&   + P( (I_{j-1}, T_{j-1}, Y_{j-1}, T_j, Y_j, S_j) = (1,0, y,0,y, k) )\\
&  + P( (I_{j-1}, T_{j-1}, Y_{j-1}, T_j, Y_j, S_j) = (1,1, y,0,y, k) ).
\label{eq:iytys}
\end{split}
\ee
The first term in \eqref{eq:iytys} corresponds to $Y_{j-1}$ being an original input bit, the second term to $Y_{j-1}$ being a duplication, and the third to
$Y_{j-1}$ being a complementary insertion. Each of these terms can be calculated in a manner very similar to
 \eqref{eq:del_pys_y} and \eqref{eq:del_pbarys_y} in Proposition \ref{prop:sy}. Combined with the convergence of $P(I_j)$ to the  stationary distribution
 \eqref{eq:ipr_stat}, we obtain that $P((S_j, Y_{j-1},  Y_j, T_j)=(k,y,y,0))$ converges to the distribution
\be
\begin{split}
& \pi(S_j=k, Y_{j-1}={y},  Y_j=y, T_j=0)=\\
& \left\{
\begin{array}{l}
\frac{1}{2(1+i')} \left[ i'\alpha + (1-i'\bar{\alpha}) \frac{\gamma (1-d)}{1-\gamma d} + i' \bar{\alpha} \frac{(1-\gamma)(1-d)}{(1-\gamma d)^2} \right], \ k=0\\
\frac{1-i'\bar{\alpha}}{2(1+i')} \frac{(1-\gamma)(1-d)}{(1-\gamma d)^2} \left(\frac{d(1-\gamma)}{1-\gamma d}\right)^k, \ k=1,3, \ldots\\
\frac{i'\bar{\alpha}}{2(1+i')} \frac{(1-\gamma)(1-d)}{(1-\gamma d)^2} \left(\frac{d(1-\gamma)}{1-\gamma d}\right)^k, \ k=2,4,\ldots
\end{array}
\right.
\end{split}
\label{eq:delins_sjyyt}
\ee
Similarly, we can also show that $P(S_j=k, Y_{j-1}=\bar{y},  Y_j=y, T_j=0)$ converges to the distribution
\be
\begin{split}
& \pi(S_j=k, Y_{j-1}=\bar{y},  Y_j=y, T_j=0)= \\
& \left\{
\begin{array}{l}
\frac{1}{2(1+i')} \left[ (1-i'\bar{\alpha})\frac{(1-\gamma)(1-d)}{(1-\gamma d)^2} + i' \bar{\alpha} \frac{\gamma (1-d)}{1-\gamma d}  \right], \  k=0\\
\frac{i'\bar{\alpha}}{2(1+i')} \frac{(1-\gamma)(1-d)}{(1-\gamma d)^2} \left(\frac{d(1-\gamma)}{1-\gamma d}\right)^k, \  k=1,3, \ldots\\
\frac{1- i'\bar{\alpha}}{2(1+i')} \frac{(1-\gamma)(1-d)}{(1-\gamma d)^2} \left(\frac{d(1-\gamma)}{1-\gamma d}\right)^k, \ k=2,4,\ldots
\end{array}
\right.
\end{split}
\label{eq:delins_sjbaryyt}
\ee
As a sanity check, it can be verified that when the joint distributions given by \eqref{eq:delins_sjyyt} and \eqref{eq:delins_sjbaryyt} are summed over $k$, they yield the distributions specified in \eqref{eq:delins_yyt} and \eqref{eq:delins_ybaryt}, respectively.

The expression in \eqref{eq:delins_sjyyt_split} can  nowbe computed by substituting from \eqref{eq:delins_yyt} and \eqref{eq:delins_ybaryt} for the
$P(Y_{j-1}, Y_j, T_j)$ terms and calculating  the entropy terms using the joint distributions in \eqref{eq:delins_sjyyt} and \eqref{eq:delins_sjbaryyt}. (The calculation is elementary but somewhat tedious, and hence omitted.) This yields the limiting value of $H_P(S_j|Y_{j-1}, Y_j, T_j)$. The lemma then follows from \eqref{eq:indel_first_step} and \eqref{eq:lim_SYYT}.
\qedfilled

\IEEEtriggeratref{10}
\section*{Acknowledgements}
We thank the associate editor and the anonymous reviewers for their comments which helped us strengthen our initial results and improve the paper.

\end{document}